\documentclass{article}
\usepackage{caption}
\usepackage{longtable}
\usepackage{cellspace}
\usepackage{booktabs}
\usepackage{arydshln}
\usepackage{amsfonts}
\usepackage{amssymb}
\usepackage{graphicx}
\usepackage{amsmath}
\usepackage{amsthm}
\usepackage{amsmath}
\usepackage{amsfonts}
\usepackage{amsthm}
\usepackage{mathabx}
\usepackage[colorlinks=true,linkcolor=blue,citecolor= red,bookmarksopen=true]{hyperref}
\nonstopmode
\newtheorem{theorem}{Theorem}[section]
\newtheorem{proposition}[theorem]{Proposition}
\newtheorem{lemma}[theorem]{Lemma}

\newtheorem{corollary}[theorem]{Corollary}

\newtheorem{assumption}[theorem]{Assumption}
\newtheorem{standassumption}[theorem]{Standing Assumption}
\newtheorem{setup}[theorem]{Setup}
\newtheorem{definition}[theorem]{Definition}
\newtheorem{notation}[theorem]{Notation}

\theoremstyle{remark}
\newtheorem{example}[theorem]{Example}
\newtheorem{remark}[theorem]{Remark}
\newcommand{\probp}{P}
\newcommand{\probq}{Q}
\newcommand{\R}{{\mathbb R}}

\newcommand\abs[1]{\left|#1\right|}

\newcommand\Eq[1]{E_{Q} \left[#1\right]}

\begin{document}

\title{Entropy Martingale Optimal Transport and \\
Nonlinear Pricing-Hedging Duality}
\date{\today}
\author{Alessandro Doldi\thanks{%
Dipartimento di Matematica, Universit\`{a} degli Studi di Milano, Via
Saldini 50, 20133 Milano, Italy, $\,\,$\emph{alessandro.doldi@unimi.it}. }
\and Marco Frittelli\thanks{%
Dipartimento di Matematica, Universit\`a degli Studi di Milano, Via Saldini
50, 20133 Milano, Italy, \emph{marco.frittelli@unimi.it}. }}
\maketitle

\begin{abstract}
\noindent The objective of this paper is to develop a duality between a
novel Entropy Martingale Optimal Transport problem $(A)$ and an associated
optimization problem $(B)$. In $(A$) we follow the approach taken in the
Entropy Optimal Transport (EOT) primal problem by Liero et al.
\textquotedblleft Optimal entropy-transport problems and a new
Hellinger-Kantorovic distance between positive measures\textquotedblright ,
Invent. math. 2018, but we add the constraint, typical of Martingale Optimal
Transport (MOT) theory, that the infimum of the cost functional is taken
over martingale probability measures, instead of finite positive measures,
as in Liero et al. The Problem $(A)$ differs from the corresponding problem
in Liero et al. not only by the martingale constraint, but also because we
admit less restrictive penalization terms $\mathcal{D}_{U}$, which may not
have a divergence formulation. In Problem (B) the objective functional,
associated via Fenchel conjugacy to the terms $\mathcal{D}_{U}$, is not any
more linear, as in OT or in MOT. This leads to a novel optimization problem
which also has a clear financial interpretation as a non linear subhedging
value. Our theory allows us to establish a nonlinear robust pricing-hedging
duality, which covers a wide range of known robust
results. We also focus on Wasserstein-induced penalizations and we study how
the duality is affected by variations in the penalty terms, with a special
focus on the convergence of EMOT to the extreme case of MOT.
\end{abstract}

\parindent=0em \noindent


\section{Introduction}
\label{secintrod}
\noindent In this research we exploit Optimal Transport (OT) theory to develop
the duality\footnote{%
The notations are explained in this Introduction and the duality is further
discussed in Section \ref{introEMOTdiscussion}.} 
\begin{equation}
A:=\inf_{Q\in \mathrm{Mart}(\Omega )}\left( E_{Q}\left[ c\right] +\mathcal{D}%
_{U}(Q)\right) =\sup_{\Delta \in \mathcal{H}}\sup_{\varphi \in \mathbf{\Phi }%
_{\Delta }(c)}S^{U}\left( \varphi \right) =:B.  \label{eq1}
\end{equation}%
In $(A)$ we recognize the approach taken in the Entropy Optimal Transport
primal problem (Liero et al. \cite{Sava}) with the additional constraints,
typical of Martingale Optimal Transport (MOT), that the infimum of the cost
functional $c$ is taken over \textit{martingale probability} measures,
instead of finite positive \textit{measures}, as in \cite{Sava}. This is a
consequence of the additional supremum over the integrands $\Delta \in 
\mathcal{H}$ in problem $(B)$, and of the cash additivity of the functional $%
S^{U}$. The functional $S^{U}$ is associated to a, typically non linear,
utility functional $U$ and represents the pricing rule over suitable
continuous functions $\varphi $ defined on $\Omega $. We observe that the
marginal constraints, typical of OT problems, in $(A)$
are relaxed by introducing the functional $\mathcal{D}_{U}$, also associated
to the map $U$, which may have a divergence formulation. The counterpart of
this in Problem $(B)$ is that the functional $S^{U}$, associated via Fenchel
conjugacy to the penalization functional $\mathcal{D}_{U}$ is not\textit{\
necessarily linear}, as in OT or in MOT. Both $S^{U}$ and $\mathcal{D}_{U}$
may also depend on the marginals of some martingale measure $\widehat{Q}$. The
duality (\ref{eq1}) generalizes the well known robust pricing hedging
duality in financial mathematics.

We provide a clear financial interpretation of both problems and observe
that the novel concept of a non linear subhedging value expressed by ($B$)
was not previously considered in the literature.

\subsection{Pricing-hedging duality in financial mathematics}

The notion of subhedging price is one of the most analyzed concepts in
financial mathematics. In this Introduction we will take the point of view
of the subhedging price, but obviously an analogous theory for the
superhedging price can be developed as well. We are considering a discrete time
market model with zero interest rate. It may be convenient for the reader to
have at hand the summary described in Table \ref{tab1} on page \pageref{tab1}%
.

\paragraph{The classical setup}

In the classical setup of stochastic securities market models, one considers
an adapted stochastic process $X=(X_{t})_{t}$, $t=0,...,T,$ defined on a
filtered probability space $(\Omega ,\mathcal{F},(\mathcal{F}_{t})_{t},P),$
representing the price of some underlying asset. Let $\mathcal{P}(P)$ be the
set of all probability measures on $\Omega $ that are absolutely continuous
with respect to $P$, $\mathrm{Mart}(\Omega )$ be the set of all probability
measures on $\Omega $ under which $X$ is a martingale and $\mathcal{M}(P)=$ $%
\mathcal{P}(P)\cap \mathrm{Mart}(\Omega ).$ We also let $\mathcal{H}$ be the
class of admissible integrands and $I^{\Delta }:=I^{\Delta }(X)$ be the
stochastic integral of $X$ with respect to $\Delta \in \mathcal{H}$. Under
reasonable assumptions on $\mathcal{H}$, the equality 
\begin{equation}
E_{Q}\left[ I^{\Delta }(X)\right] =0  \label{EE}
\end{equation}
holds for all $Q\in \mathcal{M}(P)$ and, as well known, all linear pricing
functionals compatible with no arbitrage are expectations $E_{Q}[\cdot ]$
under some probability $Q\in \mathcal{M}(P)$ such that $Q\sim P$.

We denote with $p$ the \textbf{subhedging price} of a contingent claim $c:%
\mathbb{R\rightarrow R}$ written on the payoff $X_{T}$ of the underlying
asset. Let $\mathcal{L}(P)\subseteq L^{0}((\Omega ,\mathcal{F}_{T},P))$ be
the space of random payoffs satisfying appropriate integrability or
boundedness assumptions and let $Z:=c(X_{T})\in \mathcal{L}(P)$. Then $p:%
\mathcal{L}(P)\rightarrow \mathbb{R}$ is defined by 
\begin{equation}
p(Z):=\sup \left\{ m\in \mathbb{R}\mid \exists \Delta \in \mathcal{H}\text{
s.t. }m+I^{\Delta }(X)\leq Z\text{, }P-\text{a.s.}\right\} .  \label{subrep}
\end{equation}%
%
%
%
%
%
%
%
%
%
%
%
%
%
%
%
The subhedging price is independent from the preferences of the agents, but
it depends on the reference probability measure via the class of $P$-null
events. It satisfies the following two key properties:

\begin{enumerate}
\item[(CA)] Cash Additivity on $\mathcal{L}(P)$: $p(Z+k)=p(Z)+k,$ for all $%
k\in \mathbb{R}$, $Z\in \mathcal{L}(P).$

\item[(IA)] Integral Additivity on $\mathcal{L}(P)$: $p(Z+I^{\Delta })=p(Z),$
for all $\Delta \in \mathcal{H}$, $Z\in \mathcal{L}(P).$
\end{enumerate}

When a functional $p$ satisfies (CA), then $Z, k$ and $p(Z)$ must be
expressed in the same monetary unit and this allows for the \textit{monetary}
interpretation of $p$, as the price of the contingent claim. This will be
one of the key features that we will require also in the novel definition of
the nonlinear subhedging value. The (IA) property and $p(0)=0$ imply that
the $p$ price of any stochastic integral $I^{\Delta }(X)$ is equal to zero,
as in (\ref{EE}).

Since the seminal works of El Karoui and Quenez \cite{ElKarouiQuenez95},
Karatzas \cite{Karatzas95}, Delbaen and Schachermayer \cite%
{DelbaenSchachermayer94}, it was discovered that, under the no arbitrage
assumption, the dual representation of the subhedging price $p$ is 
\begin{equation}
p(Z)=\inf_{Q\in \mathcal{M}(P)}E_{Q}\left[ Z\right] .  \label{dualp}
\end{equation}%
More or less in the same period, the concept of \textbf{coherent risk measure%
} was introduced in the pioneering work by Artzner et al. \cite{ADEH99}. A
Coherent Risk Measure $\rho :\mathcal{L}(P)\rightarrow \mathbb{R}$
determines the minimal capital required to make acceptable a financial
position and its dual formulation is assigned by%
\begin{equation}
-\rho (Y)=\inf_{Q\in \mathcal{Q}\subseteq \mathcal{P}(P)}E_{Q}\left[ Y\right]
,  \label{dualrho}
\end{equation}%
where $Y$ is a random variable representing future profit-and-loss and $%
\mathcal{Q}\subseteq \mathcal{P}(P)$. Coherent Risk Measures $\rho $ are
convex, cash additive, monotone and positively homogeneous. We take the
liberty to label both the representations in (\ref{dualp}) and in (\ref%
{dualrho}) as the \textquotedblleft \textit{sublinear case}%
\textquotedblright .

\bigskip

In the study of incomplete markets the concept of the (buyer) \textbf{%
indifference price} $p^{b}$, originally introduced by Hodges and Neuberger\ 
\cite{Hodges89}, received, in the early 2000, increasing consideration (see
Frittelli \cite{Frittelli00}, Rouge and El Karoui \cite{RougeElkaroui00},
Delbaen et al. \cite{fiveauthors02}, Bellini and Frittelli \cite%
{BelliniFrittelli02}) as a tool to assess, \textit{consistently with the no
arbitrage principle}, the value of non replicable contingent claims, and not
just to determine an upper bound (the superhedging price) or a lower bound
(the subhedging price) for the price of the claim. Differently from the
notion of subhedging, $p^{b}$ is based on some concave increasing utility
function $u:\mathbb{R}\rightarrow \lbrack -\infty ,+\infty )$ of the agent.
By defining the indirect utility function 
\begin{equation*}
U(w_{0}):=\sup_{\Delta \in \mathcal{H}}E_{P}[u(w_{0}+I^{\Delta }(X))],
\end{equation*}%
where $w_{0}\in \mathbb{R}$ is the initial wealth, the indifference price $%
p^{b}$ is defined as%
\begin{equation*}
p^{b}(Z):=\sup \left\{ m\in \mathbb{R}\mid U(Z-m)\geq U(0)\right\} .
\end{equation*}%
Under suitable assumptions, the dual formulation of $p^{b}$ is%
\begin{equation}
p^{b}(Z)=\inf_{Q\in \mathcal{M}(P)}\left\{ E_{Q}\left[ Z\right] +\alpha
_{u}(Q)\right\} ,  \label{dualbuyer}
\end{equation}%
and the penalty term $\alpha _{u}:\mathcal{M}(P)\rightarrow \lbrack
0,+\infty ]$ is associated to the particular utility function $u$ appearing
in the definition of $p^{b}$ via the Fenchel conjugate of $u$. We observe
that in case of the exponential utility function $u(x)=1-\exp (-x),$ the
penalty is $\alpha _{\exp }(Q):=H(Q,P)-\min_{Q\in \mathcal{M}(P)}H(Q,P),$
where%
\begin{equation*}
H(Q,P):=\int_{\Omega }F\left( \frac{\mathrm{d}Q}{\mathrm{d}P}\right) \,%
\mathrm{d}P\text{,\quad if }Q\ll P\text{ and }F(y)=y\ln (y),
\end{equation*}%
is the relative entropy. In this case, the penalty $\alpha _{\exp }$ is a
divergence functional, as those that will be considered below in Section \ref%
{SecDivergence}. Observe that the functional $p^{b}$ is concave, monotone
increasing and satisfies both properties (CA) and (IA), but it is not
necessarily linear on the space of all contingent claims. As recalled in the
conclusion of Frittelli \cite{Frittelli00}, \textquotedblleft there is no
reason why a price functional defined on the whole space of bundles and
consistent with no arbitrage should be linear also outside the space of
marketed bundles\textquotedblright .

It was exactly the particular form \eqref{dualbuyer} of the indifference
price that suggested to Frittelli and Rosazza Gianin \cite%
{FrittelliRosazza02} the introduction of the concept of \textbf{Convex Risk Measure}
(also independently introduced by Follmer and Schied \cite{FollmerSchied02a}%
), as a map $\rho :\mathcal{L}(P)\rightarrow \mathbb{R}$ that is convex,
cash additive and monotone decreasing. Under good continuity properties, the
Fenchel-Moreau Theorem shows that any convex risk measure admits the
following representation%
\begin{equation}
-\rho (Y)=\inf_{Q\in \mathcal{P}(P)}\left\{ E_{Q}\left[ Y\right] +\alpha
(Q)\right\}  \label{dualconvex}
\end{equation}%
for some penalty $\alpha :\mathcal{P}(P)\rightarrow \lbrack 0,+\infty ]$. We
will then label functional in the form (\ref{dualbuyer}) or (\ref{dualconvex}%
) as the \textquotedblleft \textit{convex case}\textquotedblright . As a
consequence of the cash additivity property, in the dual representations (%
\ref{dualbuyer}) or (\ref{dualconvex}) the infimum is taken with respect to 
\textit{probability measures}, namely with respect to normalized non
negative elements in the dual space, which in this case can be taken as $%
L^{1}(P)$. Differently from the indifference price $p^{b}$, convex risk
measures do not necessarily take into account the presence of the stochastic
security market, as reflected by the absence of any reference to martingale
measures in the dual formulation (\ref{dualconvex}) and (\ref{dualrho}), in
contrast to (\ref{dualbuyer}) and (\ref{dualp}).

\paragraph{Pathwise finance}

As a consequence of the financial crisis in 2008, the uncertainty in the
selection of a reference probability $P$ gained increasing attention and led
to the investigation of the notions of arbitrage and of the pricing hedging
duality in different settings. On the one hand, the single reference
probability $P$ was replaced with a family of - a priori non dominated -
probability measures, leading to the theory of Quasi-Sure Stochastic
Analysis. 
On the other hand, taking an even more radical approach, a probability free,
pathwise, theory of financial markets was developed, as in Acciaio et al. 
\cite{AcciaioBeiglbockSchachermayerPenkner16}, Burzoni et al. \cite%
{BurzoniFrittelliMaggis16}, Burzoni et al. \cite{BurzoniFrittelliMaggis17},
Burzoni et al. \cite{BurzoniFrittelliHouMaggisObloj19}, Riedel \cite%
{Riedel15}. In such a framework, Optimal Transport theory became a very
powerful tool to prove pathwise pricing hedging duality results with very
relevant contributions by many authors (Beiglb{\"{o}}ck et al. \cite{Beig},
Davis et al. \cite{DavisOblojRaval14}, Dolinksi and Soner \cite%
{DolinksiSoner14}, Dolinsky and Soner \cite{DolinksiSoner15}; Galichon et
al. \cite{GalichonLabordereTouzi14}, Henry-Labord\`{e}re \cite{Labordere13},
Henry-Labord\`{e}re et al. \cite{LabordereOblojSpoidaTouzi16}; Hou and Ob\l 
\'{o}j \cite{HouObloj18}, Tan and Touzi \cite{TanTouzi13}). Recent works on
the topic include also Bartl et al.  \cite{BartlKupperPromelTangpi19},
Cheridito et al.  \cite{Cheridito2020ER}, Guo and Ob\l \'{o}j  
\cite{GuoObloj19}, Wiesel  \cite{Wiesel2019ContinuityOT}. Stability
issues have been studied in Backhoff-Veraguas and Pammer  \cite%
{BackhoffVeraguasPammer20} and Neufeld and Sester  \cite%
{NeufeldSester21}. These contributions mainly deal with what we labeled
above as the sublinear case, while our main interest in this paper is to
develop the convex case theory, as explained below. Pennanen and Perkki\"{o} 
\cite{PennanenPerkkio19}  also developed a generalized Optimal
Transport duality, which can be applied to study the pricing-hedging duality
in a context similar to our additive setup of Section \ref{secadditive}.

The addition of an entropic term to optimal transport problems was
popularized by Cuturi  \cite{Cuturi13}, with several applications
especially from the computational point of view (see for examples the
survey/monograph Peyré and Cuturi  \cite{PeyreCuturi19}). Sinkhorn’s algorithm can be applied with the entropic regularization procedure described in these works (see Benamou et al.  \cite{BenamouCarlierCuturi15} for some advantages). Convergence for this algorithm is studied e.g. in Ireland and Kullback  
\cite{IrelandKullback68} and R\"{u}schendorf  \cite{Ruschendorf95}. After the present paper was posted on ArXiv, several relevant advances were made
regarding such topic. We mention here Nutz and Wiesel (2021) \cite%
{NutzWiesel21}, Berntonet al. (2021) \cite{BerntonGhosalNutz21}, Ghosal et al. (2021) \cite%
{BerntonGhosalNutz21b}. We stress that these papers address a \textit{%
different} problem: the addition of the entropic term in Cuturi and
subsequent works is made without smoothing the strict marginal constraints,
which are kept, unlike in our problems where we add uncertainty regarding
the marginals themselves. The works \cite{BerntonGhosalNutz21} and \cite%
{BerntonGhosalNutz21b} also study geometric properties of minimizers of the entropic OT, by means of the concept of cyclical invariance. This is a counterpart to the characterization, using $c$-cyclical monotonicity,  of the geometry of optimal transport plans in the classical framework of OT.  Even though a similar study of geometric properties for optimizers of EMOT would be of great interest,  such a  topic  is beyond the scope of this paper and is left for future research.
We also mention that a Sinkhorn algorithm approach was adopted in De March
and Henry-Labord\`{e}re  \cite{DeMarchHenryLabordere20} for building
an arbitrage-free implied volatility surface from bid-ask quotes, while
Henry-Labord\`{e}re  \cite{HenryLabordere19} studies a problem related
to the entropic relaxation of an optimal transportation problem and Blanchet et al. \cite{BlanchetJambulapatiKentSidford20}  studies the number of operations needed for approximation of the transport cost with given accuracy, in the case of entropic regularization.

In the framework of \cite{Sava} (i.e., with penalizations of the marginals induced by divergence functions) and after the first version of the present work was posted on ArXiv, duality results were obtained in the context of Weak Martingale Optimal Entropy Transport Problems by Chung and Trinh (2021) \cite{ChungTrinh21}.

\bigskip

To introduce the pathwise robust approach to the pricing-hedging duality,
from now on we will work without a reference probability measure. We
consider $T\in \mathbb{N}$, $T\geq 1$, and 
\begin{equation*}
\label{prodcpts}
\Omega :=K_{0}\times \dots \times K_{T}
\end{equation*}%
for closed (possibly non compact) subsets $K_{0},\dots ,K_{T}$ of $\mathbb{R}
$ and denote with $X_{0},\dots ,X_{T}$ the canonical projections $%
X_{t}:\Omega \rightarrow K_{t}$, for $t=0,1,...,T$. We will work later in a
multidimensional setup for the price process, but we stick to the one
dimensional case here for notational simplicity. We write
\begin{equation*}
\mathrm{Mart}(\Omega ):=\{\text{Martingale probability measures for the
canonical process of }\Omega \}\,,
\end{equation*}%
and, when $\mu $ is a measure defined on the Borel $\sigma $-algebra of $%
(K_{0}\times \dots \times K_{T})$, its marginals will be denoted with $\mu
_{0},\dots ,\mu _{T}$. We consider a contingent claim $c:\Omega \rightarrow
(-\infty ,+\infty ]$ which is now allowed to depend on the whole path and we
admit semistatic trading strategies for hedging. This means that in
addition to dynamic trading in $X$ via the admissible integrands $\Delta \in 
\mathcal{H}$, we may invest in \textquotedblleft vanilla\textquotedblright\
options $\varphi _{t}:K_{t}\rightarrow \mathbb{R}$. For modeling purposes we
take vector subspaces $\mathcal{E}_{t}\subseteq \mathcal{C}(K_{t})$ for $%
t=0,\dots ,T$ , where $\mathcal{C}(K_{t})$ is the space of real-valued
continuous functions on $K_{t}$. Some requirements on the spaces $\mathcal{E}%
_{t}$ will be added later on when describing the technical setup. For each $%
t,$ $\mathcal{E}_{t}$ is the set of static options that can be used for
hedging, say affine combinations of vanilla options with different strikes
and same maturity $t$. The key assumption in the robust, Optimal Transport
based formulation is that the marginals $\mathcal{(}\widehat{Q}_{0},\widehat{%
Q}_{1},...,\widehat{Q}_{T})$ of the underlying price process $X$ are known,
see the seminal papers by Breeden and Litzenberger \cite%
{BreedenLitzenberger78} and Hobson \cite{Hobson98}, as well as the many
contributions by Hobson \cite{Hobson11}, Cox and Ob\l \'{o}j \cite%
{CoxObloj11a}, \cite{CoxObloj11b}, Cox and Wang \cite{CoxWang13}, Labord\`{e}%
re et al. \cite{LabordereOblojSpoidaTouzi16}, Brown et al. \cite%
{BrownHobsonRogers01}, Hobson and Klimmerk \cite{HobsonKlimmerk13}. Such
marginals can be identified knowing a (very) large number of prices of plain
vanilla options maturing at each intermediate date, for example knowing the
prices of \textit{all the call options} with such intermediate maturities
and ranging strikes. In this case, the class of arbitrage-free pricing
measures that are compatible with the observed prices of the options is
given by%
\begin{equation*}
\mathcal{M}(\widehat{Q}_{0},\widehat{Q}_{1},...\widehat{Q}_{T}):=\left\{
Q\in \mathrm{Mart}(\Omega )\mid X_{t}\sim _{Q}\mathcal{\widehat{Q}}_{t}\text{
for each }t=0,\dots ,T\right\} .
\end{equation*}%
Let $\mathcal{C}_{b}(K_{0}\times \dots \times K_{t})$ consist of real valued
continuous and bounded functions defined on $K_{0}\times \dots \times K_{t}$
and set 
\begin{align}
\mathcal{H}& :=\left\{ \Delta =[\Delta _{0},\dots ,\Delta _{T-1}]\mid \Delta
_{t}\in \mathcal{C}_{b}(K_{0}\times \dots \times K_{t})\right\} ,
\label{strategies} \\
\mathcal{I}& :=\left\{ I^{\Delta }(x)=\sum_{t=0}^{T-1}\Delta
_{t}(x_{0},\dots ,x_{t})(x_{t+1}-x_{t})\mid \Delta \in \mathcal{H}\right\} .
\label{StoIntegra}
\end{align}%
In this framework, the sub-hedging duality, obtained in \cite{Beig} Theorem
1.1, takes the form
\begin{equation}
\inf_{Q\in \mathcal{M}(\widehat{Q}_{0},\widehat{Q}_{1},...\widehat{Q}%
_{T})}E_{Q}\left[ c\right] =\sup \left\{ \sum_{t=0}^{T}E_{\widehat{Q}%
_{t}}[\varphi _{t}]\mid \exists \Delta \in \mathcal{H}\text{ s.t. }%
\sum_{t=0}^{T}\varphi _{t}(x_{t})+I^{\Delta }(x)\leq c(x)\text{ }\forall
x\in \Omega \right\} ,  \label{robustDuality}
\end{equation}%
where the RHS of (\ref{robustDuality}) is known as the \textbf{robust
subhedging price} of $c$. Comparing (\ref{robustDuality}) with the duality
between (\ref{subrep}) and (\ref{dualp}), we observe that: (i) the $P-$a.s.
inequality in (\ref{subrep}) has been replaced by an inequality that holds
for all $x\in \Omega $; (ii) in (\ref{robustDuality}) the infimum of the
price of the contingent claim $c$ is taken under all martingale measure
compatible with the option prices, with no reference to the probability $P$;
(iii) static hedging with options is allowed.

As can be seen from the LHS of (\ref{robustDuality}), this case falls into
the category labeled above as the \textit{sublinear case}, and the purpose
of this paper is to investigate the \textit{convex case}, in the robust
setting, using the tools from Entropy Optimal Transport (EOT) recently
developed in Liero et al. \cite{Sava}.

Let us first describe the financial interpretation of the problems that we
are going to study.

\paragraph{The dual problem}

The LHS of (\ref{robustDuality}), namely $\inf_{Q\in \mathcal{M(}\widehat{Q}%
_{0},\widehat{Q}_{1},...\widehat{Q}_{T})}E_{Q}\left[ c\right] $, represents
the dual problem in the financial application, but is typically the primal
problem in \textbf{Martingale Optimal Transport} (MOT). 
We label this case as the \textit{sublinear case} of MOT. In \cite{Sava},
the primal \textbf{Entropy Optimal Transport} (EOT) problem takes the form%
\begin{equation}
\inf_{\mu \in \mathrm{Meas}\mathcal{(}\Omega )}\left( \int_{\Omega }c\mathrm{%
d}\mu +\sum_{t=0}^{T}\mathcal{D}_{F_{t},\widehat{Q}_{t}}(\mu _{t})\right) ,
\label{EOT}
\end{equation}%
where $\mathrm{Meas}(\Omega )$ is the set of all positive finite measures $%
\mu $ on $\Omega ,$ and $\mathcal{D}_{F_{t},\widehat{Q}_{t}}(\mu _{t})$ is a
divergence in the form: 
\begin{equation}
\mathcal{D}_{F_{t},\widehat{Q}_{t}}(\mu _{t}):=\int_{K_{t}}F_{t}\left( \frac{%
\mathrm{d}\mu _{t}}{\mathrm{d}\widehat{Q}_{t}}\right) \,\mathrm{d}\widehat{Q}%
_{t}\text{, if }\mu _{t}\ll \widehat{Q}_{t}\text{,\quad }t=0,...,T,
\label{Divergence}
\end{equation}%
otherwise $\mathcal{D}_{F_{t},\widehat{Q}_{t}}(\mu _{t}):=+\infty $. We
label with $F:=(F_{t})_{_{t=0,...,T}}$ the family of divergence functions $%
F_{t}:\mathbb{R}\rightarrow \mathbb{R\cup }\left\{ +\infty \right\} $
appearing in (\ref{Divergence}). Problem \eqref{EOT} represents the \textit{%
convex case} of OT theory. Notice that in the EOT primal problem \eqref{EOT}
the typical constraint that $\mu $ has prescribed marginals $(\widehat{Q}%
_{0},\widehat{Q}_{1},...\widehat{Q}_{T})$ is relaxed (as the infimum is
taken with respect to all positive finite measures) by introducing the
divergence functional $\mathcal{D}_{F_{t},\widehat{Q}_{t}}(\mu _{t})$, which
penalizes those measures $\mu $ that are \textquotedblleft
far\textquotedblright\ from some reference marginals $(\widehat{Q}_{0},%
\widehat{Q}_{1},...\widehat{Q}_{T}).$ We are then naturally led to the study
of the convex case of MOT, i.e. to the Entropy Martingale Optimal Transport
(EMOT) problem%
\begin{equation}
\mathfrak{D}_{F,\widehat{Q}}(c):=\inf_{Q\in \mathrm{Mart}(\Omega )}\left(
E_{Q}\left[ c\right] +\sum_{t=0}^{T}\mathcal{D}_{F_{t},\widehat{Q}%
_{t}}(Q_{t})\right)  \label{D1div}
\end{equation}%
having also a clear financial interpretation. The marginals are not any more
fixed a priori, as in (\ref{robustDuality}), because we may not have
sufficient information to detect them with enough accuracy, for example, in
case there are not sufficiently many traded call and put options on the
underlying assets in the market and it is not possible to extract precisely
the marginals via the Breeden and Litzenberger \cite{BreedenLitzenberger78}
approach. Alternatively, the exact prices of the options might be unknown,
say by market impact effects. 

Again, the infimum is taken over \textit{all
martingale probability measures}, but those that are far from some estimate $%
(\widehat{Q}_{0},\widehat{Q}_{1},...\widehat{Q}_{T})$ are appropriately
penalized through $\mathcal{D}_{F_{t},\widehat{Q}_{t}}$. This is a key
difference with classical MOT problem and also with the approach taken in the above mentioned works
\cite{BerntonGhosalNutz21}, \cite{BerntonGhosalNutz21b}, \cite{Cuturi13}, 
\cite{NutzWiesel21} and \cite{PeyreCuturi19}. 
 When $\mathcal{D}_{F_{t},%
\widehat{Q}_{t}}(\cdot )=\delta _{\widehat{Q}_{t}}(\cdot )$ , the EMOT reduces to the classical MOT problem, where only martingale probability
measures with fixed marginals are allowed. Here $\delta_{A}$ is the characteristic function of  a set $A$, as customarily defined in convex analysis. 
 Our framework also allows for the use of the penalization  $\probq\mapsto\sum_{t=0}^{T}\delta _{\widehat{Q}_{t}}(Q_t )+\widetilde{\mathcal{}D}(Q) $, for some entropic term $\widetilde{\mathcal{}D}$, so that the EMOT
reduces to the MOT problem with an additional entropic regularization term, as analyzed in the above mentioned literature.

Observe that in addition to the
martingale property, the elements $Q\in \mathrm{Mart}(\Omega )$ in (\ref%
{D1div}) are required to be probability measures, while in the EOT theory in
(\ref{EOT}) all positive finite measures are allowed. As it was recalled
after equation (\ref{dualconvex}), this normalization feature of the dual
elements ($\mu (\Omega )=1$) is not surprising when one deals with dual
problems of primal problems with a cash additive objective functional as,
for example, in the theory of coherent and convex risk measures.

Potentially, we could push our smoothing argument above even further: in
place of the functionals $\mathcal{D}_{F_{t},\widehat{Q}_{t}}(\mu _{t})$, $%
t=0,...,T$, we might as well consider more general marginal penalizations,
not necessarily in the divergence form \eqref{Divergence}, yielding the
problem 
\begin{equation}
\mathfrak{D}(c):=\inf_{Q\in \mathrm{Mart}(\Omega )}\left( E_{Q}\left[ c%
\right] +\sum_{t=0}^{T}\mathcal{D}_{t}(Q_{t})\right) \,.  \label{D1}
\end{equation}%
These penalizations $\mathcal{D}_{0},\dots ,\mathcal{D}_{T}$ will be better
specified later. Here we only observe that such penalization terms could be induced by market prices or by a Wassaerstein distance.

\paragraph{The primal problem: the Nonlinear Subhedging Value}

We provide the financial interpretation of the primal problem which will
yield the EMOT problem $\mathfrak{D}_{F,\widehat{Q}}$ as its dual. It is
convenient to reformulate the robust subhedging price in the RHS of (\ref%
{robustDuality}) in a more general setting.

\begin{definition}
\label{def11}Consider a measurable function $c:\Omega \rightarrow {\mathbb{R}%
}$ representing a (possibly path dependent) option, the set $\mathcal{V}$ of
hedging instruments and a suitable pricing functional $\pi :\mathcal{V}%
\rightarrow {\mathbb{R}}$. Then the robust Subhedging Value of $c$ is
defined by 
\begin{equation*}
\Pi _{\pi ,\mathcal{V}}(c)=\sup \left\{ \pi (v)\mid v\in \mathcal{V}\text{
s.t. }v\leq c\right\} .
\end{equation*}
\end{definition}

In the classical setting, functionals of this form (and even with a more
general formulation) are known as general capital requirement, see for
example Frittelli and Scandolo \cite{FrittelliScandolo06}. We stress however
that in Definition \ref{def11} the inequality $v\leq c$ holds for all
elements in $\Omega $ with no reference to a probability measure whatsoever.
The novelty in this definition is that a priori $\pi $ may not be linear and
it is crucial to understand which evaluating functional $\pi $ we may use.
For our discussion, we assume that the vector subspaces $\mathcal{E}%
_{t}\subseteq \mathcal{C}(K_{t})$ satisfies $\mathcal{E}_{t}+\mathbb{R=}%
\mathcal{E}_{t}$, for $t=0,\dots ,T$. We$\ $let $\mathcal{E}:=\mathcal{E}%
_{0}\times \dots \times \mathcal{E}_{T},$ and $\mathcal{V}:=\mathcal{E}%
_{0}+\dots +\mathcal{E}_{T}+\mathcal{I}$. Suppose we took a linear pricing
rule $\pi :\mathcal{V}\rightarrow {\mathbb{R}}$ defined via a $\widehat{Q}%
\in \mathrm{Mart}(\Omega )$ by 
\begin{equation}
\pi (v):=E_{\widehat{Q}}\left[ \sum_{t=0}^{T}\varphi _{t}+I^{\Delta }\right] 
\overset{(i)}{=}E_{\widehat{Q}}\left[ \sum_{t=0}^{T}\varphi _{t}\right] 
\overset{(ii)}{=}\sum_{t=0}^{T}E_{\widehat{Q}_{t}}[\varphi _{t}],  \label{01}
\end{equation}%
where we used (\ref{EE}) and the fact that $\widehat{Q}_{t}$ is the marginal
of $\widehat{Q}$. In this case, we would trivially obtain for the robust
subhedging value of $c$ 
\begin{align}
\Pi _{\pi ,\mathcal{V}}(c)& =\sup \left\{ \pi (v)\mid v\in \mathcal{V}\text{
s.t. }v\leq c\right\}  \label{11} \\
& =\sup \left\{ \sum_{t=0}^{T}E_{\widehat{Q}_{t}}[\varphi _{t}]\mid \varphi
\in \mathcal{E}\text{ is s.t. }\exists \Delta \in \mathcal{H}\text{ s.t. }%
\sum_{t=0}^{T}\varphi _{t}(x_{t})+I^{\Delta }(x)\leq c(x)\text{ }\forall
x\in \Omega \right\}  \notag \\
& =\sup \left\{ m\in {\mathbb{R}}\mid \exists \Delta \in \mathcal{H}\text{, }%
\varphi \in \mathcal{E},\text{ s.t. }m-\sum_{t=0}^{T}E_{\widehat{Q}%
_{t}}[\varphi _{t}]+\sum_{t=0}^{T}\varphi _{t}+I^{\Delta }\leq c\right\} 
\notag \\
& =\sup \left\{ m\in {\mathbb{R}}\mid \exists \Delta \in \mathcal{H}\text{, }%
\varphi \in \mathcal{E},\text{ with }E_{\widehat{Q}_{t}}[\varphi _{t}]=0%
\text{ s.t. }m+\sum_{t=0}^{T}\varphi _{t}+I^{\Delta }\leq c\right\} ,
\label{14}
\end{align}%
where in the last equality we replaced $\varphi _{t}$ with $(E_{\widehat{Q}%
_{t}}[\varphi _{t}]-\varphi _{t})\in \mathcal{E}_{t}$, which satisfies:%
\begin{equation}
E_{\widehat{Q}_{t}}\left[ E_{\widehat{Q}_{t}}[\varphi _{t}]-\varphi _{t}%
\right] =0.  \label{15}
\end{equation}%
\textit{Interpretation}: $\Pi _{\pi ,\mathcal{V}}(c)$ is the supremum amount 
$m\in {\mathbb{R}}$ for which we may buy options $\varphi _{t}$ and dynamic
strategies $\Delta \in \mathcal{H}$ such that $m+\sum_{t=0}^{T}\varphi
_{t}+I^{\Delta }\leq c$ , where the value of both the options and the
stochastic integrals are computed as the expectation under the same
martingale measure ($\widehat{Q}$ for the integral $I^{\Delta }$; its
marginals $\widehat{Q}_{t}$ for each option $\varphi _{t}$).

\bigskip

However, as mentioned above when presenting the indifference price $p^{b}$,
there is a priori no reason why one has to allow only linear functional in
the evaluation of $v\in \mathcal{V}$.

\textit{We thus generalize the expression for} $\Pi _{\pi ,\mathcal{V}}(c)$ 
\textit{by considering valuation functionals} $S:\mathcal{V}\rightarrow {%
\mathbb{R}}$ \textit{and} $S_{t}:\mathcal{E}_{t}\rightarrow {\mathbb{R}}$ \textit{%
more general than} $E_{\widehat{Q}}[\cdot ]$ \textit{ and }  $E_{\widehat{Q}%
_{t}}[\cdot ].$

Nonetheless, in order to be able to repeat the same key steps we used in (%
\ref{11})-(\ref{14}) and therefore to keep the same interpretation, we shall
impose that such functionals $S$ and $S_{t}$ satisfy the property in (\ref%
{15}) and the two properties (i) and (ii) in equation (\ref{01}), that is:

\begin{enumerate}
\item[(a)] $S_{t}[\varphi _{t}+k]=S_{t}[\varphi _{t}]+k$ and $S_{t}[0]=0,$
for all $\varphi _{t}\in \mathcal{E}_{t}$, $k\in \mathbb{R}$, $t=0,\dots ,T$.

\item[(b)] $S\left[ \left( \sum\limits_{t=0}^{T}\varphi _{t}\right)
+I^{\Delta }(x)\right] =S\left[ \sum\limits_{t=0}^{T}\varphi _{t}\right] $
for all $\Delta \in \mathcal{H}$ and $\varphi \in \mathcal{E}.$

\item[(c)] $S\left[ \sum\limits_{t=0}^{T}\varphi _{t}\right]
=\sum\limits_{t=0}^{T}S_{t}[\varphi _{t}]$ for all $\varphi \in \mathcal{E}$.
\end{enumerate}

We immediately recognize that (a) is the Cash Additivity (CA) property on $%
\mathcal{E}_{t}$ of the functional $S_{t}$ and (b) implies the Integral
Additivity (IA) property on $\mathcal{V}$. As a consequence, repeating the
same steps in (\ref{11})-(\ref{14}), we will obtain as primal problem the 
\textbf{nonlinear subhedging value} of $c:$ 
\begin{align}
\mathfrak{P}(c)& =\sup \left\{ S(v)\mid v\in \mathcal{V}:v\leq c\right\} 
\notag \\
& =\sup \left\{ \sum_{t=0}^{T}S_{t}(\varphi _{t})\mid \varphi \in \mathcal{E}%
\text{ is s.t. }\exists \Delta \in \mathcal{H}\text{ s.t. }%
\sum_{t=0}^{T}\varphi _{t}(x_{t})+I^{\Delta }(x)\leq c(x)\text{ }\forall
x\in \Omega \right\}  \label{PI} \\
& =\sup \left\{ m\in {\mathbb{R}}\mid \exists \Delta \in \mathcal{H}\text{, }%
\varphi \in \mathcal{E},\text{ with }S_{t}(\varphi _{t})=0\text{ s.t. }%
m+\sum_{t=0}^{T}\varphi _{t}+I^{\Delta }\leq c\right\} ,  \notag
\end{align}%
to be compared with (\ref{14}).

\textit{Interpretation}: $\mathfrak{P}(c)$ is the supremum amount $m\in {%
\mathbb{R}}$ for which we may buy \textit{zero\ value} options $\varphi _{t}$
and dynamic strategies $\Delta \in \mathcal{H}$ such that $%
m+\sum_{t=0}^{T}\varphi _{t}+I^{\Delta }\leq c$, where the value of both the
options and the stochastic integrals are computed with the same functional $%
S.$

It is easy to check that \eqref{PI} can be rewritten as: 
\begin{equation}
\mathfrak{P}(c)=\sup_{\Delta \in \mathcal{H}}\sup_{\varphi \in \mathbf{\Phi }%
_{\Delta }(c)}\sum_{t=0}^{T}S_{t}(\varphi _{t})  \label{dualitywithSt}
\end{equation}%
where 
\begin{equation*}
\mathbf{\Phi }_{\Delta }(c):=\left\{ \varphi \in \mathcal{E},\text{ }%
\sum_{t=0}^{T}\varphi _{t}(x_{t})+I^{\Delta }(x)\leq c(x)\,\,\,\forall
\,x\in \Omega \right\} \,.
\end{equation*}%

It is quite a natural question to ask whether one can induce valuations  satisfying (a),(b),(c) above, especially starting from given functionals $S_t,t=0,\dots,T$. This can be obtained with the concept of Stock Additivity which we now discuss, before further elaborating on the formulation in (\ref{dualitywithSt}). Stock Additivity is the natural
counterpart of properties (IA) and (CA) when we are evaluating hedging
instruments depending solely on the value of the underlying stock $X$ at
some fixed date $t\in \{0,\dots ,T\}$. Let $X_{t}$ be the identity function $%
x_{t}\mapsto x_{t}$ on $K_{t}$, which can also be thought as the projection
on the $t$-th component $X:\Omega \rightarrow {\mathbb{R}}$. As before, the
set of hedging instruments is denoted by $\mathcal{E}_{t}\subseteq \mathcal{C%
}(K_{t})$ and we will suppose that $X_{t}\in \mathcal{E}_{t}$ (that is, we
can use units of stock at time $t$ for hedging) and that $\mathcal{E}_{t}+{%
\mathbb{R}}=\mathcal{E}_{t}$ (that is, deterministic amounts of cash can be
used for hedging as well). Finally, we assume for the following Definition
that the value of $X_{0}$ is known, namely $K_{0}=\{x_{0}\}$ for some $%
x_{0}\in {\mathbb{R}}$.

\begin{definition}
A functional $\ p_{t}:\mathcal{E}_{t}\rightarrow {\mathbb{R}}$ is stock
additive on $\mathcal{E}_{t}$ if $p_{t}(0)=0$ and%
\begin{equation*}
p_{t}(\varphi _{t}+\alpha _{t}X_{t}+\beta _{t})=p_{t}(\varphi _{t})+\alpha
_{t}x_{0}+\beta _{t}\,\,\,\,\,\,\,\forall \varphi _{t}\in \mathcal{E}%
_{t},\beta _{t}\in {\mathbb{R}},\alpha _{t}\in {\mathbb{R}}\,,
\end{equation*}
\end{definition}

We now clarify the role of stock additive functionals in our setup. Suppose
that $S_{t}:\mathcal{E}_{t}\rightarrow {\mathbb{R}}$ are stock additive on $%
\mathcal{E}_{t}$, $t=0,\dots ,T$. It can be shown (see Lemma \ref%
{lemmaintrostockaddfunctions}) that if there exist $\varphi ,\psi \in 
\mathcal{E}_{0}\times ...\times \mathcal{E}_{T}$ and $\Delta \in \mathcal{H}$
such that $\sum_{t=0}^{T}\varphi _{t}=\sum_{t=0}^{T}\psi _{t}+I^{\Delta }$ \
then 
\begin{equation*}
\sum_{t=0}^{T}S_{t}(\varphi _{t})=\sum_{t=0}^{T}S_{t}(\psi _{t}).
\end{equation*}%
This allows us to define a functional $S:\mathcal{V}=\mathcal{E}_{0}+\dots +%
\mathcal{E}_{T}+\mathcal{I}\rightarrow {\mathbb{R}}$ by 
\begin{equation}
S(\upsilon ):=\sum_{t=0}^{T}S_{t}(\varphi _{t}),\text{\quad for }\upsilon
=\sum_{t=0}^{T}\varphi _{t}+I^{\Delta }.  \label{SI}
\end{equation}%
Then $S$ is a well defined, integral additive functional on $\mathcal{V}$,
and $S,S_{0},\dots ,S_{T}$ satisfy the properties (a), (b), (c). 

In conclusion, when we consider stock additive functionals $S_{0},\dots ,S_{T}$
that induce the functional $S$ as explained in (\ref{SI}), we can focus our
attention to the optimization problem \eqref{dualitywithSt}, that will be
referred to as our primal problem. 

As explained in Example \ref{exutstockadd} below, there is a natural way to produce a variety of stock additive functionals,

\begin{example}
\label{exutstockadd} Consider a martingale measure $\widehat{Q}\in \mathrm{%
Mart}(\Omega )$ and a concave non decreasing utility function $u_{t}:{%
\mathbb{R}}\rightarrow \lbrack -\infty ,+\infty ),$ satisfying $u(0)=0$ and $%
u_{t}(x_{t})\leq x_{t}\,\ \forall x_{t}\in {\mathbb{R}}$. We can then take 
\begin{equation*}
S_{t}(\varphi _{t})=U_{\widehat{Q}_{t}}(\varphi _{t}):=\sup_{\alpha \in {%
\mathbb{R}},\text{ }\beta \in {\mathbb{R}}}\left( \int_{\Omega }u_{t}\left(
\varphi _{t}(x_{t})+\alpha x_{t}+\beta \right) \,\mathrm{d}\widehat{Q}%
_{t}(x_{t})-(\alpha x_{0}+\beta )\right) .
\end{equation*}%
As shown in Lemma \ref{remfrominftytofiniteness} the stock additivity
property is then satisfied for these functionals.
\end{example}

\paragraph{The Duality}

As a consequence of our main results we prove the following duality (see
Theorem \ref{mainthm3noncpt}). If 
\begin{equation*}
\mathcal{D}_{t}(Q_{t}):=\sup_{\varphi _{t}\in \mathcal{E}_{t}}\left(
S_{t}(\varphi _{t})-\int_{K_{t}}\varphi _{t}\,\mathrm{d}Q_{t}\right) \quad 
\text{for }Q_{t}\in \mathrm{Prob}(K_{t})\text{,}\quad t=0,\dots ,T,
\end{equation*}%
and $\mathfrak{D}(c)$ and $\mathfrak{P}(c)$ are defined respectively in %
\eqref{D1} and (\ref{dualitywithSt}), then 
\begin{equation}
\mathfrak{D}(c)=\mathfrak{P}(c).\label{dualityPD}
\end{equation}%
In the particular case of $S_{0},\dots ,S_{T}$ induced by utility functions,
as explained in Example \ref{exutstockadd}, the problem corresponding to %
\eqref{PI} or to \eqref{dualitywithSt} becomes
\begin{equation}
\mathfrak{P}_{U,\widehat{Q}}(c)=\sup \left\{ \sum_{t=0}^{T}U_{\widehat{Q}%
_{t}}(\varphi _{t})\mid \varphi \in \mathcal{E}\text{ is s.t. }\exists
\Delta \in \mathcal{H}\text{ s.t. }\sum_{t=0}^{T}\varphi
_{t}(x_{t})+I^{\Delta }(x)\leq c(x)\text{ }\forall x\in \Omega \right\} .
\label{PIdiv}
\end{equation}%
We also show the duality between (\ref{D1div}) and (\ref{PIdiv}), 
namely we prove in Section \ref{secdivinducedbyutils} 
\begin{equation}
\mathfrak{D}_{F,\widehat{Q}}(c):=\inf_{Q\in \mathrm{Mart}(\Omega )}\left(
E_{Q}\left[ c\right] +\sum_{t=0}^{T}\mathcal{D}_{F_{t},\widehat{Q}%
_{t}}(Q_{t})\right) =\mathfrak{P}_{U,\widehat{Q}}(c).  \label{key}
\end{equation}


The divergence functions $F_{t}$ appearing in $\mathfrak{D}_{F,\widehat{Q}}$
(via $\mathcal{D}_{F_{t},\widehat{Q}_{t}}$) are associated to the utility
functions $u_{t}$ appearing in $U_{\widehat{Q}_{t}}$ and in $\mathfrak{P}_{U,%
\widehat{Q}}$ via the conjugacy relation: 
\begin{equation*}
F_{t}(y):=v_{t}^{\ast }(y)=\sup_{x_{t}\in \mathbb{R}}\left\{
x_{t}y-v(y)\right\} =\sup_{x_{t}\in \mathbb{R}}\left\{
u_{t}(x_{t})-x_{t}y\right\} ,
\end{equation*}%
where $v(y):=-u(-y)$. Thus, depending on which utility function $u$ is
selected in the primal problem $\mathfrak{P}_{U,\widehat{Q}}(c)$ to evaluate
the options through $U_{\widehat{Q}_{t}}$, the penalization term $\mathcal{D}%
_{F_{t},\widehat{Q}_{t}}$ in the the dual formulation $\mathfrak{D}_{F,%
\widehat{Q}}(c)$ has a particular form induced by $F_{t}=v_{t}^{\ast }$. In
the special case of linear utility functions $u_{t}(x_{t})=x_{t}$, we
recover the sublinear MOT theory. Indeed, in this case, $v_{t}^{\ast
}(y)=+\infty ,$ for all $y\neq 1$ and $v_{t}^{\ast }(1)=0$, so that $%
\mathcal{D}_{F_{t},\widehat{Q}_{t}}(\cdot )=\delta _{\widehat{Q}_{t}}(\cdot
) $ and thus we obtain the robust pricing-hedging duality (\ref%
{robustDuality}) of the classical MOT.

\subsection{EMOT}

\label{introEMOTdiscussion} 

To describe our main result (Theorem \ref{mainEMOTtheoremgeneral}) we
introduce two general functionals $U$ and $\mathcal{D}_{U}$
that are associated through a Fenchel-Moreau type relation, see 
\eqref{dualrepdivgenabstract}. 
The valuation functional $U:\mathcal{E}\rightarrow \lbrack -\infty
,+\infty )$ is defined on the space of hedging instruments $\mathcal{E}$, where the vector space $\mathcal{E}$ consists of
vectors of continuous functions $\varphi =[\varphi _{0},\dots ,\varphi _{T}]$
satsfying the sublinearity condition of the form $\varphi _{t}(x_{t})\leq
\alpha (1+\left\vert x_{t}\right\vert )$ for some $\alpha \geq 0$. As the map $%
U $ is not necessarily cash additive, we rely on the notion of the
Optimized Certainty Equivalent (OCE), that was introduced in Ben Tal and
Teboulle \cite{BenTalTeboulle86} and further analyzed in Ben Tal and
Teboulle \cite{BenTalTeboulle07}.  We introduce the Generalized Optimized Certainty Equivalent
associated to $U$ as the functional $S^{U}:\mathcal{E}\rightarrow \lbrack
-\infty ,+\infty ]$ defined by 
\begin{equation}
S^{U}(\varphi ):=\sup_{\beta \in {\mathbb{R}}^{T+1}}\left( U(\varphi +\beta
)-\sum_{t=0}^{T}\beta _{t}\right) \text{,}\,\,\,\,\,\varphi \in \mathcal{E}%
\text{.}  \label{defSU}
\end{equation}%
Thus we obtain a cash additive map $S^{U}(\varphi +\beta )=S^{U}(\varphi
)+\sum_{t=0}^{T}\beta _{t}$, which will guarantee that in the problem (\ref%
{EOT}) the elements $\mu \in $Meas$(\Omega )$ are normalized, i.e. are
probability measures. Then the duality \eqref{dualityPD}, $\mathfrak{D}(c)=\mathfrak{P}(c),$ will
take the form%
\begin{equation}
\inf_{Q\in \mathrm{Mart}(\Omega )}\left( E_{Q}\left[ c\right] +\mathcal{D}%
_{U}(Q)\right) =\sup_{\Delta \in \mathcal{H}}\sup_{\varphi \in \mathbf{\Phi }%
_{\Delta }(c)}S^{U}\left( \varphi \right) ,  \label{infsup}
\end{equation}%
and we will also prove the existence of the optimizer for the problem in the LHS
of \eqref{infsup}. The primal problem in \eqref{dualitywithSt} can be
recovered from the more general expression in the RHS of \eqref{infsup}:
indeed, it is enough to take $U(\varphi ):=\sum_{t=0}^{T}S_{t}(\varphi _{t})$%
, with stock additive functional $S_{t},$ and observe that by cash
additivity $S^{U}(\varphi )=U(\varphi )=\sum_{t=0}^{T}S_{t}(\varphi _{t})$.
The generality consists in considering valuation of the \textit{process} $\varphi =[\varphi _{0},\dots ,\varphi _{T}] \in \mathcal{E}$ 
rather than the valuation of the terminal values $\varphi_{t}(x_t)$ only.

The penalization term $\mathcal{D}$:=$\mathcal{D}_{U}$ associated to $U$
does not necessarily have an additive structure, $\mathcal{D}
(Q)=\sum_{t=0}^{T}\mathcal{D}_{t}(Q_{t})$, as in \eqref{D1}, nor needs to have the divergence
formulation, as described in (\ref{Divergence}), and so it does not
necessarily depend on a given martingale measure $\widehat{Q}$ (see e.g.  Section \ref{CasePenalization1gen}). This additional
flexibility in choosing $\mathcal{D}$ constitutes one key  generalization of the Entropy Optimal
Transport theory of \cite{Sava}. Of course, the other additional difference
with EOT is the presence in (\ref{infsup}) of the additional supremum with
respect to admissible integrand $\Delta \in \mathcal{H}$. As a consequence,
in the LHS of (\ref{infsup}) the infimum is now taken with respect to
martingale measures. We also point out that in \cite{Sava}, the cost
functional $c$ is required to be lower semicontinuous and nonnegative and
that the theory is developed only for the bivariate case ($t=0,1)$, while in
this paper we take $c$ lower semicontinuous and with superlinear growth (as
given by \eqref{controlfrombelow}), and consider the multivariate case ($%
t=0,...,T)$ and a multidimensional process. %

Our framework allows to establish and comprehend several different duality
results, even if under different type of assumptions:

\begin{enumerate}
\item The new non linear robust pricing-hedging duality with options
described in (\ref{key}) and proved in
Corollary \ref{corolldivinduced}.\label{elenco}

\item The new non linear robust pricing-hedging duality with options and
singular components, proved in Corollary \ref{corwithsingular}.

\item The linear robust pricing-hedging duality with options (see \cite{Beig}
Theorem 1.1, or  \cite{AcciaioBeiglbockSchachermayerPenkner16}
Theorem 1.4) described in (\ref{robustDuality}) and proved in Corollary \ref%
{CorNew}.

\item The linear robust pricing-hedging duality without options (see for
example  \cite{BurzoniFrittelliMaggis17} Theorem 1.1) proved
in Corollary \ref{corollrob}.

\item A new robust pricing-hedging duality with penalization function based
on market data (see Section \ref{CasePenalization1gen}).

\item A new robust pricing-hedging duality with penalty terms given via
Wasserstein distance (see Section \ref{secwasser}).

\item A new dual robust representation for the Optimized Certainty
Equivalent functional (see Section \ref{GOCE}).
\end{enumerate}

As already mentioned before, we work with (possibly) noncompact $K_{0},\dots
,K_{T}$ (which corresponds to working with possibly unbounded price
processes), and we allow for traded options $\varphi _{t},t=0,\dots ,T$
which satisfy a sublinearity condition of the form $\abs{\varphi _{t}(x_{t})}\leq
\alpha (1+\left\vert x_{t}\right\vert )$ for some $\alpha \geq 0$. Our order
continuity-type assumptions (see\eqref{condconvergencecalls}) in the main
result, Theorem \ref{mainEMOTtheoremgeneral}, are essentially automatically
satisfied when the compactness assumption on $K_{0},\dots ,K_{T}$ is added
(see Corollary \ref{maincorcasecpt}).

\bigskip

One additional feature of the paper consists in replacing the set of
stochastic integrals $\mathcal{I}$ with a general set $\mathcal{A}$  of suitable hedging
instruments, that will be a general convex cone.
Particular choices of such a set $\mathcal{A}$, apart from the usual set of
stochastic integrals, allow us to work with $\varepsilon $-martingale
measures, supermartingales and submartingales in the duality (see Subsection \ref{Examplesnoncpt}). This extends
EMOT beyond the strict martingale property in the strict sense. 

Finally, Section \ref{Secstab} is devoted to \textit{stability and
convergence} issues, as we analyze how the duality is affected by variations in
the penalty terms.  In the Examples \ref{examplepenalmkttomot}, \ref{exampleemotwasserdistance}
 and \ref{exampleutilitiestoMOT} we apply this result to the convergence of EMOT to the
\textquotedblleft extreme\textquotedblright case of MOT,  and in Subsection \ref{secwasser} we
focus on Wasserstein-induced penalization terms. 
\bigskip

We summarize the preceding discussion in the following Table and we point
out that in \textit{this paper we develop the duality theory sketched in the
last line of the Table and provide its financial interpretation}.
Differently from rows $1$, $2$, $5$, $6$, in rows $3$, $4$, $7$, $8$, the
financial market is present and martingale measures are involved in the dual
formulation. In rows $1$, $2$, $3$, $4$ we illustrate the classical setting,
where the conditions in the functional form hold $P$-a.s., while in the last
four rows Optimal Transport is applied to treat the robust versions, where
the inequalities holds for all elements of $\Omega $.

\newpage \captionsetup[table]{labelfont=bf, font=footnotesize} %
\renewcommand{\arraystretch}{1.5}

{\scriptsize \centering
\begin{longtable}{c|c|c|c|c|}

\caption{\label{tab1} $\Pi$($\Omega$) is the set of all
probabilities on $\Omega$; $\mathcal{P}(P)=\left\{Q \in \Pi(\Omega) \mid
Q \ll P\right\}$;  $\mathrm{Mart}(\Omega)$ is the set of all martingale probabilities on
$\Omega$; $\mathcal{M}(P)=\mathrm{Mart}(\Omega)\cap   \mathcal{P}(P) $;
$\Pi(\probq_1,\probq_2)=\left \{Q \in \Pi (\Omega) \text{  with given
marginals}\right\}$; $\mathrm{Mart}(\probq_1,\probq_2)=\left \{Q \in
\mathrm{Mart}(\Omega) \text{ with given marginals}\right\}$; $\mathrm{Meas}(\Omega)$ is the set of all
positive finite measures on $ \Omega$; $\text{Sub}(c)$ is the set of static parts of semistatic subhedging strategies for $c$; $U$ is a concave proper utility functional and $S^U$ is the associated generalized Optimized Certainty Equivalent.}; 
\\
\multicolumn{2}{l|}{ } & FUNCTIONAL FORM  &    SUBLINEAR    &    CONVEX   \\
\hline
\endhead
\hline
\endfoot

1 &- Coherent R.M.      & $-\inf\{m\mid c+m\in\mathcal{A}\},\,\mathcal{A}\text{ cone}$    &  $\inf\limits_{\probq\in\mathcal{Q}\subseteq\mathcal{P}(\probp)}\Eq{c}$ &  \\
\hdashline
         
2 &- Convex R.M.         & $-\inf\{m\mid c+m\in\mathcal{A}\},\,\mathcal{A}\text{ convex}$ &  & $\inf\limits_{\probq\in\mathcal{P}(\probp)}(\Eq{c}+\alpha_{\mathcal{A}}(\probq))$       \\

\hline

3 & Subreplic. price     & $\sup \left\{ m \mid \exists \Delta : m+I^{\Delta }(X)\leq c\right\} $   &  $\inf\limits_{%
Q\in \mathcal{M}(P)}\Eq{c}$ &  \\
\hdashline
4 & Indiff. price       & $\sup \left\{ m\mid U(c-m)\geq U(0)\right\}$ &  & $\inf\limits_{%
Q\in \mathcal{M}(P)}(\Eq{c}+\alpha_{U}(\probq))$       \\

\hline

5 & O.T.   &   $\sup\limits_{\varphi+\psi\leq c}\left({E}_{\probq_1}[\varphi]+{E}_{\probq_2}[\psi]\right)$   & $\inf\limits_{\probq\in \Pi(\probq_1,\probq_2)}\Eq{c}$   &   \\
\hdashline
6 & E.O.T. &   $\sup\limits_{\varphi + \psi\leq c} U(\varphi,\psi)$   &  & $\inf\limits_{\probq\in \mathrm{Meas} (\Omega)  }(\Eq{c}+\mathcal{D}_{U}(\probq))$    \\
\hline
7 & M.O.T.   & $\sup\limits_{[\varphi,\psi]\in\text{Sub}(c)}\left({E}_{\probq_1}[\varphi]+{E}_{\probq_2}[\psi]\right)$      & $\inf\limits_{\probq\in \mathrm{Mart}(\probq_1,\probq_2)}\Eq{c}$   &    \\
\hdashline
8 & E.M.O.T. & $\sup\limits_{[\varphi,\psi]\in\text{Sub}(c)}S^U(\varphi,\psi)$      &   & $\inf\limits_{\probq\in \mathrm{Mart} (\Omega) }(\Eq{c}+\mathcal{D}_{U}(\probq))$    \\

\end{longtable} 
} 

\section{The Entropy Martingale Optimal Transport Duality\label{EMOT new}}

In this section we present the detailed mathematical setting, the main
results and their proofs. The main Theorem \ref{mainEMOTtheoremgeneral}
relies on: (i) a Fenchel-Moreau argument applied to the dual system $%
(C_{0:T},(C_{0:T})^{\ast }),$ where $C_{0:T}$ is a set of appropriately
weighted continuous functions; (ii) Daniell-Stone Theorem that guarantees
that the elements, in the dual space $(C_{0:T})^{\ast },$ that enter in the
dual representation can be represented by probability measures. In order to
make this possible, an order  continuity type assumption on the valuation
functional is enforced (see \eqref{condconvergencecalls}).

\subsection{Setting}
\label{sectionsetting}
For a metric space $\mathbb{X}$, $\mathcal{B}(\mathbb{X})$ denotes the Borel 
$\sigma $-algebra and $m\mathcal{B}(\mathbb{X})$ denotes the class of
real-valued, Borel-measurable functions on $\mathbb{X}$. We define the
following sets: 
\begin{align*}
\mathrm{ca}(\mathbb{X})& :=\left\{ \gamma :\mathcal{B}(\mathbb{X}%
)\rightarrow (-\infty ,+\infty )\mid \gamma \text{ is finite signed Borel
measure on }\mathbb{X}\right\} , \\
\mathrm{Meas}(\mathbb{X})& :=\{\mu :\mathcal{B}(\mathbb{X})\rightarrow
\lbrack 0,+\infty )\mid \mu \text{ is a non negative finite Borel measure on 
}\mathbb{X}\}, \\
\mathrm{Prob}(\mathbb{X})& :=\{Q:\mathcal{B}(\mathbb{X})\rightarrow \lbrack
0,1]\mid Q\text{ is a probability Borel measure on }\mathbb{X}\}, \\
\mathcal{C}(\mathbb{X})& :=\{\varphi :\mathbb{X}\rightarrow {\mathbb{R}}\mid
\varphi \text{ is continuous on }\mathbb{X}\}, \\
\mathcal{C}_{b}(\mathbb{X})& :=\{\varphi :\mathbb{X}\rightarrow {\mathbb{R}}%
\mid \varphi \text{ is bounded and continuous on }\mathbb{X}\}.
\end{align*}%
We now introduce the following families of continuous functions. For a $\psi
\in \mathcal{C}(\mathbb{X})$ we set 
\begin{equation*}
C_{\psi }:=\left\{ \phi \in \mathcal{C}(\mathbb{X})\mid \left\Vert \phi
\right\Vert _{\psi }:=\sup_{x\in \mathbb{X}}\frac{\left\vert \phi
(x)\right\vert }{1+\left\vert \psi (x)\right\vert }<+\infty \right\} \,.
\end{equation*}%
As it can be easily verified just following the classical case of bounded
continuous functions with sup-norm, $C_{\psi }$ is a Banach lattice under
the norm $\left\Vert \cdot \right\Vert _{\Psi }$.

Notice also that $\mathcal{C}_{b}(\mathbb{X})\ni \varphi \mapsto
(1+\left\vert \psi \right\vert )\varphi \in C_{\psi }$ defines an
isomorphism between Banach spaces. Hence if $\mathbb{X}$ is separable, so is 
$C_{\psi }$. 
The topological dual of $C_{\psi }$ will be denoted by $(C_{\psi })^{\ast }$.
\\
Fix now $d\in \mathbb{N},d\geq 1$, modeling the number of
stocks in the market, and fix $d(T+1)$ closed subsets of ${\mathbb{R}}$ : $%
K_{0}^{1},\dots ,K_{0}^{d},\dots ,K_{T}^{1},\dots ,K_{T}^{d}.$ 
For $0 \leq s \leq t \leq T$ we take
\begin{equation*}
\Omega_{s:t}: =\bigtimes_{u=s}^{t}\bigtimes_{j=1}^{d}K_{u}^{j}\,\,\,\,\,\, \text{and} \,\,\,\,\,\,\Omega :=\Omega_{0:T}=\bigtimes_{t=0}^{T}\bigtimes_{j=1}^{d}K_{t}^{j}.
\end{equation*}%

We will consider the following weighted spaces of continuous functions \begin{equation*}
C_{s:t}:=
\left\{ \phi \in \mathcal{C}\left(\Omega_{s:t}\right) \mid \left\Vert
\phi \right\Vert _{s:t}:=\sup_{x\in \bigtimes_{u=s}^t \bigtimes_{j=1}^{d}K_{u}^{j} }\frac{\left\vert \phi
(x)\right\vert }{1+\sum_{u=s}^{t}\sum_{j=1}^{d}\left\vert x_{u}^{j}\right\vert }<+\infty \right\}\,,
\end{equation*}%
\begin{equation*}
C_{t}=C_{t:t}:=\left\{ \phi \in \mathcal{C}\left( \bigtimes_{j=1}^{d}K_{t}^{j}\right) \mid \left\Vert
\phi \right\Vert _{t}:=\sup_{x\in \bigtimes_{j=1}^{d}K_{t}^{j} }\frac{\left\vert \phi
(x)\right\vert }{1+\sum_{j=1}^{d}\left\vert x_{t}^{j}\right\vert }<+\infty \right\}\,.
\end{equation*}%
In general, for an index set $I\subseteq \{1,\dots ,d\}\times \{0,\dots ,T\}$
we introduce 
\begin{equation*}
C_{I}:=
\left\{ \phi \in \mathcal{C}\left( \bigtimes_{(j,t)\in
I}K_{t}^{j}\right) \mid \left\Vert
\phi \right\Vert _{0:T}:=\sup_{x\in \bigtimes_{(j,t)\in
I}K_{t}^{j} }\frac{\left\vert \phi
(x)\right\vert }{1+\sum_{(j,t)\in I}\left\vert x_{t}^{j}\right\vert }<+\infty \right\}\,.
\end{equation*}%
The corresponding norms will be denoted by $\left\Vert \cdot \right\Vert
_{s:t},\left\Vert \cdot \right\Vert _{t},\left\Vert \cdot \right\Vert _{I}$
respectively. 
\begin{remark}
Notice that if $K_{0}^{1},\dots ,K_{0}^{d},\dots ,K_{T}^{1},\dots ,K_{T}^{d}$
are compact sets then 
\begin{equation*}
C_{0:T}=\mathcal{C}_{b}(\Omega )\text{ and }(C_{0:T})^{\ast }=\mathrm{ca}(\Omega ).
\end{equation*}
\end{remark}

Analogously, we introduce the spaces $B_{\psi }$ and $B_{I}$ in a similar
fashion, just substituting the condition $\varphi \in \mathcal{C}(\mathbb{X}%
) $ with the condition $\varphi \in m\mathcal{B}(\mathbb{X})$ in the
definitions. In particular 
\begin{equation*}
B_{0:T}:=\left\{ \phi \in m\mathcal{B}\left( \Omega\right) \mid \left\Vert
\phi \right\Vert _{0:T}:=\sup_{x\in \Omega }\frac{\left\vert \phi
(x)\right\vert }{1+\sum_{t=0}^{T}\sum_{j=1}^{d}\left\vert
x_{t}^{j}\right\vert }<+\infty \right\}
\end{equation*}%
turns out to be a Banach lattice under the norm $\left\Vert \cdot
\right\Vert _{0:T}$. Observe that by slight abuse of notation (regarding the
domains of the functions) for index sets $I\subseteq J\subseteq \{1,\dots
,d\}\times \{0,\dots ,T\}$ we have a constant $0<\theta \leq 1$ such that 
\begin{equation}
C_{I}\subseteq C_{J},\,\,\,\theta \left\Vert \phi \right\Vert _{I}\leq
\left\Vert \phi \right\Vert _{J}\leq \left\Vert \phi \right\Vert
_{I}\,\forall \phi \in C_{I}\,.  \label{consistencyrestrictions}
\end{equation}%
Such a constant $\theta $ is given by 
\begin{equation*}
\theta =\frac{1+\alpha }{1+\alpha +\sum_{(j,t)\in J\setminus
I}\min_{x_{t}^{j}\in K_{t}^{j}}\left\vert x_{t}^{j}\right\vert }
\end{equation*}%
where $\alpha =\min \left\{ \sum_{(j,t)\in I}\left\vert x_{t}^{j}\right\vert
\mid x_{t}^{j}\in K_{t}^{j}\text{ }\forall \,(j,t)\in I\right\} $.
\\

As already mentioned in \cite{PennanenPerkkio19} and \cite%
{CheriditoKupperTangpi17}, every finite signed measure $\gamma $ on the
Borel $\sigma $-algebra $\mathcal{B}(\mathbb{X})$ such that $%
C_{0:T}\subseteq L^{1}(\mathbb{X},\mathcal{B}(\mathbb{X}),\left\vert \gamma
\right\vert )$ induces a continuous linear functional $\lambda \in
(C_{0:T})^{\ast }$ via integration: 
\begin{equation*}
c\mapsto \langle c,\lambda \rangle =\int_{\mathbb{X}}c\,\mathrm{d}\gamma
\,,\,\,\,\forall \,c\in \,C_{0:T}\,.
\end{equation*}%
The collection of such functionals, identified with the corresponding
measures, will be denoted by $\mathrm{ca}^{1}(\mathbb{X})$, that is 
\begin{equation*}
\mathrm{ca}^{1}(\mathbb{X})=\{\gamma \mid \gamma \text{ is a finite signed
measure on }\mathcal{B}(\mathbb{X})\text{ with }C_{0:T}\subseteq L^{1}(%
\mathbb{X},\mathcal{B}(\mathbb{X}),\left\vert \gamma \right\vert )\},
\end{equation*}%
while the classes of non negative measures and probability measures in $%
\mathrm{ca}^{1}(\mathbb{X})$ will be denoted by 
\begin{equation*}
\mathrm{Meas}^{1}(\mathbb{X}) \,\,\,\,\,\,\,\text{and} \,\,\,\,\,\, \mathrm{Prob}^{1}(\mathbb{X}).
\end{equation*}

In a discrete time framework with finite horizon $T$ and assuming zero interest rate, we model a market with $d$ stocks using the canonical $d$-dimensional
process given by $X_{t}^{j}(x)=x_{t}^{j},j=1,\dots ,d,t=0\dots ,T$.

Observe that every $\phi \in C_{0:T}$ satisfies: $\left\vert \phi
(x)\right\vert \leq \left\Vert \phi \right\Vert _{0:T}\left(
1+\sum_{t=0}^{T}\sum_{j=1}^{d}\left\vert x_{t}^{j}\right\vert \right) $ and
so, for any measure $\mu \in \mathrm{Meas}^{1}(\mathbb{X}),$ we have: $%
C_{0:T}\subseteq L^{1}(\mathbb{X},\mathcal{B}(\mathbb{X}),\mu )$ iff $%
X_{t}^{j}\in L^{1}(\mathbb{X},\mathcal{B}(\mathbb{X}),\mu )$ for all $j$ and
all $t.$

Fix now vector subspaces $\mathcal{E}_{0},\dots ,\mathcal{E}_{T}$ with ${%
\mathbb{R}}\subseteq \mathcal{E}_{t}\subseteq C_{0:t},t=0,\dots ,T$. We set 
\begin{equation*}
\mathcal{E}=\mathcal{E}_{0}\times \dots \times \mathcal{E}_{T}\,.
\end{equation*}%
The space $\mathcal{E}$ represents the class of financial instruments that can be used for static hedging, which may, for example, contain plain vanilla options.

Let $U:\mathcal{E}\rightarrow \lbrack -\infty ,+\infty )$ be a proper,
concave functional, representing the evaluation functional of the hedging instruments in $\mathcal{E}$. Consider the proper, convex functional 
\begin{equation*}
V(\varphi ):=-U(-\varphi )\,,
\end{equation*}%
and set%
\begin{equation}
\mathrm{dom}(U):=\{\varphi \in \mathcal{E}\mid U(\varphi )>-\infty
\}\,,\,\,\,\,\,\,\,\,\mathrm{dom}(V)=\left\{ \varphi \in \mathcal{E}\mid
V(\varphi )<+\infty \right\} \,.  \notag
\end{equation}%
We define the (convex) conjugate $\mathcal{D}:\bigtimes_{t=0}^{T}(C_{0:t})^{\ast }\rightarrow
(-\infty ,+\infty ]$ of the functional $U$ by%
\begin{equation}
\mathcal{D}(\gamma _{0},\dots ,\gamma _{T}):=\sup_{\varphi \in \mathcal{E}%
}\left( U(\varphi )-\sum_{t=0}^{T}\langle \varphi _{t},\gamma _{t}\rangle
\right) =\sup_{\varphi \in \mathcal{E}}\left( \sum_{t=0}^{T}\langle \varphi
_{t},\gamma _{t}\rangle -V(\varphi )\right) ,\,\,\,\,\,\gamma \in \bigtimes%
_{t=0}^{T}(C_{0:t})^{\ast }.  \label{dualrepdivgenabstract}
\end{equation}%
$\mathcal{D}$ is a convex functional and is $\sigma (\bigtimes%
_{t=0}^{T}(C_{0:t})^{\ast },\mathcal{E})$- lower semicontinuous, even if we
do not require that $U$ is $\sigma (\mathcal{E},\bigtimes%
_{t=0}^{T}(C_{0:t})^{\ast })$-upper semicontinuous. 
When a $\gamma \in (C_{0:T})^{\ast }$ is given, we somehow improperly write $%
\mathcal{D}(\gamma )=\mathcal{D}(\gamma _{0},\dots ,\gamma _{T}),$ where $%
\gamma _{t}$ is the restriction of $\gamma $ to $C_{0:t}$. We also set 
\begin{equation*}
\mathrm{dom}(\mathcal{D})=\left\{ [\gamma _{0},\dots ,\gamma _{T}]\in %
\bigtimes_{t=0}^{T}(C_{0:t})^{\ast }\mid \mathcal{D}(\gamma _{0},\dots
,\gamma _{T})<+\infty \right\} \,.
\end{equation*}%
As an immediate consequence of the definitions, the Fenchel inequality
holds: if $[\varphi _{0},\dots ,\varphi _{T}]\in \mathcal{E}$ and $[\gamma
_{0},\dots ,\gamma _{T}]\in \bigtimes_{t=0}^{T}(C_{0:t})^{\ast }$ 
\begin{equation}
\sum_{t=0}^{T}\langle \varphi _{t},\gamma _{t}\rangle \leq \mathcal{D}%
(\gamma _{0},\dots ,\gamma _{T})+V(\varphi _{0},\dots ,\varphi _{T})\,.
\label{fenchelgeneral}
\end{equation}

\begin{remark}
\label{remfenchelmoreau}Another way to introduce our setting, that will be
used in Subsections \ref{CasePenalization1gen} and \ref{secwasser}, is to start initially with a
proper convex functional $\mathcal{D}:\mathrm{ca}^{1}(\Omega )\rightarrow
(-\infty ,+\infty ]$ which is $\sigma (\mathrm{ca}^{1}(\Omega ),\mathcal{E})$%
-lower semicontinuous for an $\mathcal{E}=\mathcal{E}_{0}\times \dots \times 
\mathcal{E}_{T}\subseteq (C_{0:T})^{T+1}$. By the Fenchel-Moreau Theorem we
then have the representation 
\begin{equation*}
\mathcal{D}(\gamma )=\sup_{\varphi \in \mathcal{E}}\left(
\sum_{t=0}^{T}\int_{\Omega }\varphi _{t}\,\mathrm{d}\gamma -V(\varphi
)\right) ,
\end{equation*}%
where now $V$ is the Fenchel-Moreau (convex) conjugate of $\mathcal{D}$,
namely 
\begin{equation}  \label{fenchelmoreau}
V(\varphi ):=\sup_{\gamma \in \mathrm{ca}^{1}(\Omega )}\left(
\sum_{t=0}^{T}\int_{\Omega }\varphi _{t}\,\mathrm{d}\,\gamma -\mathcal{D}%
(\gamma )\right) .
\end{equation}%
Setting 
\begin{equation}  \label{eqdefU}
U(\varphi ):=-V(-\varphi )\text{,}\,\,\,\,\,\,\,\,\varphi \in \mathcal{E}%
\text{,}
\end{equation}%
we get back that $\mathcal{D}$ satisfies \eqref{dualrepdivgenabstract} and
additionally that $U$ is $\sigma (\mathcal{E},\mathrm{ca}^{1}(\Omega ))$%
-upper semicontinuous. In conclusion, a pair $(U,\mathcal{D})$ satisfying %
\eqref{dualrepdivgenabstract} might be defined either providing a proper
concave $U:\mathcal{E}\rightarrow \lbrack -\infty ,+\infty )$, as described
at the beginning of this section, or assigning a proper convex and $\sigma (%
\mathcal{E},\mathrm{ca}^{1}(\Omega ))$-lower semicontinuous functional $%
\mathcal{D}:\mathrm{ca}^{1}(\Omega )\rightarrow (-\infty ,+\infty ]$ as
explained in this Remark.
\end{remark}

For a given proper and concave functional $U:\mathcal{E}\rightarrow \lbrack
-\infty ,+\infty )$ we define, as in (\ref{defSU}), the functional $S^{U}:%
\mathcal{E}\rightarrow \lbrack -\infty ,+\infty ]$ by%
\begin{equation*}
S^{U}(\varphi ):=\sup_{\beta \in {\mathbb{R}}^{T+1}}\left( U(\varphi +\beta
)-\sum_{t=0}^{T}\beta _{t}\right) ,\text{\quad with\quad }\mathrm{dom}%
(S^{U}):=\{\varphi \in \mathcal{E}\mid S^{U}(\varphi )>-\infty \},
\end{equation*}%
whose properties are collected in Lemma \ref{remalmostconcave}.

\begin{definition}
Given a convex cone $\mathcal{A}\subseteq C_{0:T}$ and a measurable $c\in m%
\mathcal{B}(\Omega )$ we define 
\begin{equation}
\mathfrak{P}(c):=\sup_{z\in -\mathcal{A}}\sup_{\varphi \in \mathbf{\Phi }%
_{z}(c)}S^{U}\left( \varphi \right) \in \lbrack -\infty ,+\infty ]
\label{deffrakPgeneral}
\end{equation}%
where%
\begin{equation*}
\mathbf{\Phi }_{z}(c):=\left\{ \varphi \in \mathrm{dom}(S^{U})\mid
\sum_{t=0}^{T}\varphi _{t}(x_0,\dots,x_{t})+z(x)\leq c(x)\,\,\,\forall \,x\in \Omega
\right\} ,
\end{equation*}%
and the usual convention $\sup \emptyset =-\infty $ is adopted.
\end{definition}

From (\ref{dualitywithSt}) we recognize that the problem $\mathfrak{P}(c)$
in \eqref{deffrakPgeneral} is a generalized robust subhedging problem for $c$%
, with a general set, namely $-\mathcal{A}$, replacing the set of terminal
values of stochastic integrals used before. Some relevant examples for
choices of $\mathcal{A}$ are provided in Section \ref{Examplesnoncpt}. The
use of $-\mathcal{A}$ in place of $\mathcal{A}$ is somehow a matter of
taste, as explained in Remark \ref{whyminus}. Observe also that in this
notation the \textit{superhedging} problem (with respect to $\mathcal{A}$)
for $c$ is

\begin{equation*}
\mathfrak{S}(c):=\inf_{z\in \mathcal{A}}\inf_{\varphi \in \mathbf{\Psi }%
_{z}(c)}S_{V}\left( \varphi \right) \in \lbrack -\infty ,+\infty ]
\end{equation*}%
where%
\begin{equation*}
S_{V}(\varphi ):=-S^{U}(-\varphi ),\quad \mathrm{dom}(S_{V}):=\{\varphi \in 
\mathcal{E}\mid S_{V}(\varphi )<+\infty \}=-\mathrm{dom}(S^{U})
\end{equation*}%
\begin{equation*}
\mathbf{\Psi }_{z}(c):=\left\{ \varphi \in \mathrm{dom}(S_{V}),\text{ }%
\sum_{t=0}^{T}\varphi _{t}(x_0,\dots,x_{t})+z(x)\geq c(x)\,\,\,\forall \,x\in \Omega
\right\} .
\end{equation*}

Using $-\mathcal{A}$ for $\mathfrak{P}$ and $\mathcal{A}$ for $\mathfrak{S}$
we then get that the two are linked by $\mathfrak{S}(c)=-\mathfrak{P}(-c)$,
and the duality results for $\mathfrak{P}$ can then easily be translated in
duality results for $\mathfrak{S}$. Of course, when $\mathcal{A}$ is a
vector space as in the case of stochastic integrals (see \eqref{StoIntegrad}
and Example \ref{exampleamtg}) we have $\mathcal{A}=-\mathcal{A}$ and there
is no need for distinguishing the two possibilities.

We define the polar $\mathcal{A}^{\circ }$ of the cone $\mathcal{A}$ to be
the set 
\begin{equation*}
\mathcal{A}^{\circ }:=\{\lambda \in (C_{0:T})^{\ast }\mid \langle z,\lambda
\rangle \leq 0\,\,\forall z\in \mathcal{A}\}\,
\end{equation*}%
and we observe that for any $\lambda \in (C_{0:T})^{\ast }$ 
\begin{equation*}
\sigma _{\mathcal{A}}(\lambda ):=\sup_{z\in \mathcal{A}}\,\langle z,\lambda
\rangle =%
\begin{cases}
0 & \lambda \in \mathcal{A}^{\circ } \\ 
+\infty & \text{otherwise}%
\end{cases}%
\,.
\end{equation*}

\subsection{Examples}

\label{Examplesnoncpt}

These examples deal with possible choices of the convex cone $\mathcal{A}$ and its polar $\mathcal{A}^\circ$ (observe that the set $\mathrm{Prob}^{1}(\Omega )\cap \mathcal{A}^{\circ
} $ appears in the main duality (\ref{dualrepreabstract})).

\begin{example}
{(Martingale measures)} \label{exampleamtg} To introduce martingale measures
in this setup we set 
\begin{align}
\mathcal{H}^{d}& :=\left\{ \Delta =[\Delta _{0},\dots ,\Delta _{T-1}]\mid
\Delta _{t}\in (\mathcal{C}_{b}(K_{0}\times \dots \times K_{t}))^{d}\right\}\,,
\notag \\
I^{\Delta }(x)& :=\sum_{t=0}^{T-1}\sum_{j=1}^{d}\Delta _{t}^{j}(x_{0},\dots
,x_{t})(x_{t+1}^{j}-x_{t}^{j})\,\,\,\forall \,x\in \Omega \,, \notag \\
\mathcal{A}& =\mathcal{I}:=\left\{ I^{\Delta }\mid \Delta \in \mathcal{H}%
^{d}\right\} \subseteq C_{0:T}\,.  \label{StoIntegrad}
\end{align}%
Thus the space $\mathcal{H}^{d}$ is the class of admissible trading
strategies and $\mathcal{I}$ is the set of elementary stochastic integral.
The (possibly empty) class of martingale measures for the canonical process
is denoted by $\mathrm{Mart}(\Omega )$, and consists of all the probability
measures on $\mathcal{B}(\Omega )$ which make each of the processes $%
(X_{t}^{j})_{t}$ a martingale under the natural filtration $\mathcal{F}%
_{t}:=\sigma (X_{s}^{j},s\leq t,j=1,...,d),t=0,\dots ,T$. Equivalently, 
\begin{equation*}
\mathrm{Mart}(\Omega ):=\{Q\in \mathrm{Prob}^{1}(\Omega )\mid {E}%
_{Q}[I^{\Delta }]=0,\,\forall \Delta \in \mathcal{H}^{d}\}=\mathrm{Prob}%
^{1}(\Omega )\cap \mathcal{A}^{\circ }.
\end{equation*}%
%
%
%
%
%
%
\end{example}

\begin{example}
($\varepsilon $-martingale measures) \label{ExamplesA} For every $%
\varepsilon \geq 0$, the set of $\varepsilon $-martingale measures (see \cite%
{GuoObloj19}) is 
\begin{equation*}
\mathrm{Mart}_{\varepsilon }(\Omega ):=\left\{ Q\in \mathrm{Prob}^{1}(\Omega
)\mid E_{Q}\left[ I^{\Delta }\right] \leq \varepsilon
\sum_{t=0}^{T-1}\max_{j=1,\dots,d}\left\Vert \Delta _{t}^{j}\right\Vert _{\infty
},\,\forall \Delta \in \mathcal{H}^{d}\right\} \,.
\end{equation*}%
Thus, taking 
\begin{equation}
\mathcal{A}^{\varepsilon }:=\mathrm{convex}\left( \left\{ I^{\Delta
}-\varepsilon \sum_{t=0}^{T-1}\max_{j=1,\dots,d}\left\Vert \Delta
_{t}^{j}\right\Vert _{\infty },\Delta \in \mathcal{H}^{d}\right\} \right)
\label{polarepsilonmarts}
\end{equation}%
(here $\mathrm{convex}(\cdot )$ stands for the convex hull in $C_{0:T}$,
which is easily seen to be a cone since $\mathcal{H}^{d}$ is a vector
space), one sees that 
\begin{equation*}
\mathrm{Mart}_{\varepsilon }(\Omega )=\mathrm{Prob}^{1}(\Omega )\cap (%
\mathcal{A}^{\varepsilon })^{\circ }.
\end{equation*}%
Alternatively, we observe that $\mathrm{Mart}_{\varepsilon }(\Omega )=%
\mathrm{M}_{\varepsilon }^{1}(\Omega )\cap \mathrm{Prob}^{1}(\Omega )$ where 
\begin{equation*}
\mathrm{M}_{\varepsilon }^{1}(\Omega ):=\left\{ \lambda \in (C_{0:T})^{\ast
}\mid \lambda \geq 0,\,\langle I^{\Delta },\lambda \rangle \leq \varepsilon
\sum_{t=0}^{T-1}\max_{j=1,\dots,d}\left\Vert \Delta _{t}^{j}\right\Vert _{\infty
}\lambda (1),\,\forall \Delta \in \mathcal{H}^{d}\right\}
\end{equation*}%
is a $\sigma ((C_{0:T})^{\ast },C_{0:T})$ (i.e. weak$^{\ast }$) closed
convex cone by direct computation. Suppose $\mathrm{M}_{\varepsilon
}^{1}(\Omega )\neq \emptyset $ and set $\mathcal{A}=\left( \mathrm{M}%
_{\varepsilon }^{1}(\Omega )\right) ^{\circ }\subseteq C_{0:T}.$\ Then by
the Bipolar Theorem $\mathcal{A}$ is a closed convex cone with $\mathcal{A}%
^{\circ }= \mathrm{M}_{\varepsilon }^{1}(\Omega ) $.
Consequently, for a $Q\in \mathrm{Prob}^{1}(\Omega )$ we have: either $%
\sigma _{\mathcal{A}}(Q)=0$, which happens if and only if $Q\in \mathrm{Mart}%
_{\varepsilon }(\Omega )$, or $\sigma _{\mathcal{A}}(Q)=+\infty $ otherwise.
From this it follows that $\mathcal{A}$ is the weak closure of $\mathcal{A}%
^{\varepsilon }$, since they are both convex cones with the same polar set.
Taking in particular $\varepsilon =0$ we have $\mathrm{Mart}_{0}(\Omega )=%
\mathrm{Mart}(\Omega )$ as in Example \ref{exampleamtg}.

It is interesting to notice that for any sequence $\varepsilon
_{n}\downarrow _{n}0$ we have 
\begin{equation}
\sigma _{\mathcal{A}^{\varepsilon _{n}}}(Q)\uparrow _{n}\sigma _{\mathcal{I}%
}(Q)\,\,\,\,\,\forall \,Q\in \mathrm{Prob}^{1}(\Omega )\,.
\label{convergenceepsmarttomart}
\end{equation}
\end{example}

\begin{example}
{(Super/submartingale measures)} Alternative choices for the set $\mathcal{A}
$ are $\mathcal{A}^{\pm }=\{I^{\Delta }\mid \Delta \in (\mathcal{H}^{\pm
})^{d}\}$ where $\mathcal{H}^{+}=\{\Delta \in \mathcal{H}\mid \Delta
_{t}\geq 0$ $\forall t=0,\dots ,T\}$ and $\mathcal{H}^{-}=-\mathcal{H}^{+}$.
 $\mathcal{A}^{+}$ models dynamic trading with no short selling and yields 
\begin{equation*}
\{\text{supermartingale measures for the canonical process}\}=\mathrm{Prob}%
^{1}(\Omega )\cap (\mathcal{A}^+)^{\circ }\,.
\end{equation*}
\end{example}

\begin{example}
\label{ExEOT}Let $(C_{0:T})_{+}:=\{f\in C_{0:T}\mid f\geq 0\}.$ For any set $%
\mathcal{A}$ such that $\{0\}\subseteq \mathcal{A}\subseteq -(C_{0:T})_{+}$
we obtain $\mathrm{Prob}^{1}(\Omega )=\mathrm{Prob}^{1}(\Omega )\cap 
\mathcal{A}^{\circ }$. As explained in Section \ref{SecEOT}, this choice
will then determine the Entropy Optimal Transport duality with no dynamic
hedging.
\end{example}

\subsection{The main results \label{EMOT}}

Before providing all the mathematical details, we wish to identify a
suitable candidate for the objective functional for our (primal) problem. To
do so, we start from the EMOT (dual) problem and proceed with an heuristic
argument. In order to avoid all integrability issues and keep notation as
simple as possible, in this introductory discussion we suppose that $d=1$,
that $K_{t}^{1}:=K_{t},$ $t=0,\dots ,T$ are all compact sets and that $%
\mathcal{A}=\mathcal{I}$, as in Example \ref{exampleamtg}.

\begin{align}
& \inf_{Q\in \mathrm{Mart}(\Omega )}\left( E_{Q}\left[ c\right] +\mathcal{D}%
(Q)\right)  \label{eqinitial0} \\
& =\inf_{Q\in \mathrm{Prob}(\Omega )}\sup_{\Delta \in \mathcal{H}}\left(
E_{Q}\left[ c-\sum_{t=0}^{T-1}\Delta _{t}(X_{0},\dots ,X_{t})(X_{t+1}-X_{t})%
\right] +\mathcal{D}(Q)\right)  \label{eqinitial2} \\
& =\inf_{\mu \in \mathrm{Meas}(\Omega )}\sup_{\substack{ \Delta \in \mathcal{%
H}  \\ \beta \in {\mathbb{R}}^{T+1}}}\left( \int_{\Omega }\left[ c-I^{\Delta
}+\sum_{t=0}^{T}\beta _{t}\right] \mathrm{d}\mu -\sum_{t=0}^{T}\beta _{t}+%
\mathcal{D}(\mu )\right)  \label{eqinitial6} \\
& =\inf_{\mu \in \mathrm{Meas}(\Omega )}\sup_{\substack{ \Delta \in \mathcal{%
H}  \\ \beta \in {\mathbb{R}}^{T+1}}}\left( \int_{\Omega }\left[ c-I^{\Delta
}+\sum_{t=0}^{T}\beta _{t}\right] \mathrm{d}\mu -\sum_{t=0}^{T}\beta
_{t}+\sup_{\varphi \in \mathcal{E}}\left( U(\varphi )-\int_{\Omega }\left(
\sum_{t=0}^{T}\varphi _{t}\right) \mathrm{d}\mu \right) \right)
\label{eqinitial7} \\
& =\sup_{\substack{ \Delta \in \mathcal{H},\varphi \in \mathcal{E},  \\ \beta
\in {\mathbb{R}}^{T+1}}}\left( \inf_{\mu \in \mathrm{Meas}(\Omega
)}\int_{\Omega }\left[ c-I^{\Delta }+\sum_{t=0}^{T}\beta
_{t}-\sum_{t=0}^{T}\varphi _{t}\right] \mathrm{d}\mu -\sum_{t=0}^{T}\beta
_{t}+U(\varphi )\right)  \label{eqinitial8} \\
& =\sup_{\substack{ \Delta \in \mathcal{H},  \\ \beta \in {\mathbb{R}}^{T+1}}}%
\sup \left\{ U(\varphi )-\sum_{t=0}^{T}\beta _{t}\mid \varphi \in \mathcal{E}%
,c-I^{\Delta }+\sum_{t=0}^{T}\beta _{t}-\sum_{t=0}^{T}\varphi _{t}\geq
0\right\}  \label{eqinitial9} \\
& =\sup_{\substack{ \Delta \in \mathcal{H},  \\ \beta \in {\mathbb{R}}^{T+1}}}%
\sup \left\{ U(\varphi +\beta )-\sum_{t=0}^{T}\beta _{t}\mid \varphi \in 
\mathcal{E},\sum_{t=0}^{T}\varphi _{t}+I^{\Delta }\leq c\right\}
\label{eqinitial10} \\
& =\sup_{\Delta \in \mathcal{H}}\sup \left\{ \sup_{\beta \in {\mathbb{R}}%
^{T+1}}\left( U(\varphi +\beta )-\sum_{t=0}^{T}\beta _{t}\right) \mid \varphi
\in \mathcal{E},\sum_{t=0}^{T}\varphi _{t}+I^{\Delta }\leq c\right\} \,.
\label{eqinitial11}
\end{align}

The equality chain above is justified as follows: as $K_{t}$ is compact, $%
X_{t}\in L^{1}(K_{t},\mathbb{B}(K_{t}),\mu )$ for all $\mu \in \mathrm{Meas}%
(\Omega )$ and then \eqref{eqinitial0}=\eqref{eqinitial2} follows using the
same argument as in \cite{Beig} Lemma 2.3, which yields that the inner
supremum explodes to $+\infty $ unless $Q$ is a martingale measure on $%
\Omega $ ; \eqref{eqinitial2}=\eqref{eqinitial6} follows observing that the
inner supremum over $\beta \in {\mathbb{R}}^{T+1}$ explodes to $+\infty $
unless $\mu (\Omega )=1$; \eqref{eqinitial6}=\eqref{eqinitial7} exploits %
\eqref{dualrepdivgenabstract}; in \eqref{eqinitial7}=\eqref{eqinitial8} we
proceed heuristically interchanging supremum and infimum; \eqref{eqinitial8}=%
\eqref{eqinitial9} is motivated observing that the infimum in %
\eqref{eqinitial8} equals $-\infty $ unless the inequality $c-I^{\Delta
}+\sum_{t=0}^{T}\beta _{t}-\sum_{t=0}^{T}\varphi _{t}\geq 0$ holds on $\Omega $%
; \eqref{eqinitial9}=\eqref{eqinitial10} is a simple rewriting, and so is %
\eqref{eqinitial10}=\eqref{eqinitial11}. To conclude, we expect our duality
to look as follows: 
\begin{equation*}
\inf_{Q\in \mathrm{Mart}(\Omega )}\left( E_{Q}\left[ c\right] +\mathcal{D}%
(Q)\right) =\sup_{\Delta \in \mathcal{H}}\sup \left\{ S^{U}(\varphi )\mid
\varphi \in \mathcal{E},\sum_{t=0}^{T}\varphi _{t}+I^{\Delta }\leq c\right\}
\,.
\end{equation*}%
A rigorous proof of this duality, based on the above argument and on a
minimax type Theorem that justify \eqref{eqinitial7}=\eqref{eqinitial8}, can
be found in a previous version of this paper (see A. Doldi and M. Frittelli,
Entropy Martingale Optimal Transport and Nonlinear Pricing-Hedging Duality,
Preprint: arXiv:2005.12572v1, 2020, where we also assumed that all the sets $%
K_{t}$ were compact). We now proceed providing a rigorous proof of the
duality above under the more general setup stated in Assumption \ref%
{asscallstozero}. In the proof we will rely on a Fenchel-Moreau type Theorem
for the functional $\mathfrak{P}$, rather than exploiting a minimax argument.

\begin{assumption}
\label{asscallstozero}$\,$

\begin{itemize}
\item[(i)] Let $K_{0}^{1},\dots ,K_{0}^{d},\dots ,K_{T}^{1},\dots ,K_{T}^{d}$
be closed subset of $\mathbb{R}$ and let $\Omega =\bigtimes_{t=0}^{T}%
\bigtimes_{j=1}^{d}K_{t}^{j}.$ The vector subspaces $\mathcal{E}_{0},\dots ,%
\mathcal{E}_{T}$ satisfy ${\mathbb{R}}\subseteq \mathcal{E}_{t}\subseteq
C_{0:t},$ $t=0,\dots ,T$ \ and we set $\mathcal{E}=\mathcal{E}_{0}\times
\dots \times \mathcal{E}_{T}\,.$ The functional $U:\mathcal{E}\rightarrow
\lbrack -\infty ,+\infty )$ is  concave with $U(0)\in {\mathbb{R}}$%
. Moreover, $\mathcal{A}\subseteq C_{0:T}$ is a convex cone with $0\in 
\mathcal{A}$.

\item[(ii)] For every $t=0,\dots ,T$ there exist a sequence of compact sets $%
\mathfrak{K}_{t}(n)\subseteq \bigtimes_{j=1}^d K^j_{t},n\geq 1$ and a
sequence of functions $0\leq f_{t}^{n}\in \mathcal{E}_{t},n\geq 1$ such
that: 
\begin{equation}
1+\sum_{t=0}^{T}\sum_{j=1}^{d}\left\vert x_{t}^{j}\right\vert \leq
\sum_{t=0}^{T}f_{t}^{n}(x_{0},\dots ,x_{t})\,\,\,\,\,\forall \lbrack
x_{0},\dots ,x_{T}]\in \Omega \setminus \mathfrak{K}_{0}(n)\times \dots
\times \mathfrak{K}_{T}(n)  \label{ineqoutsidecpt}
\end{equation}%
and 
\begin{equation}
V(\Gamma f_{0}^{n},\dots ,\Gamma f_{T}^{n})\rightarrow
_{n}0\,\,\,\,\,\,\,\forall \,\Gamma \in {\mathbb{R}},\Gamma >0\,.
\label{condconvergencecalls}
\end{equation}
\end{itemize}
\end{assumption}

Assumption \ref{asscallstozero}.(ii) is inspired by \cite%
{CheriditoKupperTangpi17} and admits a very simple interpretation, described
in Remark \ref{remcompactispart}.

\begin{example}
\label{examplecalls} Let 
\begin{equation}
\begin{split}
f^{\alpha }(x)& :=(x-\alpha )^{+}+(-\alpha -x)^{+}=(\left\vert x\right\vert
-\alpha )^{+},\,\,\,x\in {\mathbb{R}},\alpha \geq 0, \\
f_{j,t}^{\alpha }& :=f^{\alpha }|_{K_{t}^{j}},j=1,\dots ,d,t=0,\dots ,T
\end{split}
\label{defeqcalls}
\end{equation}%
and suppose that $f_{j,t}^{\alpha }\in \mathcal{E}_{t}$ for every $\alpha
\geq 0$, $j=1,\dots ,d$, $t=0,\dots ,T$. Set also $f_{t}^{n}:=%
\sum_{j=1}^{d}f_{j,t}^{\frac{n}{\beta }}$ for $\beta =\beta (d,T)$ given in
Proposition \ref{propcontrolcones}. In order to guarantee that (\ref%
{ineqoutsidecpt}) and (\ref{condconvergencecalls}) are satisfied, it is
enough to request that $V$ is (componentwise) nondecreasing on $\bigtimes%
_{t=0}^{T}\mathcal{E}_{t}$, $V(0)=0$ and 
\begin{equation}  \label{valuationoncomponents}
V_{t}(\alpha f_{j,t}^{\frac{n}{\beta }}):=V(0,\dots ,0,\alpha f_{j,t}^{\frac{%
n}{\beta }},0,\dots ,0)\rightarrow _{n}0
\end{equation}%
for every $j=1,\dots ,d$ and $t=0,\dots ,T$. $V_{t}$ is here represents a
valuation of the static hedging strategy consisting of a single option
having maturity $t$. To see that the former assumptions imply (\ref%
{ineqoutsidecpt}), one can select $\mathfrak{K}_{t}(n)=\bigtimes%
_{j=1}^dK^j_{t}\cap \lbrack -n,n]^{d}$ and observe that Proposition \ref%
{propcontrolcones} guarantees that \eqref{ineqoutsidecpt} holds. Moreover,
by convexity of $V$ we have%
\begin{align*}
0& =V(0)\leq V(\alpha f_{0}^{n},\dots ,\alpha f_{T}^{n})\leq \sum_{t=0}^{T}%
\frac{1}{T+1}V_{t}(\alpha (T+1)f_{t}^{n}) \\
& =\frac{1}{T+1}\sum_{t=0}^{T}V_{t}\left( \frac{1}{d}\sum_{j=1}^{d}d(T+1)%
\alpha f_{j,t}^{\frac{n}{\beta }}\right) \leq \frac{1}{d(t+1)}%
\sum_{t=0}^{T}\sum_{j=1}^{d}V_{t}\left( d(T+1)\alpha f_{j,t}^{\frac{n}{\beta 
}}\right) \rightarrow _{n}0\,,
\end{align*}%
that is (\ref{condconvergencecalls}).
\end{example}

\begin{remark}
\label{remcompactispart} In the particular case $K^1_{0},\dots
,K^d_0,\dots,K_T^1,\dots,K^d_{T}\subseteq \lbrack 0,+\infty )$, the
conditions $V_{t}(\alpha f_{j,t}^{\frac{n}{\beta }})\rightarrow _{n}0$ in %
\eqref{valuationoncomponents} means that the valuations over a suitable
sequence of call options on the underlying stocks converge to zero when the
corresponding strikes diverge to infinity.
\end{remark}

%

\begin{theorem}
\label{mainEMOTtheoremgeneral} Suppose Assumption \ref{asscallstozero} is
fulfilled.

\begin{itemize}
\item[(i)] {If }%
\begin{equation}
{\mathfrak{P}(\widehat{c})<+\infty }\text{ for some }{\widehat{c}\in B_{0:T},%
}  \label{BB}
\end{equation}%
{then $\mathfrak{P}(c)\in {\mathbb{R}}$ for every $c\in B_{0:T}$ and $%
\mathfrak{P}:B_{0:T}\rightarrow {\mathbb{R}}$ is norm continuous, cash
additive, concave and nondecreasing on $B_{0:T}$;}

\item[(ii)] {For every lower semicontinuous $c:\Omega \rightarrow (-\infty
,+\infty ]$ satisfying }%
\begin{equation}
c(x)\geq -A\left( 1+\sum_{t=0}^{T}\sum_{j=1}^{d}\left\vert
x_{t}^{j}\right\vert \right) \quad \forall x\in \Omega \text{,\quad {for some%
}}{\ A\in \lbrack 0,+\infty ){,}}  \label{controlfrombelow}
\end{equation}%
we have 
\begin{equation}
\mathfrak{P}(c):=\sup_{z\in -\mathcal{A}}\sup_{\varphi \in \mathbf{\Phi }%
_{z}(c)}S^{U}\left( \varphi \right)=\inf_{Q\in \mathrm{Prob}^{1}(\Omega )\cap\mathcal{A}%
^\circ}\left( E_{Q} \left[c\right]+\mathcal{D}(Q)\right)
\label{dualrepreabstract}
\end{equation}%
for $\Omega_{0:t}:=\bigtimes_{s=0}^t\bigtimes_{j=1}^d K^j_s$ and 
\begin{equation*}
\mathcal{D}(Q)=\sup_{\varphi \in \mathcal{E}}\left( U(\varphi
)-\sum_{t=0}^{T}\int_{\Omega_{0:t}}\varphi _{t}\mathrm{d}Q_{t}\right)
=\sup_{\varphi \in \mathcal{E}}\left( \sum_{t=0}^T\int_{\Omega_{0:t}}\varphi
_{t}\mathrm{d}Q_{t}-V(\varphi )\right)\,,
\end{equation*}%
where $Q_{t}$ is the marginal of $Q\in \mathrm{Prob}^{1}(\Omega )$ on $%
\mathcal{B}(\Omega_{0:t})$. Furthermore, if $\mathfrak{P}(c)<+\infty $ the
infimum in RHS of \eqref{dualrepreabstract} is a minimum.
\end{itemize}
\end{theorem}

\begin{remark}
\label{remsufficientfiniteabovep} Notice that the condition $\mathfrak{P}(%
\widehat{c})<+\infty $ for some $\widehat{c}\in B_{0:T}$ is not required for
the validity of Theorem \ref{mainEMOTtheoremgeneral} Item (ii). In addition, 
we allow in \eqref{dualrepreabstract} $\mathrm{Prob}^1(\Omega)\cap\mathcal{A}%
^{\circ }=\emptyset $ with the usual convention $\inf \emptyset =+\infty $.
We now provide conditions ensuring that $\mathfrak{P}(0)<+\infty $. By Item
(i) this will then imply{\ that $\mathfrak{P}(c)\in {\mathbb{R}}$ for every $%
c\in B_{0:T}.$}

\begin{itemize}
\item[(a)] If there exists a $\lambda \in \mathcal{A}^{\circ }\cap \partial
U(0)\subseteq (C_{0:T})^{\ast }$ then $\mathfrak{P}(0)<+\infty $ (here, $%
\partial U(0)\subseteq \bigtimes_{t=0}^{T}(C_{0:T})^{\ast }$ is the
supergradient of $U$ at $0\in \mathcal{E}$, and we are identifying $\lambda $
with the vector of its restrictions in writing improperly $\lambda \in
\partial U(0)$). To see this, let $\lambda $ satisfy $S^{U}(\varphi )\leq
\sum_{t=0}^{T}\langle \varphi _{t},\lambda _{t}\rangle ,\,\forall \varphi
\in \mathcal{E}$. In particular then for all $z\in -\mathcal{A}$ and all $%
\varphi \in \Phi _{z}(0)$ it holds that $S^{U}(\varphi )\leq \langle
\sum_{t=0}^{T}\varphi _{t},\lambda \rangle \leq \langle
\sum_{t=0}^{T}\varphi _{t}+z,\lambda \rangle \leq 0$, as $\langle z,\lambda
\rangle \geq 0$ for all $\lambda \in \mathcal{A}^{\circ }$, which in turns
yields $\mathfrak{P}(0)\leq 0$.

\item[(b)] 
\begin{equation*}
\text{There exists } Q\in \mathrm{Prob}^{1}(\Omega )\cap \mathcal{A}^{\circ }%
\text{ such that }\mathcal{D}(Q)<+\infty \Leftrightarrow \mathfrak{P}%
(0)<+\infty \text{.}
\end{equation*}%
Indeed, by definition we have $\mathfrak{P}(0)\leq \int_{\Omega }0\,\mathrm{d%
}Q+\mathfrak{P}^{\ast }(Q)=\mathfrak{P}^{\ast }(Q)$, but from Lemma \ref%
{propformaconiugata} (which does not rely on Lemma \ref{propdualrepre}) we
have $\mathfrak{P}^{\ast }(Q)=\mathcal{D}(Q)+\sigma _{\mathcal{A}}(Q)=%
\mathcal{D}(Q)$ (the latter inequality coming from $Q\in \mathcal{A}^{\circ
} $). Hence $\mathfrak{P}(0)\leq \mathcal{D}(Q)<+\infty $. Conversely, $%
\mathfrak{P}(0)<+\infty $ implies the existence of a minimum point in %
\eqref{dualrepreabstract}.
\item [(c)] The existence of the optimizer in MOT implies that $\mathcal{M}(\widehat{Q}_{0},\widehat{Q}_{1},...\widehat{Q}_{T})$ is not empty and that the marginals must be in convex order. In EMOT the optimizer $Q^{\ast}$ belongs to $\mathrm{Prob}^{1}(\Omega )\cap \mathcal{A}^{\circ }$ with  $\mathcal{D}(Q^{\ast})<+\infty$ with no other requirement.\end{itemize}
\end{remark}

\begin{remark}
\label{whyminus} From the proof of Theorem \ref{mainEMOTtheoremgeneral} it
becomes clear that the use of $-\mathcal{A}$ in place of $\mathcal{A}$ in
defining $\mathfrak{P}(c)$ is somehow a matter of taste. Now the infimum in %
\eqref{dualrepreabstract} is in fact taken over measures in the polar $%
\mathcal{A}^{\circ }$. Instead, without the minus sign in defining $%
\mathfrak{P}(c),$ we would work with $(-\mathcal{A})^{\circ }$, which is
less comfortable in the computations of the proof.
\end{remark}

\begin{corollary}
\label{mainthm}Suppose that Assumption \ref{asscallstozero}.(i) holds with $%
K_{0}^{1},\dots ,K_{0}^{d},\dots ,K_{T}^{1},\dots ,K_{T}^{d}$ compact
subsets of ${\mathbb{R}}$, that $c:\Omega \rightarrow (-\infty ,+\infty ]$
is lower semicontinuous and that $U(\varphi )=0$ for some $\varphi \in 
\mathcal{E}$. Then (\ref{dualrepreabstract}) holds true and if ${\mathfrak{P}%
(c)<+\infty }$ then there exists an optimum in the RHS of (\ref%
{dualrepreabstract}).
\end{corollary}

\begin{proof}
When $K_{0}^{1},\dots ,K_{0}^{d},\dots ,K_{T}^{1},\dots ,K_{T}^{d}$ are
compact then $C_{0:T}=\mathcal{C}_{b}(\Omega ).$ If $U(\varphi )=0$ for some 
$\varphi \in \mathcal{E}$, then \eqref{ineqoutsidecpt} and %
\eqref{condconvergencecalls} are automatically satisfied: indeed one can
take $\mathfrak{K}(n)=\Omega $ and $f_{t}^{n}\equiv -\varphi _{t},t=0,\dots
,T,n\geq 1$. Obviously, a possible choice for such a $\varphi $ is $\varphi
=0.$
\end{proof}

\label{secusefulrephrase} We now rephrase our findings in Theorem \ref%
{mainEMOTtheoremgeneral}, with minor additions, to get the formulations in
Corollary \ref{maincor} and Corollary \ref{maincorcasecpt} which will
simplify our discussion of Section \ref{secpartcases}. In particular, this
reformulation will come in handy when dealing with subhedging dualities in
Corollaries \ref{corolldivinduced}-\ref{corollrob} and Proposition \ref%
{propdualwithbestfit}.

For a given proper concave $U:\mathcal{E}\rightarrow {\mathbb{R}}$, recall
the definition of $S^{U}$ in \eqref{defSU} and, for $V(\cdot )=-U(-\cdot )$,
the definition $S_{V}(\varphi ):=-S^{U}(-\varphi )$.

Furthermore, given functions $c:\Omega \rightarrow (-\infty ,+\infty ]$, $%
g:\Omega \rightarrow \lbrack -\infty ,+\infty )$ we introduce the sets 
\begin{equation}
\mathcal{S}_{sub}(c):=\left\{ \varphi \in \mathrm{dom}(S^U)\mid \exists
\Delta \in \mathcal{H}\text{ s.t. }\sum_{t=0}^{T}\varphi (x_0,\dots,x_{t})+I^{\Delta
}(x)\leq c(x)\,\,\,\forall \,x\in \Omega \right\}  \label{ssub}
\end{equation}%
\begin{equation}
\mathcal{S}_{sup}(g):=\left\{ \varphi \in \mathrm{dom}(S_V)\mid \exists
\Delta \in \mathcal{H}\text{ s.t. }\sum_{t=0}^{T}\varphi (x_0,\dots,x_{t})+I^{\Delta
}(x)\geq g(x)\,\,\,\forall \,x\in \Omega \right\} \,  \label{ssup}
\end{equation}%
and observe that $\mathcal{S}_{sup}(g)=-\mathcal{S}_{sub}(-g)$.

\begin{corollary}
\label{maincor} 
Suppose that the assumptions in Theorem \ref{mainEMOTtheoremgeneral} are
satisfied, that $g:\Omega \rightarrow \lbrack -\infty ,+\infty )$ is upper
semicontinuous and that also condition (\ref{controlfrombelow}) holds
replacing $c $ with $-g$. Then the following hold 
\begin{equation}
\inf_{Q\in \mathrm{Mart}(\Omega )}\left( E_{Q}\left[ c\right] +\mathcal{D}%
(Q)\right) =\sup_{\varphi \in \mathcal{S}_{sub}(c)}S^{U}\left( \varphi
\right) \,,  \label{subhedgingcor}
\end{equation}%
\begin{equation}
\sup_{Q\in \mathrm{Mart}(\Omega )}\left( E_{Q}\left[ g\right] -\mathcal{D}%
(Q)\right) =\inf_{\varphi \in \mathcal{S}_{sup}(g)}S_{V}\left( \varphi
\right) \,.  \label{superhedgingcor}
\end{equation}%
Finally, if LHS of \eqref{subhedgingcor} (resp. \eqref{superhedgingcor}) is
finite, then an optimum exists in the LHS of \eqref{subhedgingcor} (resp. %
\eqref{superhedgingcor}).
\end{corollary}

\begin{proof}
Equation \eqref{subhedgingcor} is an easy rephrasing of the corresponding %
\eqref{dualrepreabstract}, taking $\mathcal{A}=\mathcal{I}$ as in Example %
\ref{exampleamtg} so that $\mathrm{Prob}^1(\Omega)\cap\mathcal{A}^{\circ }=%
\mathrm{Mart}(\Omega )$. Existence of an optimum in the case of finiteness
of LHS of \eqref{subhedgingcor} follows again from Theorem \ref%
{mainEMOTtheoremgeneral}. As to \eqref{superhedgingcor}, we observe that for 
$c:=-g$ we get from \eqref{subhedgingcor} 
\begin{equation*}
\sup_{\varphi \in \mathcal{S}_{sub}(-g)}S^{U}\left( \varphi \right)
=\inf_{Q\in \mathrm{Mart}(\Omega )}\left( E_{Q}\left[ -g\right] +\mathcal{D}%
(Q)\right) =-\sup_{Q\in \mathrm{Mart}(\Omega )}\left( E_{Q}\left[ g\right] -%
\mathcal{D}(Q)\right) \,.
\end{equation*}%
From $\mathcal{S}_{sup}(g)=-\mathcal{S}_{sub}(-g)$ and $S_{V}(\cdot
)=-S^{U}(-\cdot ),$ we get $\sup_{\varphi \in \mathcal{S}_{sub}(-g)}S^{U}%
\left( \varphi \right) =-\inf_{\varphi \in \mathcal{S}_{sup}(g)}S_{V}\left(
\varphi \right) $. Existence of an optimum when LHS of %
\eqref{superhedgingcor} is finite can be inferred in a similar way.
\end{proof}

\begin{corollary}
\label{maincorcasecpt}{If }$d=1$ and{\ $\Omega :=K_{0}\times \dots \times
K_{T}$ for compact sets $K_{0},\dots ,K_{T}\subseteq {\mathbb{R}}$, %
\eqref{subhedgingcor} and \eqref{superhedgingcor}, as well as existence of
optima, are guaranteed by the following simplified set of assumptions: $%
c:\Omega \rightarrow (-\infty ,+\infty ]$ is lower semicontinuous, $g:\Omega
\rightarrow (-\infty ,+\infty ]$ is upper semicontinuous and $U(\varphi )=0$
for some $\varphi \in \mathcal{E}$.}
\end{corollary}

\begin{proof}
When $K_{0},\dots .K_{T}\subseteq {\mathbb{R}}$ are compact, we may repeat
the proof of Corollary \ref{maincor} invoking the Corollary \ref{mainthm} in
place of the more general Theorem \ref{mainEMOTtheoremgeneral}.
\end{proof}

%
%
%

Observe that in the previous results, we allow to choose, for static
hedging, for each $t=0,\dots ,T$, a subspace $\mathcal{E}_{t}\subseteq {C}%
_{0:t}$, potentially allowing to consider also Asian and path dependent
options in the sets $\mathcal{E}_{t}$. 
In some of the subsequent results
 we will instead consider for static
hedging the sets $\mathcal{E}_{t}\subseteq {C}_{t},\,t=0,\dots ,T$
consisting of deterministic amounts, units of underlying stock at time $t$
and call options with different strike prices and same maturity $t$. 

\bigskip 

In the subsequent sections we will only formulate the statements regarding
the subhedging price, as the corresponding statements for the superhedging
price can be obtained in the obvious way just described.

\subsection{Proof of Theorem \protect\ref{mainEMOTtheoremgeneral}}

\begin{remark}
Set 
\begin{equation*}
\widetilde{\mathbf{\Phi }}_{z}(c):=\left\{ \varphi \in \mathcal{E},\text{ }%
\sum_{t=0}^{T}\varphi _{t}(x_0,\dots,x_{t})+z(x)\leq c(x)\,\,\,\forall \,x\in \Omega
\right\}
\end{equation*}%
and observe that $\mathbf{\Phi }_{z}(c)=\widetilde{\mathbf{\Phi }}%
_{z}(c)\cap \mathrm{dom}(S^{U})$. Then, under the convention $\sup \emptyset
=-\infty $, 
\begin{equation}  \label{deftildefrakPgeneral}
\mathfrak{P}(c):=\sup_{z\in -\mathcal{A}}\sup_{\varphi \in \mathbf{\Phi }%
_{z}(c)}S^{U}\left( \varphi \right) =\sup_{z\in -\mathcal{A}}\sup_{\varphi
\in \widetilde{\mathbf{\Phi }}_{z}(c)}S^{U}\left( \varphi \right) \,.
\end{equation}%
To see this, we consider different cases for a fixed $z\in -\mathcal{A}$.

\textit{Case 1}: $\mathbf{\Phi }_{z }(c)=\emptyset$, which means $%
\sup_{\varphi \in \mathbf{\Phi }_{z }(c)}S^{U}(\varphi)=-\infty$ by
convention. If $\widetilde{\mathbf{\Phi }}_{z }(c)=\emptyset$ then $%
\sup_{\varphi \in \widetilde{\mathbf{\Phi }}_{z }(c)}S^{U}(\varphi)=-\infty$
by convention, if $\widetilde{\mathbf{\Phi }}_{z }(c)\neq \emptyset$ then $%
\sup_{\varphi \in \widetilde{\mathbf{\Phi }}_{z }(c)}S^{U}(\varphi)=-\infty$
since for every $\varphi\in \widetilde{\mathbf{\Phi }}_{z }(c)$ we have $%
S^U(\varphi)=-\infty$, as $\varphi\notin\mathrm{dom}(S^U)$ .

\textit{Case 2}: $\mathbf{\Phi }_{z}(c)\neq \emptyset $. Then $\widetilde{%
\mathbf{\Phi }}_{z}(c)\neq \emptyset $ too, and $\sup_{\varphi \in 
\widetilde{\mathbf{\Phi }}_{z}(c)}S^{U}(\varphi )=\sup_{\varphi \in {\mathbf{%
\Phi }}_{z}(c)}S^{U}(\varphi )$ since we can ignore all the $\varphi \in 
\widetilde{\mathbf{\Phi }}_{z}(c)\setminus {\mathbf{\Phi }}_{z}(c)$ (which
produce values $S^{U}(\varphi )=-\infty $).
\end{remark}

The proof of Theorem \ref{mainEMOTtheoremgeneral} is split in the following
Lemmas \ref{propdualrepre}, \ref{propformaconiugata}, \ref%
{proprepreconintegralsc}, \ref{claimdualpartial}, \ref{corcpt} which are
then combined in Lemma \ref{claimconclusion}.

\begin{lemma}
\label{propdualrepre} Under Assumptions \ref{asscallstozero} and (\ref{BB}),
Item (i) in Theorem \ref{mainEMOTtheoremgeneral} holds. Moreover, the
restriction of $\mathfrak{P}$ to $C_{0:T}$ satisfies 
\begin{equation}
\mathfrak{P}(c)=\min_{\substack{ \lambda \in (C_{0:T})^{\ast },  \\ \lambda
\geq 0,\lambda (1)=1}}\left( \langle c,\lambda \rangle +\mathfrak{P}^{\ast
}(\lambda )\right) \,\,\,\,\,\forall \,c\in C_{0:T}
\label{dualrepreabstractgeneral}
\end{equation}%
for 
\begin{equation*}
\mathfrak{P}^{\ast }(\lambda )=\sup_{c\in C_{0:T}}\left( \mathfrak{P}%
(c)-\langle c,\lambda \rangle \right) ,\,\,\,\,\,\,\,\lambda \in
(C_{0:T})^{\ast }\,.
\end{equation*}
\end{lemma}

\begin{proof}
Suppose that $\mathfrak{P}(\widehat{c})<+\infty $ for some $\widehat{c}\in
B_{0:T}$. To prove that $\mathfrak{P}(c)>-\infty $ for every $c\in B_{0:T}$,
it is (more than) enough to show that

\begin{equation}
\mathbf{\Phi }_{z}(c)\neq \emptyset \,\,\,\,\,\,\,\forall z\in -\mathcal{A}%
\,.  \label{setforsupnotempty}
\end{equation}%
Set $\mathfrak{K}_{n}=\mathfrak{K}_{0}(n)\times \dots \times \mathfrak{K}%
_{T}(n)\subseteq \Omega $. Observe that whenever $c\in B_{0:T}$ is given we
have for every $n\geq 1$ {\small 
\begin{equation*}
c(x)-z(x)\geq -\sup_{x\in \mathfrak{K}_{n}}\left\vert c(x)-z(x)\right\vert
\geq -\left\Vert c-z\right\Vert _{0:T}\sup_{x\in \mathfrak{K}%
_{n}}\left( 1+\sum_{t=0}^{T}\sum_{j=1}^{d}\left\vert x_{t}^{j}\right\vert
\right) >-\infty \,\,\,\forall x\in \mathfrak{K}(n)
\end{equation*}%
}and%
\begin{align*}
c(x)-z(x)& \geq -\left\Vert c-z\right\Vert _{0:T}\left( 1+\sum_{s\leq u\leq
t}\sum_{j=1}^{d}\left\vert x_{t}^{j}\right\vert \right) \\
& \overset{\eqref{ineqoutsidecpt}}{\geq }-\left\Vert c-z\right\Vert
_{0:T}-\left\Vert c-z\right\Vert _{0:T}\sum_{t=0}^{T}f_{t}^{n}(x_{0},\dots
,x_{t})\,\,\,\,\,\forall x\in \Omega \setminus \mathfrak{K}(n)\,.
\end{align*}%
Thus, 
\begin{equation*}
c(x)-z(x)\geq -\left\Vert c-z\right\Vert _{0:T}-\sup_{x\in \mathfrak{K}%
_{n}}\left\vert c(x)-z(x)\right\vert -\left\Vert c-z\right\Vert
_{0:T}\sum_{t=0}^{T}f_{t}^{n}(x_{0},\dots ,x_{t})\,\,\,\,\,\forall x\in
\Omega \,.
\end{equation*}

If we now show that for $n$ big enough $[-\left\Vert c-z\right\Vert
_{0:T}f_{t}^{n}]_{0\leq t\leq T}\in \mathrm{dom}(S^{U})$, we then conclude
that $[-\left\Vert c-z\right\Vert _{0:T}-\sup_{x\in K_{n}}\left\vert
c(x)-z(x)\right\vert -\left\Vert c-z\right\Vert _{0:T}f_{t}^{n}]_{0\leq
t\leq T}\in \mathrm{dom}(S^{U})$ by cash additivity of $S^{U}$ and at the
same time $[-\left\Vert c-z\right\Vert _{0:T}-\sup_{x\in \mathfrak{K}%
_{n}}\left\vert c(x)-z(x)\right\vert -\left\Vert c-z\right\Vert
_{0:T}f_{t}^{n}]_{0\leq t\leq T}\in \mathbf{\Phi }_{z}(c)$ by definition.
This in particular proves $\mathfrak{P}(c)>-\infty $. Going then back to
checking $[-\left\Vert c-z\right\Vert _{0:T}f_{t}^{n}]_{0\leq t\leq T}\in 
\mathrm{dom}(S^{U})$, observe that 
\begin{align*}
S^{U}([-\left\Vert c-z\right\Vert _{0:T}f_{t}^{n}]_{t})& =\sup_{\alpha \in {%
\mathbb{R}}^{T+1}}\left( U([-\left\Vert c-z\right\Vert
_{0:T}f_{t}^{n}]_{t}+\alpha )-\sum_{t=0}^{T}\alpha _{t}\right) \\
& \geq U(-\left\Vert c-z\right\Vert _{0:T}[f_{t}^{n}]_{t})=-V(\left\Vert
c-z\right\Vert _{0:T}[f_{t}^{n}]_{t})\rightarrow _{n}0>-\infty
\end{align*}%
by Assumption \ref{asscallstozero}. The fact that $\mathfrak{P}(c)<+\infty $
will follow once we show monotonicity, cash additivity and concavity.
Monotonicity is trivial: if $c_{1}\leq c_{2}$ then $\mathbf{\Phi }%
_{z}(c_{1})\subseteq \mathbf{\Phi }_{z}(c_{2})$ for every $z\in -\mathcal{A}$
(both the sets might be empty). The cash additivity property can be seen as
follows: given $\beta \in {\mathbb{R}}$ and setting $\mathbf{1}=[1,\dots
,1]\in {\mathbb{R}}^{T}$, observe that whenever $z\in \mathcal{A}$ is given $%
\varphi \in \mathbf{\Phi }_{z}(c+\beta )\Leftrightarrow \varphi -\frac{\beta }{%
T+1}\mathbf{1}\in \mathbf{\Phi }_{z}(c)$ since, by cash additivity of $S^{U}$%
, $\mathrm{dom}(S^{U})+{\mathbb{R}}^{T+1}=\mathrm{dom}(S^{U})$.
Consequently, 
\begin{align*}
\mathfrak{P}(c+\beta )& =\sup_{z\in -\mathcal{A}}\sup_{\varphi \in \mathbf{%
\Phi }_{z}(c+\beta )}S^{U}\left( \varphi \right) =\sup_{z\in -\mathcal{A}%
}\sup_{\varphi \in \mathbf{\Phi }_{z}(c)}S^{U}\left( \varphi +\frac{\beta }{T+1%
}\mathbf{1}\right) \\
& =\sum_{t=0}^{T}\frac{\beta }{T+1}+\sup_{z\in -\mathcal{A}}\sup_{\varphi \in 
\mathbf{\Phi }_{z}(c)}S^{U}(\varphi )=\mathfrak{P}(c)+\beta \,.
\end{align*}%
Coming to concavity, it is convenient to rewrite $\mathfrak{P}(c)$ in a
slightly more convenient form%
\begin{equation}
\mathfrak{P}(c)=\sup \left\{ S^{U}(\varphi )\mid \varphi \in \mathrm{dom}%
(S^{U}),\exists \,z\in -\mathcal{A}\text{ s.t. }\sum_{t=0}^{T}\varphi
_{t}(x_{0},\dots ,x_{t})+z(x)\leq c(x)\,\forall x\in \Omega \right\}
\label{rephrasefrakP}
\end{equation}%
and to recall that whenever $c\in B_{0:T}$ is given, the set over which we
take the supremum in RHS of \eqref{rephrasefrakP} is not empty by %
\eqref{setforsupnotempty}. Take then $c_{i}\in B_{0:T}$ and associated $%
z_{i}\in -\mathcal{A},\varphi ^{i}\in \mathrm{dom}(S^{U})$ with $%
\sum_{t=0}^{T}\varphi _{t}^{i}+z_{i}\leq c_{i}$. Define $c_{\alpha }=\alpha
c_{1}+(1-\alpha )c_{2}$ and analogously $z_{\alpha }$ and $\varphi ^{\alpha }$ for 
$\alpha \in \lbrack 0,1]$. Then clearly $\sum_{t=0}^{T}\varphi _{t}^{\alpha
}+z_{\alpha }\leq c_{\alpha }$. Combining this with Lemma \ref{remalmostconcave}
we obtain
\begin{align*}
& \alpha S^{U}(\varphi ^{1})+(1-\alpha )S^{U}(\varphi ^{2})\overset{%
\eqref{almostconcave}}{\leq }S^{U}(\varphi ^{\alpha }) \\
& \leq \sup \left\{ S^{U}(\varphi )\mid \varphi \in \mathrm{dom}%
(S^{U}),\exists \,z\in -\mathcal{A}\text{ s.t. }\sum_{t=0}^{T}\varphi
_{t}(x_{0},\dots ,x_{t})+z(x)\leq c_{\alpha }(x)\,\forall x\in \Omega \right\}
\\
& \overset{\eqref{rephrasefrakP}}{=}\mathfrak{P}(c_{\alpha })\,.
\end{align*}%
Taking now the supremum over $z_{i},\varphi ^{i}$ with $\sum_{t=0}^{T}%
\varphi _{t}^{i}+z_{i}\leq c_{i}$ we obtain

\begin{equation}  \label{almostconcavep}
\alpha \mathfrak{P}(c_1)+(1-\alpha)\mathfrak{P}(c_2)\leq\mathfrak{P}(\alpha
c_1+(1-\alpha) c_2)\,\,\,\,\forall \alpha\in[0,1], c_1,c_2\in B_{0:T}\,.
\end{equation}

Notice that up to this point we have $\mathfrak{P}(c_i)\in(-\infty,+\infty]$
so \eqref{almostconcavep} makes sense.

Now we can combine \eqref{almostconcavep} with the fact that $\mathfrak{P}%
(c)>-\infty $ for every $c\in B_{0:T}$ to show that $\mathfrak{P}(c)<+\infty 
$ for every $c\in B_{0:T}$. Indeed, suppose that $\mathfrak{P}(\widetilde{c}%
)=+\infty $ for some $\widetilde{c}\in B_{0:T}$. We know by hypothesis that $%
\mathfrak{P}(\widehat{c})<+\infty $ for some $\widehat{c}\in B_{0:T}$, and
by what we have previously proved we know that $\mathfrak{P}(2\widehat{c}-%
\widetilde{c})>-\infty $. Observing that $\widehat{c}=\alpha (2\widehat{c}-%
\widetilde{c})+\left( 1-\alpha \right) \widetilde{c}$ for $\alpha =\frac{1}{2}$,
we have from \eqref{almostconcavep} 
\begin{equation*}
+\infty =\alpha \mathfrak{P}(2\widehat{c}-\widetilde{c})+(1-\alpha )\mathfrak{P}(%
\widetilde{c})\leq \mathfrak{P}(\alpha (2\widehat{c}-\widetilde{c})+\left(
1-\alpha \right) \widetilde{c})=\mathfrak{P}(\widehat{c})<+\infty \,.
\end{equation*}%
This yields a contradiction, thus there can be no $\widetilde{c}\in B_{0:T}$
with $\mathfrak{P}(\widetilde{c})=+\infty $. Hence {$\mathfrak{P}%
:B_{0:T}\rightarrow {\mathbb{R}}$ is cash additive, concave and
nondecreasing on $B_{0:T}$. } Then it is automatically norm continuous on $%
B_{0:T}$ by the Extended Namioka-Klee Theorem (see \cite{bfnam}). The
Fenchel-Moreau type dual representation \eqref{dualrepreabstractgeneral}
holds, again by the Extended Namioka-Klee Theorem, this time applied on the
restriction of $\mathfrak{P}$ to $C_{0:T}$, plus standard arguments
involving monotonicity and cash additivity to prove that $\mathfrak{P}^{\ast
}(\lambda )<+\infty \Rightarrow \lambda \geq 0,\lambda (1)=1$. See for
example \cite{FollmerSchied2} Theorem 4.16 for an exploitable technique for
a similar argument.
\end{proof}

\begin{remark}
Under Assumptions \ref{asscallstozero} and (\ref{BB}), $S^{U}(\varphi
)<+\infty $ for every $\varphi \in \mathrm{dom}(S^{U})$. Indeed, choosing $%
c_{\varphi }:=\sum_{t=0}^{T}\varphi _{t}$ we get that $\varphi \in {\Phi }%
_{0}(c_{\varphi })$ and thus $S^{U}(\varphi )\leq \mathfrak{P}(c_{\varphi
})<+\infty $, by Lemma \ref{propdualrepre}.
\end{remark}

\begin{lemma}
\label{propformaconiugata} 
For every $\lambda \in (C_{0:T})^{\ast }$ such that $\lambda \geq
0$ we have 
\begin{equation*}
\mathfrak{P}^{\ast }(\lambda )=(S^{U})^{\ast }(\lambda _{0},\dots ,\lambda_{T})+\sigma _{\mathcal{A}}(\lambda ).
\end{equation*}
If in addition $\lambda(1)=1$ then
\begin{equation*}
(S^{U})^{\ast }(\lambda _{0},\dots ,\lambda _{T}):=\sup_{\varphi \in 
\mathcal{E}}\left( S^{U}\left( \varphi \right) -\sum_{t=0}^{T}\langle
\varphi _{t},\lambda _{t}\rangle \right) =\mathcal{D}(\lambda _{0},\dots
,\lambda _{T}) = \mathcal{D}(\lambda)\,.
\end{equation*}
\end{lemma}

\begin{proof}
See Appendix \ref{appnoncpt}.
\end{proof}

\begin{remark}
\label{remDbddbelow} Under Assumption \ref{asscallstozero}, $U(0)\in {%
\mathbb{R}}$ and therefore $(S^{U})^{\ast }(\lambda )\geq S^{U}(0)\geq
U(0)>-\infty $ for every $0\leq \lambda \in (C_{0:T})^{\ast }$.
\end{remark}

\begin{lemma}
\label{proprepreconintegralsc} Under Assumptions \ref{asscallstozero} and (%
\ref{BB}) let $0\leq \lambda \in (C_{0:T})^{\ast }$, $\lambda (1)=1$ be
given and define $0\leq \lambda _{t}=\lambda |_{C_{0:t}}\in (C_{0:t})^{\ast
} $. If $[\lambda _{0},\dots ,\lambda _{T}]\in \mathrm{dom}(\mathcal{D})$
then there exists a unique $Q\in \mathrm{Prob}^{1}(\Omega )$ which
represents $\lambda $ on $C_{0:T}$, i.e. 
\begin{equation*}
\langle \varphi ,\lambda \rangle =E_{Q} \left[\varphi\right]\,,\,\,\,\forall
\,\varphi \in C_{0:T}\,.
\end{equation*}
\end{lemma}

\begin{proof}
The proof is an adaption of \cite{Bogachev07} Theorem 7.10.6. We first
stress the fact that $\lambda_t=\lambda|_{C_{0:t}}\in (C_{0:t})^*$ is a
consequence of \eqref{consistencyrestrictions}. We will apply Proposition %
\ref{propdaniellstone}. To do so, we show that for a fixed $\varepsilon>0$ and for $n$ big
enough, 
we may define a set $K_\varepsilon:=\mathfrak{K}_0(n)\times\dots \mathfrak{K}_T(n)$
that is compact (since so are all the factors) and satisfies  the
assumptions in Proposition \ref{propdaniellstone}.  Suppose that a given $\varphi\in C_{0:T}$ satisfies: $\varphi(x)=0$
for every $x\in \bigtimes_{t=0}^T\mathfrak{K}_t(n)$. We also have
automatically that 
\begin{equation*}
\left|\varphi(x)\right|\leq \left\|
\varphi\right\|_{0:T}\left(1+\sum_{t=0}^T\sum_{j=1}^d\left|x^j_t\right|%
\right)\,\,\,\forall\,x\in\Omega\,.
\end{equation*}
By Assumption \ref{asscallstozero} we then have: 
\begin{equation*}
\left|\varphi(x)\right|\leq \left\| \varphi\right\|_{0:T}\left(\sum_{t=0}^T
f^{n}_t(x_0,\dots,x_t)\right)\,\,\,\,\forall [x_0,\dots,x_T]\in
\Omega\setminus \mathfrak{K}_0(n)\times\dots \times\mathfrak{K}_T(n)\,.
\end{equation*}
Since moreover by assumption $\varphi\equiv 0$ on $\bigtimes_{t=0}^T 
\mathfrak{K}_t(n)$ we get: 
\begin{equation}  \label{dominatingg}
\left|\varphi(x)\right|\leq \left\| \varphi\right\|_{0:T}\left(\sum_{t=0}^T
f^{n}_t(x_0,\dots,x_t)\right)\,\,\,\,\forall [x_0,\dots,x_T]\in \Omega\,.
\end{equation}

Then, the following chain of inequalities holds for every $\Gamma>0$
(additional motivation will be provided immediately after): 
\begin{align}
\left|\langle \varphi,\lambda\rangle\right|&\leq \langle
\left|\varphi\right|,\lambda\rangle\leq \langle \left\|
\varphi\right\|_{0:T}\sum_{t=0}^T f^{n}_t,\lambda\rangle=\left\|
\varphi\right\|_{0:T}\sum_{t=0}^T\langle f^{n}_t,\lambda\rangle=\left\|
\varphi\right\|_{0:T}\sum_{t=0}^T\langle f^{n}_t,\lambda_t\rangle  \label{34}
\\
&=\left\| \varphi\right\|_{0:T}\frac{1}{\Gamma}\sum_{t=0}^T\langle \Gamma
f^{n}_t,\lambda_t\rangle \leq \left\| \varphi\right\|_{0:T}\left(\frac{1}{%
\Gamma}\mathcal{D}(\lambda_0,\dots,\lambda_T)+\frac{1}{\Gamma}V\left( \Gamma
f^{n}_0,\dots,\Gamma f^{n}_T\right)\right)\,.  \label{56}
\end{align}

Here, \eqref{34} follows from positivity of $\lambda$, from %
\eqref{dominatingg}, from linearity and the fact that one defines $%
\lambda_t:=\lambda|_{C_{0:t}}\in(C_{0:t})^*$, while \eqref{56} follows from
linearity and from Fenchel inequality \eqref{fenchelgeneral}. %
Since we are assuming by hypothesis that $[\lambda_0,\dots,\lambda_T] \in%
\mathrm{dom}(\mathcal{D})$, we can select $\Gamma>0$ such that $\frac{1}{%
\Gamma}\mathcal{D}(\lambda_0,\dots,\lambda_T)\leq \frac{\varepsilon}{2}$.
Select now $n$ in such a way that $\frac{1}{\Gamma}V\left( \Gamma
f^{n}_0,\dots,\Gamma f^{n}_T\right)\leq \frac{\varepsilon}{2}$ for every $%
s\leq T$ (which is possible by Assumption \ref{asscallstozero}). Then
continuing from \eqref{56} we get 
\begin{equation*}
\left|\langle \varphi,\lambda\rangle\right|\leq \left\|
\varphi\right\|_{0:T}\left(\frac{\varepsilon}{2}+\frac{\varepsilon}{2}%
\right)\leq\varepsilon \left\| \varphi\right\|_{0:T}\,.
\end{equation*}
The thesis now follows combining Proposition \ref{propdaniellstone} and
Daniell-Stone Theorem \ref{thmdaniellstone}.
\end{proof}

\begin{lemma}
\label{claimdualpartial}

Under Assumptions \ref{asscallstozero} and (\ref{BB}), the equation %
\eqref{dualrepreabstract} holds for every $c\in C_{0:T}$, with a minimum in
place of the infimum.
\end{lemma}

\begin{proof}

Combining Lemma \ref{propdualrepre}, Lemma \ref{propformaconiugata}, and
Lemma \ref{proprepreconintegralsc} we have 
\begin{align}
\mathfrak{P}(c)&\overset{\text{L.}\ref{propdualrepre}}{=}\min_{\substack{ %
\lambda\in (C_{0:T})^*,  \\ \lambda\geq 0,\lambda(1)=1}}\left(\langle
c,\lambda\rangle+\mathfrak{P}^*(\lambda)\right)  \notag \\
&\overset{\text{L.}\ref{propformaconiugata} }{=}\min_{\substack{ \lambda\in
(C_{0:T})^*,  \\ \lambda\geq 0,\lambda(1)=1}}\left(\langle c,\lambda\rangle+%
\mathcal{D}(\lambda)+\sigma_\mathcal{A}(\lambda)\right) 
\notag \\
&\overset{(\star)}{=}\min_{\substack{ \lambda\in
(C_{0:T})^*,[\lambda_0,\dots,\lambda_T]\in\mathrm{dom}(\mathcal{D})  \\ %
\lambda\geq 0,\lambda(1)=1}}\left(\langle c,\lambda\rangle+\mathcal{D}%
(\lambda)+\sigma_\mathcal{A}(\lambda)\right)  \notag \\
&\overset{\text{L.}\ref{proprepreconintegralsc} }{=}\min_{\substack{ Q\in 
\mathrm{Prob}^1(\Omega),  \\ Q\in\mathrm{dom}(\mathcal{D})}}\left(E_{Q} %
\left[c\right]+\mathcal{D}(Q)+\sigma_\mathcal{A}(Q)\right)
\label{identmeasuresandfunct} \\
&\overset{(\star) }{=}\min_{Q\in \mathrm{Prob}^1(\Omega)}\left(E_{Q} \left[c%
\right]+\mathcal{D}(Q)+\sigma_\mathcal{A}(Q)\right)  \notag
\end{align}
where in $(\star)$ we used the fat that $\mathcal{D}$ is bounded from below
by $S^U(0)$ by Remark \ref{remDbddbelow}, hence $[\lambda_0,\dots,\lambda_T]%
\in\mathrm{dom}(\mathcal{D})\Leftrightarrow\mathcal{D}(\lambda_0,\dots,%
\lambda_T)<+\infty$, and in \eqref{identmeasuresandfunct} we identified
probability measures $Q\in\mathrm{Prob}^1(\Omega)$ and their induced
functionals, as well as the marginals $Q_t$ of such measures with the
restrictions of such functionals to ${C}_{0:t}$.
\end{proof}

\begin{lemma}
\label{corcpt} Under Assumptions \ref{asscallstozero} and (\ref{BB}) the
sublevel 
\begin{equation*}
\{Q\in \mathrm{Prob}^{1}(\Omega )\cap \mathcal{A}^{\circ }\mid \mathcal{D}%
(Q)\leq \Xi \}
\end{equation*}%
is $\sigma ((C_{0:T})^{\ast },C_{0:T})|_{\mathrm{Prob}^{1}(\Omega )}$%
-(sequentially) compact for every $\Xi \in {\mathbb{R}}$.
\end{lemma}

\begin{proof}
We show that $\{\lambda \in (C_{0:T})^{\ast }\mid \lambda \geq 0,\lambda
(1)=1,\mathfrak{P}^{\ast }(\lambda )\leq \Xi \}$ is weak$^{\ast }$
-(sequentially) compact. To do so, it is  enough to prove that $%
\{\lambda \in (C_{0:T})^{\ast }\mid \mathfrak{P}^{\ast }(\lambda )\leq \Xi
\} $ is weak$^{\ast }$-(sequentially) compact. Once this is done notice
that, combining the fact that $\sigma _{\mathcal{A}}=\delta _{\mathcal{A}%
^{\circ }} $ and Lemma \ref{propformaconiugata}, $\{\lambda \in
(C_{0:T})^{\ast }\mid \mathfrak{P}^{\ast }(\lambda )\leq \Xi \}=\{\lambda
\in (C_{0:T})^{\ast }\mid \lambda \geq 0,$ $\lambda (1)=1,\mathcal{D}%
(\lambda )\leq \Xi \}\cap \mathcal{A}^{\circ }$. Since by Lemma \ref%
{proprepreconintegralsc} there is a natural identification between
normalized nonnegative functionals in $\mathrm{dom}(\mathfrak{P}^{\ast })$
and the measures in $\mathrm{Prob}^{1}(\Omega )$, the proof will be
complete. In order to prove that $\{\lambda \in (C_{0:T})^{\ast }\mid 
\mathfrak{P}^{\ast }(\lambda )\leq \Xi \} $ is weak$^{\ast }$-(sequentially)
compact, observe first that by \eqref{dualrepreabstractgeneral} we have for
every $r>0$ and $\lambda \in (C_{0:T})^{\ast }\text{ s.t. }\mathfrak{P}%
^{\ast }(\lambda )\leq \Xi $

\begin{equation}
\sup_{\substack{ c\in C_{0:T},  \\ \left\Vert c\right\Vert _{0:T}\leq r}}%
\left\vert \langle c,\lambda \rangle \right\vert =\sup_{\substack{ c\in
C_{0:T},  \\ \left\Vert c\right\Vert _{0:T}\leq r}}\langle c,\lambda \rangle
\leq \sup_{\substack{ c\in C_{0:T},  \\ \left\Vert c\right\Vert _{0:T}\leq r 
}}\left( -\mathfrak{P}(-c)\right) +\mathfrak{P}^{\ast }(\lambda )\leq \Xi
+\sup _{\substack{ c\in C_{0:T},  \\ \left\Vert c\right\Vert _{0:T}\leq r}}%
\left( -\mathfrak{P}(-c)\right) \,.  \label{boundfornorms}
\end{equation}%
Now since $-\mathfrak{P}(\cdot )$ is real valued, convex and continuous on $%
C_{0:T}$ (Lemma \ref{propdualrepre}) it follows from \cite{Aliprantis}
Theorem 5.43 that the RHS in (\ref{boundfornorms}) is finite for some $r>0$.
Then, the operator norms of elements of the set $\{\lambda \in
(C_{0:T})^{\ast }\mid \mathfrak{P}^{\ast }(\lambda )\leq \Xi \}$ are
uniformly bounded, implying that $\{\lambda \in (C_{0:T})^{\ast }\mid 
\mathfrak{P}^{\ast }(\lambda )\leq \Xi \}$ is contained in some (weak$^{\ast
}$ compact, by Banach Alaoglu Theorem, and sequentially compact by
separability of $C_{0:T}$, see \cite{Aliprantis} Theorem 6.30) ball of $%
(C_{0:T})^{\ast }$. Since $\mathfrak{P}^{\ast }$ is weak$^{\ast }$ lower
semicontinuous by its own very definition, its sublevel sets are weak$^{\ast
}$ closed. This concludes the proof of weak$^{\ast }$-(sequential)
compactness of $\{\lambda \in (C_{0:T})^{\ast }\mid \mathfrak{P}^{\ast
}(\lambda )\leq \Xi \}$.
\end{proof}

\begin{lemma}
\label{claimconclusion} Under Assumption \ref{asscallstozero}, for every
lower semicontinuous $c:\Omega \rightarrow (-\infty ,+\infty ]$ satisfying %
\eqref{controlfrombelow}, the duality \eqref{dualrepreabstract} holds and,
if $\mathfrak{P}(c)<+\infty $, the infimum in \eqref{dualrepreabstract} is a
minimum.
\end{lemma}

\begin{proof}
Take $c$ as in the statement. Observe that, from the definition of $%
\mathfrak{P}$ and the Fenchel inequality on $S^{U},$ for any $Q\in \mathrm{%
Prob}^{1}(\Omega )\cap \mathcal{A}^{\circ }$ we have 
\begin{align*}
\mathfrak{P}(c)& =\sup_{z\in -\mathcal{A}}\sup_{\varphi \in \mathbf{\Phi }%
_{z}(c)}S^{U}\left( \varphi \right) \leq \sup_{z\in -\mathcal{A}%
}\sup_{\varphi \in \mathbf{\Phi }_{z}(c)}\left( (S^{U})^{\ast }({Q}%
_{0},\dots ,Q_{T})+E_{Q}\left[ \sum_{t=0}^{T}\varphi _{t}\right] \right) \\
& \overset{{Q}\in \mathrm{Prob}^{1}(\Omega )\cap \mathcal{A}^{0}}{\leq }%
\sup_{z\in -\mathcal{A}}\sup_{\varphi \in \mathbf{\Phi }_{z}(c)}\left(
(S^{U})^{\ast }({Q}_{0},\dots ,Q_{T})+E_{Q}\left[ \sum_{t=0}^{T}\varphi
_{t}+z\right] \right) \\
& \overset{\text{L.}\ref{propformaconiugata}}{\leq }\sup_{z\in -\mathcal{A}%
}\sup_{\varphi \in \mathbf{\Phi }_{z}(c)}\left( E_{Q}\left[ c\right] +%
\mathcal{D}({Q})\right) \\
& =E_{Q}\left[ c\right] +\mathcal{D}({Q}).
\end{align*}%
Hence:%
\begin{equation}
\mathfrak{P}(c)\leq \inf_{Q\in \mathrm{Prob}^{1}(\Omega )\cap \mathcal{A}%
^{\circ }}E_{Q}\left[ c\right] +\mathcal{D}({Q})  \label{PPPinfty}
\end{equation}%
The case $\mathfrak{P}(c)=+\infty $ is thus trivial and we now focus on the
case $\mathfrak{P}(c)<+\infty $. Let $c^{A}(x):=-A\left(
1+\sum_{t=0}^{T}\sum_{j=1}^{d}\left\vert x_{t}^{j}\right\vert \right) ,x\in
\Omega .$ Then $c\geq c^{A}\in C_{0:T}$ and $\mathfrak{P}(c^{A})\leq 
\mathfrak{P}(c)<+\infty $, as can be easily verified.

%
%

A standard argument produces a sequence $(c_{n})_{n}\subseteq C_{0:T}$ with $%
c_{n}\uparrow _{n}c$ pointwise on $\Omega $. We claim that, given a sequence
of optima for the dual problems of $\mathfrak{P}(c_{n})$, taking a suitable
converging subsequence the limit $\widehat{Q}$ satisfies $\widehat{Q}\in 
\mathrm{Prob}^{1}(\Omega )\cap \mathcal{A}^{\circ }$ and $E_{\widehat{Q}}%
\left[ c\right] +\mathcal{D}(\widehat{Q})\leq \mathfrak{P}(c)$. This and (%
\ref{PPPinfty}) will then imply \eqref{dualrepreabstract}.

To prove the claim, recall from Lemma \ref{claimdualpartial} and $+\infty >%
\mathfrak{P}(c)\geq \mathfrak{P}(c_{n})$, that each dual problem for $%
\mathfrak{P}(c_{n})$ admits an optimum, call it $Q^{n}\in \mathrm{Prob}%
^{1}(\Omega )\cap \mathcal{A}^{\circ }$. We proceed observing that $\mathcal{D}(\probq^n)\in\R$ for every $n$ and
\begin{equation}
\mathfrak{P}(c_{n})=E_{Q^{n}}[c_{n}]+\mathcal{D}(Q^{n})\geq -E_{Q^{n}}\left[
\left\Vert c_{1}\right\Vert _{0:T}\left( 1+\sum_{t=0}^{T}\eta _{t}\right) %
\right] +\mathcal{D}(Q^{n})  \label{goingtosublevels}
\end{equation}%
where we set $\eta _{t}(x_{t})=\sum_{j=1}^{d}\left\vert x_{t}^{j}\right\vert
,x_{t}\in K_{t}$. Now by Fenchel inequality \eqref{fenchelgeneral},
mimicking the argument in \eqref{56}, 
\begin{equation*}
E_{Q^{n}}\left[ \left\Vert c_{1}\right\Vert _{0:T}\sum_{t=0}^{T}\eta _{t}%
\right] \leq \frac{1}{2}\mathcal{D}(Q^{n})+\frac{1}{2}V\left( 2\left\Vert
c_{1}\right\Vert _{0:T}\eta _{0},\dots ,2\left\Vert c_{1}\right\Vert
_{0:T}\eta _{T}\right) \,.
\end{equation*}%
Going back to \eqref{goingtosublevels} we then get%
\begin{equation}  \label{goingtosublevels2}
\mathfrak{P}(c_{n})\geq \zeta +\frac{1}{2}\mathcal{D}(Q^{n})
\end{equation}%
where $\zeta \in {\mathbb{R}}$ is a constant depending on $c_{1},V,\eta
_{0},\dots ,\eta _{T}$. Since now $\mathfrak{P}(c_{n})\leq \mathfrak{P}%
(c)<+\infty $ we conclude that $\sup_{n}\mathcal{D}(Q^{n})<+\infty $, which
in turns implies that the sequence $(Q_{n})_{n}$ lies in $\{Q\in \mathrm{Prob%
}^{1}(\Omega )\cap \mathcal{A}^{\circ }\mid \mathcal{D}(Q)\leq \Xi \}$ for $%
\Xi \in {\mathbb{R}}$ big enough. We know that the latter set is weak$^{\ast
}$ sequentially compact by Lemma \ref{corcpt}, thus we can extract a weak$%
^{\ast }$ converging subsequence, which we rename again $(Q^{n})_{n}$, say
converging to a $\widehat{Q}\in \mathrm{Prob}^{1}\cap \mathcal{A}^{\circ }$.
Now it is easily seen that 
\begin{align*}
E_{\widehat{Q}}\left[ c\right] +\mathcal{D}(\widehat{Q})& =\lim_{n}E_{%
\widehat{Q}}\left[ c_{n}\right] +\mathcal{D}(\widehat{Q}) \\
& \overset{(\star )}{\leq }\lim_{n}\liminf_{m}\left( E_{Q^{m}}[c_{n}]+%
\mathcal{D}(Q^{m})\right) \\
& \overset{(\star \star )}{\leq }\lim_{n}\liminf_{m}\left( E_{Q^{m}}[c_{m}]+%
\mathcal{D}(Q^{m})\right) \\
& =\liminf_{m}\left( E_{Q^{m}}[c_{m}]+\mathcal{D}(Q^{m})\right) =\lim_{m}%
\mathfrak{P}(c^{m})\leq \mathfrak{P}(c)
\end{align*}%
where in $(\star )$ we exploited the fact that $Q\mapsto E_{Q}[c_{n}]\,%
\mathrm{d}Q+\mathcal{D}(Q)$ is weak$^{\ast }$ lower semicontinuous being sum
of weak$^{\ast }$ lower semicontinuous functionals, and in $(\star \star )$
we used the fact that $c_{n}\leq c_{m}$ if $m\geq n$.
\end{proof}
\subsection{Convergence of EMOT}

\label{Secstab} In this Section \ref{Secstab} we study some stability and
convergence results for the EMOT problem. In particular, we show how under
suitable convergence assumptions on the penalty terms, one can see the
classical MOT as a limit case for EMOT.

We suppose that for each $n\in \mathbb{N}\cup \{\infty \} $ we are given a
functional $U_{n}$ and a set $\mathcal{A}_{n}\subseteq C_{0:T}$. We denote
the corresponding problem as in \eqref{deffrakPgeneral} by $\mathfrak{P}%
_{n}(c)$.

\begin{proposition}
\label{propstability} Suppose that, for each $n\in \mathbb{N}\cup \{\infty
\},$ the same assumptions of Theorem \ref{mainEMOTtheoremgeneral} hold for $%
\mathfrak{P}_{n}(c)$ and that $\mathfrak{P}_{n}(c)<+\infty $. Suppose that 
\begin{equation}
\begin{split}
& \mathcal{D}_{\infty }(Q)+\sigma _{\mathcal{A_{\infty }}}(Q)=\sup_{n\in 
\mathbb{N}}\left( \mathcal{D}_{n}(Q)+\sigma _{\mathcal{A}_{n}}(Q)\right) \\
& \mathcal{D}_{n+1}(Q)+\sigma _{\mathcal{A}_{n+1}}(Q)\geq \mathcal{D}%
_{n}(Q)+\sigma _{\mathcal{A}_{n}}(Q),\quad n\in \mathbb{N},
\end{split}
\label{nondecreasing}
\end{equation}%
for every $Q\in \mathrm{Prob}^{1}(\Omega )$. 
Then $\mathfrak{P}_{n}(c)\uparrow _{n}\mathfrak{P}_{\infty }(c)$ for every $%
c:\Omega \rightarrow (-\infty ,+\infty ]$ which is lower semicontinuous and
satisfies \eqref{controlfrombelow}.
\end{proposition}

\begin{proof}
From Lemma \ref{claimdualpartial} we see that each dual problem for $%
\mathfrak{P}(c_{n})$ admits an optimum, call it $Q^{n}\in \mathrm{Prob}%
^{1}(\Omega )\cap \mathcal{A}_{n}^{\circ }$. for each $n\in \mathbb{N}$.
Observe that $+\infty>\mathfrak{P}_{\infty }(c)\geq \sup_{n}\mathfrak{P}%
_{n}(c)=\lim_{n}\mathfrak{P}_{n}(c)$ and that, with an argument similar to
the one yielding \eqref{goingtosublevels2},%
\begin{align*}
\mathfrak{P}_{n}(c)& =E_{Q^{n}}[c]+\mathcal{D}_{n}(Q^{n})+\sigma _{\mathcal{A%
}_{n}}(Q^{n})\geq E_{Q^{n}}\left[ -\left\Vert c\right\Vert
_{0:T}\sum_{t=0}^{T}\psi _{t}\right] +\mathcal{D}_{n}(Q^{n})+\sigma _{%
\mathcal{A}_{n}}(Q^{n}) \\
& \geq -\frac{1}{2}\mathfrak{P}_{n}^{\ast }(Q^n)+\frac{1}{2}\mathfrak{P}%
_{n}\left( -2\left\Vert c\right\Vert _{0:T}\sum_{t=0}^{T}\psi _{t}\right) +%
\mathcal{D}_{n}(Q^{n})+\sigma _{\mathcal{A}_{n}}(Q^{n}) \\
& \overset{L.\ref{propformaconiugata}}{= }\frac{1}{2}\mathfrak{P}%
_{n}\left( -2\left\Vert c\right\Vert _{0:T}\sum_{t=0}^{T}\psi _{t}\right) +%
\frac{1}{2}\mathcal{D}_{n}(Q^{n})+\sigma _{\mathcal{A}_{n}}(Q^{n}) \\
& \geq \frac{1}{2}\mathfrak{P}_{1}\left( -2\left\Vert c\right\Vert
_{0:T}\sum_{t=0}^{T}\psi _{t}\right) +\frac{1}{2}\mathcal{D}%
_{1}(Q^{n})+\sigma _{\mathcal{A}_{1}}(Q^{n}).
\end{align*}%
As a consequence, for some constant $\eta $,

\begin{equation*}
+\infty >\mathfrak{P}_{\infty }(c)\geq \mathfrak{P}_{n}(c)\geq \eta +\frac{1%
}{2}\mathcal{D}_{1}(Q^{n})+\sigma _{\mathcal{A}_{1}}(Q^{n})
\end{equation*}

%
Hence, all the measures $(Q_{n})_{n}$ belong to the a sublevel in the form $%
\{Q\in \mathrm{Prob}^{1}(\Omega )\cap (\mathcal{A}_{1}\mathcal{)}^{\circ
}\mid \mathcal{D}_{1}(Q)\leq \Xi \}$ which is $\sigma (\mathrm{Prob}%
^{1}(\Omega ),C_{0:T})$-(sequentially) compact by Lemma \ref{corcpt}.
Extract a subsequence, which will be renamed again $(Q_{n})_{n}$, converging
to the limit $Q^{\infty }\in \mathrm{Prob}^{1}(\Omega )$. Since $\mathcal{D}_{n},\sigma _{\mathcal{A}_{n}}$ are lower
semicontinuous, so is $\mathcal{D}_{n}+\sigma _{\mathcal{A}_{n}}$ for every $%
n\in \mathbb{N}\cap \{\infty \}$. Hence 
\begin{align*}
\mathcal{D}_{\infty }(Q^{\infty })+\sigma _{\mathcal{A}_{\infty }}(Q^{\infty
})& \overset{\eqref{nondecreasing}}{=}\sup_{K}\left( \mathcal{D}%
_{K}(Q^{\infty })+\sigma _{\mathcal{A}_{K}}(Q^{\infty })\right) \leq
\sup_{K}\liminf_{n}\left( \mathcal{D}_{K}(Q^{n})+\sigma _{\mathcal{A}%
_{K}}(Q^{n})\right) \\
& \overset{\eqref{nondecreasing}}{\leq }\sup_{K}\liminf_{n}\left( \mathcal{D}%
_{n}(Q^{n})+\sigma _{\mathcal{A}_{n}}(Q^{n})\right) =\liminf_{n}\left( 
\mathcal{D}_{n}(Q^{n})+\sigma _{\mathcal{A}_{n}}(Q^{n})\right) \,.
\end{align*}%
Up to taking a further subsequence, again renamed $(Q_{n})_{n}$ we might as
well assume that the $\liminf $ above is in fact a limit, so that 
\begin{equation*}
\mathcal{D}_{\infty }(Q^{\infty })+\sigma _{\mathcal{A}_{\infty }}(Q^{\infty
})\leq \lim_{n}\left( \mathcal{D}_{n}(Q^{n})+\sigma _{\mathcal{A}%
_{n}}(Q^{n})\right) \,.
\end{equation*}%
Since now $c:\Omega \rightarrow (-\infty ,+\infty ]$ is lower semicontinuous
and satisfies \eqref{controlfrombelow} for some $A\geq 0$, there exists a
sequence $(c_{n})_{n}\subseteq C_{0:T}$ with $c_{n}\uparrow _{n}c$ pointwise
on $\Omega $, just as in the proof of Lemma \ref{claimconclusion}. Notice
that by Monotone Convergence Theorem we then have that $E_{Q}\left[ c\right]
=\sup_{n}E_{Q}\left[ c_{n}\right] $. We conclude that $Q\mapsto E_{Q}\left[ c%
\right] $ is the supremum of linear functional, each being continuous w.r.t
the topology induced by $\sigma ((C_{0:T})^{\ast },C_{0:T})$ on $\mathrm{Prob%
}^{1}(\Omega )$. Then, $Q\mapsto E_{Q}\left[ c\right] $ is lower
semicontinuous w.r.t. such topology and $E_{Q^{\infty }}[c]\leq
\liminf_{n}E_{Q^{n}}[c]$. Passing again to a further subsequence we can
assume that $\liminf_{n}E_{Q^{n}}[c]=\lim_{n}E_{Q^{n}}[c]$. From the
previous arguments we then get 
\begin{align*}
\mathfrak{P}_{\infty }(c)& \leq E_{Q^{\infty }}[c]+\mathcal{D}_{\infty
}(Q^{\infty })+\sigma _{\mathcal{A}_{\infty }}(Q^{\infty })\leq
\lim_{n}E_{Q^{n}}[c]+\mathcal{D}_{\infty }(Q^{\infty })+\sigma _{\mathcal{A}%
_{\infty }}(Q^{\infty }) \\
& \leq \lim_{n}E_{Q^{n}}[c]+\lim_{n}\mathcal{D}_{n}(Q^{n})+\sigma _{\mathcal{%
A}_{n}}(Q^{n})=\lim_{n}\left( E_{Q^{n}}[c]+\mathcal{D}_{n}(Q^{n})+\sigma _{%
\mathcal{A}_{n}}(Q^{n})\right) =\lim_{n}\mathfrak{P}_{n}(c)\,.
\end{align*}%
where we exploited the fact that $(Q^{n})_{n}$ are optima. Since we already
know $\lim_{n}\mathfrak{P}_{n}(c)\leq \mathfrak{P}_{\infty }(c)$ this
concludes the proof of $\mathfrak{P}_{n}(c)\uparrow _{n}\mathfrak{P}_{\infty
}(c)$. 
\end{proof}
\section{Additive structure}

\label{secadditive}

In Section \ref{EMOT new}, we did not require any particular structural form
of the functionals $\mathcal{D},U$. 
Here instead, we will assume an additive structure of $U$ and,
complementarily, an additive structure of $\mathcal{D}$. In the whole
Section \ref{secadditive} we take for each $t=0,\dots ,T$ a vector subspace $%
\mathcal{E}_{t}\subseteq \mathcal{C}_{t}$ such that $\mathcal{E}_{t}+{%
\mathbb{R}}=\mathcal{E}_{t}$ and set $\mathcal{E}=\mathcal{E}_{0}\times
\dots \times \mathcal{E}_{T}$. Observe that we automatically have $\mathcal{E%
}+{\mathbb{R}}^{T+1}=\mathcal{E}$. It is also clear that $\mathcal{E}$ is a
subspace of $(C_{0:T})^{T+1},$ if we interpret $\mathcal{E}_{0},\dots ,%
\mathcal{E}_{T}$ as subspaces of $\mathcal{C}_{0:T}$.
We also mention here that up to now we used for a $\lambda\in(C_{0:T})^*$ (resp.  for a measure $\mu\in\mathrm{ca(\Omega)}$) , the notation $\lambda_t$ (resp. $\mu_t$) for restrictions to $C_{0:t}$ (resp. marginals on $\Omega_{0:t}$). This was motivated by the fact that we were considering general $\mathcal{E}_t\subseteq C_{0:t}$. Since from now on we will mostly work with $\mathcal{E}_t\subseteq C_t$, we change notation slightly.
\begin{notation} For the whole Sections \ref{secadditive}, \ref{secpartcases},  \ref{sectionexamplesfornoncpaddition}, given a $\lambda\in(C_{0:T})^*$ (resp.  given a measure $\mu\in\mathrm{ca(\Omega)}$), we use the notation $\lambda_t$ (resp. $\mu_t$) for restrictions to $C_{t}$ (resp. marginals on $\bigtimes_{j=1}^d K_t^j$).
\end{notation}
\subsection{Additive structure of $U$}

\label{secaddU}

\begin{setup}
\label{SetupsigmaU} For every $t=0,\dots ,T$ we consider a proper concave
functional $U_{t}:\mathcal{E}_t\rightarrow \lbrack -\infty ,+\infty )$. We
define $\mathcal{D}_{t}$ on $\mathrm{ca}^1(K_{t})$ similarly to (\ref%
{dualrepdivgenabstract}) as 
\begin{equation*}
\mathcal{D}_t(\gamma_t):=\sup_{\varphi_t\in\mathcal{E}_t}\left(U_t(%
\varphi_t)-\int_{K_t}\varphi_t\,\mathrm{d}\gamma_t\right)\,,\,\,\,\,\,\,%
\gamma_t\in\mathrm{ca}^1(K_t)
\end{equation*}
and observe that $\mathcal{D}_{t}$ can also be thought to be defined on $%
\mathrm{ca}^1(\Omega )$ using for $\gamma \in \mathrm{ca}^1(\Omega )$ the
marginals $\gamma _{0},...,\gamma _{T} $ and setting $\mathcal{D}_{t}(\gamma
):=\mathcal{D}_{t}(\gamma _{t})$. We may now define, for each $\varphi\in%
\mathcal{E}$, $U(\varphi ):=\sum_{t=0}^{T}U_{t}(\varphi _{t})$ and define $%
\mathcal{D}$ on $\mathrm{ca}^1(\Omega )$ using \eqref{dualrepdivgenabstract}%
. Recall from \eqref{defSU} 
\begin{equation*}
S^{U}(\varphi ):=\sup_{\beta \in {\mathbb{R}}^{T+1}}\left( U(\varphi +\beta
)-\sum_{t=0}^{T}\beta _{t}\right) \text{,}\,\,\varphi \in \mathcal{E},\text{%
\quad }S^{U_{t}}(\varphi _{t}):=\sup_{\alpha \in {\mathbb{R}}}\left(
U(\varphi _{t}+\alpha )-\alpha \right) ,\,\ \varphi _{t}\in \mathcal{E}_{t}.
\label{defSU1}
\end{equation*}
\end{setup}

\begin{lemma}
\label{conjugateaddU} In Setup \ref{SetupsigmaU} and under the convention $%
+\infty -\infty =-\infty $ we have 
\begin{equation}
\mathcal{D}(\gamma )=\sum_{t=0}^{T}\mathcal{D}_{t}(\gamma )=\sum_{t=0}^{T}%
\mathcal{D}_{t}(\gamma _{t}),\text{\quad }\,\forall \,\gamma \in \mathrm{ca}%
^{1}(\Omega ),\,\,\,\,\,\,\,\,S^{U}(\varphi
)=\sum_{t=0}^{T}S^{U_{t}}(\varphi _{t})\text{ for all }\,\varphi \in 
\mathcal{E},  \label{additivestr}
\end{equation}%
and for all $\varphi \in \mathcal{E}$ 
\begin{equation*}
S^{U}(\varphi +\beta )=S^{U}(\varphi )+\sum_{t=0}^{T}\beta _{t}\text{, for }%
\beta \in {\mathbb{R}}^{T+1}\text{,\quad }S^{U_{t}}(\varphi _{t}+\beta
)=S^{U_{t}}(\varphi _{t})+\beta ,\text{ for }\beta \in {\mathbb{R}}\,.
\end{equation*}
\end{lemma}

\begin{proof}
We will only focus on \eqref{additivestr}, since the remaining claims are
easily checked. We have that 
\begin{align*}
\mathcal{D}(\gamma )& =\sup_{\varphi \in \mathcal{E}}\left(
\sum_{t=0}^{T}U_{t}(\varphi _{t})-\sum_{t=0}^{T}\int_{K_{t}}\varphi _{t}\,%
\mathrm{d}\gamma \right) =\sum_{t=0}^{T}\sup_{\varphi _{t}\in \mathcal{E}%
_{t}}\left( U_{t}(\varphi _{t})-\int_{K_{t}}\varphi _{t}\,\mathrm{d}\gamma_t
\right) \\
& =\sum_{t=0}^{T}\mathcal{D}_{t}(\gamma _{t})=\sum_{t=0}^{T}\mathcal{D}%
_{t}(\gamma )\,.
\end{align*}%
As to the second claim in \eqref{additivestr}, in view of $%
+\infty-\infty=-\infty$, we observe that 
\begin{equation*}
\sup_{\beta \in {\mathbb{R}}^{T+1}}\left( U(\varphi +\beta )-\sum_{t=0}^{T}\beta
_{t}\right) =\sum_{t=0}^{T}\sup_{\beta \in {\mathbb{R}}}\left( U_{t}(\varphi
_{t}+\beta )-\beta \right) =\sum_{t=0}^{T}S^{U_{t}}(\varphi _{t}).
\end{equation*}
\end{proof}


\subsection{Duality for the general Cash Additive setup}

As a consequence of Theorem \ref{mainEMOTtheoremgeneral}, we are now ready
to prove the duality $\mathfrak{D}(c)=\mathfrak{P}(c)$ announced in the
Introduction, in equation \eqref{dualityPD}.

\begin{theorem}
\label{mainthm3noncpt} Suppose that $\mathcal{E}_{t} \subseteq C_t$ with $X_{t}\in \mathcal{E}_{t}$ and that $S_{t}:\mathcal{E}_{t}\rightarrow 
{\mathbb{R}}$ is a concave, cash additive functional null in $0$.  Set $U(\varphi
):=\sum_{t=0}^{T}S_{t}(\varphi _{t}),$ for $\varphi \in \mathcal{E}=\mathcal{E}_{0}\times \dots \times \mathcal{E}_{T}$ and  suppose
that Assumption \ref{asscallstozero} is fulfilled.
Consider for every $t=0,\dots ,T$ the penalizations 
\begin{equation*}
\mathcal{D}_{t}(Q_{t}):=\sup_{\varphi _{t}\in \mathcal{E}_{t}}\left(
S_{t}(\varphi _{t})-\int_{K_{t}}\varphi _{t}\,\mathrm{d}Q_{t}\right) \,\,\,%
\text{ for }Q_{t}\in \mathrm{Prob}^{1}(K_{t}).
\end{equation*}%
Let $c:\Omega \rightarrow (-\infty ,+\infty ]$ be lower semicontinuous and
such that \eqref{controlfrombelow} holds. Then 
\begin{equation}
\begin{split}
\mathfrak{P}(c)& =\sup \left\{ \sum_{t=0}^{T}S_{t}(\varphi _{t})\mid \varphi
\in \mathcal{E}\text{ is s.t. }\exists \Delta \in \mathcal{H}\text{ with }%
\sum_{t=0}^{T}\varphi _{t}(x_{t})+I^{\Delta }(x)\leq c(x)\text{ }\forall
x\in \Omega \right\} \\
& =\inf_{Q\in \mathrm{Mart}(\Omega )}\left( E_{Q}\left[ c\right]
+\sum_{t=0}^{T}\mathcal{D}_{t}(Q_{t})\right)
\end{split}
\label{dualrepreabstractnoncpt}
\end{equation}%
and the infimum in \eqref{dualrepreabstractnoncpt} is a minimum provided
that $\mathfrak{P}(c)<+\infty $.
\end{theorem}

\begin{proof}
Let $\mathcal{D}$ be defined as in \eqref{dualrepdivgenabstract}.
Observe that we are in Setup \ref{SetupsigmaU}. Lemma \ref{conjugateaddU}
tells us that $S^{U}(\varphi )=\sum_{t=0}^{T}S_{t}(\varphi _{t})$, since $%
S_{0},\dots ,S_{T}$ are cash additive, and that $\mathcal{D}$ coincides on $%
\mathrm{Mart}(\Omega )$ with the penalization term $Q\mapsto \sum_{t=0}^{T}%
\mathcal{D}_{t}(Q_{t})$, as provided in the statement of this Theorem. All
the assumptions of Theorem \ref{mainEMOTtheoremgeneral} are fulfilled, so
that we can apply Corollary \ref{maincor}, which yields exactly $\mathfrak{D}%
(c)=\mathfrak{P}(c)$.
\end{proof}

\subsection{Additive structure of $\mathcal{D}.$}

The results of this subsection will be applied in Subsections \ref%
{CasePenalization1gen} and \ref{secwasser}. In the spirit of Remark \ref{remfenchelmoreau}, we
may now reverse the procedure taken in the previous subsection: we start
from some functionals $\mathcal{D}_{t}$ on $\mathrm{ca}^1(K_{t}),$ for $%
t=0,\dots ,T$, and build an additive functional $\mathcal{D}$ on $\mathrm{ca}%
^1(\Omega ).$ Our aim is to find the counterparts of the results in Section %
\ref{secaddU}. \label{additiveD}

\begin{setup}
\label{SetupsigmaD} For every $t=0,\dots ,T$ we consider a proper, convex, $%
\sigma (\mathrm{ca}^1(K_{t}),\mathcal{E}_{t})$-lower semicontinuous
functional $\mathcal{D}_{t}:\mathrm{ca}^1(K_{t})\rightarrow (-\infty
,+\infty ]$. We can then extend the functionals $\mathcal{D}_{t}$ to $%
\mathrm{ca}^1(\Omega )$ by using, for any $\gamma \in \mathrm{ca}(\Omega )$,
the marginals $\gamma _{0},\dots ,\gamma _{T}$. If $\gamma \in \mathrm{ca}%
^1(\Omega )$, we set 
\begin{equation*}
\mathcal{D}_{t}(\gamma ):=\mathcal{D}_{t}(\gamma
_{t})\,\,\,\,\,\,\,\,\,\,\,\,\,\text{ and }\,\,\,\,\,\,\,\,\,\,\,\,\,%
\mathcal{D}(\gamma ):=\sum_{t=0}^{T}\mathcal{D}_{t}(\gamma ) =\sum_{t=0}^{T}%
\mathcal{D}_{t}(\gamma _{t})\,.  \label{extensionfromKtoOmega}
\end{equation*}

We define $V(\varphi)$ for $\varphi\in\mathcal{E}$ and $V_{t}(\varphi_t)$
for $\varphi_t\in\mathcal{E}_{t},$ for $t=0,\dots ,T$ similarly to %
\eqref{fenchelmoreau}, as 
\begin{equation*}
V(\varphi):=\sup_{\gamma\in\mathrm{ca}^1(\Omega)}\left(\int_{\Omega}\left(%
\sum_{t=0}^T\varphi_t\right)\,\mathrm{d}\gamma-\mathcal{D}(\gamma)\right)%
\text{ and }V_t(\varphi_t):=\sup_{\gamma\in\mathrm{ca}^1(K_t)}\left(%
\int_{K_t}\varphi_t\,\mathrm{d}\gamma-\mathcal{D}_t(\gamma)\right)\,.
\end{equation*}

We define on $\mathcal{E}$ the functional $U(\cdot )=-V(-\cdot )$, as in %
\eqref{eqdefU}, and similarly $U_{t}(\cdot )=-V_{t}(-\cdot )$ on $\mathcal{E}%
_{t},$ for $t=0,\dots ,T$. Finally, $S^{U}(\varphi ),$ $S^{U_{0}}(\varphi
_{0}),\dots ,S^{U_{T}}(\varphi _{T})$ are defined as in Setup \ref%
{SetupsigmaU}. 
\end{setup}

\begin{lemma}
\label{sumislsc} In Setup \ref{SetupsigmaD} we have:

\begin{enumerate}
\item $\mathcal{D}_{0},\dots ,\mathcal{D}_{T}$, as well as $\mathcal{D}$,
are $\sigma (\mathrm{ca}^1(\Omega ),\mathcal{E})$-lower semicontinuous.

\item Under the additional assumption that $\mathrm{dom}(\mathcal{D}%
_t)\subseteq\mathrm{Prob}^1(K_t)$ for every $t=0,\dots,T$, for any $\varphi
=[\varphi _{0},\dots ,\varphi _{T}]\in \mathcal{E}_{0}\times \dots \times 
\mathcal{E}_{T}$ 
\begin{align}
U(\varphi ) &=\sum_{t=0}^{T}U_{t}(\varphi
_{t})=\sum_{t=0}^{T}-V_{t}(-\varphi _{t})\,,  \label{Usum} \\
S^{U}(\varphi ) &=\sum_{t=0}^{T}S^{U_{t}}(\varphi _{t})\,.  \label{Ssum}
\end{align}
\end{enumerate}
\end{lemma}

\begin{proof}
$\,$

\textbf{Item 1}. For each $t=0,\dots ,T$ $\mathcal{D}_{t}(\gamma )=\mathcal{D%
}_{t}\circ \pi _{t}(\gamma ),$ where $\mathcal{D}_{t}$ is $\sigma (\mathrm{ca%
}^1(K_{t}),\mathcal{E}_{t})$-lower semicontinuous and $\pi _{t}$, the
projection to the $t$-th marginal, is $\sigma (\mathrm{ca}^1(\Omega ),%
\mathcal{E})-\sigma (\mathrm{ca}^1(K_{t}),\mathcal{E}_{t})$ continuous.
Hence, for each $t=0,\dots ,T$ $\gamma \mapsto \mathcal{D}_{t}(\gamma )$ is $%
\sigma (\mathrm{ca}^1(\Omega ),\mathcal{E})$-lower semicontinuous. Lower
semicontinuity of $\mathcal{D}$ is then a consequence of the fact that the
sum of lower semicontinuous functions is lower semicontinuous.

\textbf{Item 2, equation (\ref{Usum})}. We have that for $\psi =-\varphi $ 
\begin{align*}
& -U(\varphi )=V(\psi )=\sup_{\gamma \in \mathrm{ca}^{1}(\Omega )}\left(
\int_{\Omega }\left( \sum_{t=0}^{T}\psi _{t}\right) \mathrm{d}\gamma -%
\mathcal{D}(\mu )\right) =\sup_{\gamma \in \mathrm{ca}^{1}(\Omega
)}\sum_{t=0}^{T}\left( \int_{K_{t}}\psi _{t}\mathrm{d}\gamma -\mathcal{D}%
_{t}(\mu _{t})\right)  \\
& \overset{(i)}{=}\sup \left\{ \sum_{t=0}^{T}\left( \int_{K_{t}}\psi _{t}%
\mathrm{d}\gamma _{t}-\mathcal{D}_{t}(\gamma _{t})\right) \mid \gamma \in 
\mathrm{ca}^{1}(\Omega )\text{ with }\gamma _{t}\in \mathrm{Prob}%
^{1}(K_{t})\,\,\forall t=0,\dots ,T\right\}  \\
& \overset{(ii)}{=}\sup \left\{ \sum_{t=0}^{T}\left( \int_{K_{t}}\psi _{t}%
\mathrm{d}Q_{t}-\mathcal{D}_{t}(Q_{t})\right) \mid \lbrack Q_{0},\dots
,Q_{T}]\in \mathrm{Prob}^{1}(K_{0})\times \dots \times \mathrm{Prob}%
^{1}(K_{T})\right\}  \\
& =\sum_{t=0}^{T}\sup_{Q_{t}\in \mathrm{Prob}^{1}(K_{t})}\left(
\int_{K_{t}}\psi _{t}\mathrm{d}Q_{t}-\mathcal{D}_{t}(Q_{t})\right) \overset{%
(iii)}{=}\sum_{t=0}^{T}\sup_{\gamma _{t}\in \mathrm{ca}^{1}(K_{t})}\left(
\int_{K_{t}}\psi _{t}\mathrm{d}\gamma _{t}-\mathcal{D}_{t}(\gamma
_{t})\right)  \\
& =\sum_{t=0}^{T}V_{t}(\psi _{t})=\sum_{t=0}^{T}-U_{t}(\varphi _{t})
\end{align*}%
where for $(i)$ we used 
\begin{equation*}
\mathrm{dom}(\mathcal{D})\subseteq \mathcal{Z}:=\left\{ \gamma \in \mathrm{ca%
}^{1}(\Omega )\mid \gamma _{t}\in \mathrm{Prob}^{1}(K_{t})\,\,\forall
t=0,\dots ,T\right\} .
\end{equation*}%
In $(ii)$ we applied the facts that: (1) any vector of probability measures $%
[Q_{0},\dots ,Q_{T}]$ with $Q_{t}\in \mathrm{Prob}^{1}(K_{t})$, $t=0,\dots ,T
$ identifies $\gamma :=Q_{0}\otimes \dots \otimes Q_{T}\in \mathcal{Z}$%
\footnote{%
Note that this does not hold for a general vector of (signed) measures,
which is why we need the additional assumption on the domains of the
penalization functionals for Item 2} with $\mathcal{D}(\gamma
)=\sum_{t=0}^{T}\mathcal{D}_{t}(Q_{t});$ (2) for every $\gamma \in \mathcal{Z%
}$ (setting $Q_{t}:=\gamma _{t}\in \mathrm{Prob}^{1}(K_{t}))$ we have $%
\mathcal{D}(\gamma )=\sum_{t=0}^{T}\mathcal{D}_{t}(Q_{t})$. The equality $%
(iii)$ follows from $\mathrm{dom}(\mathcal{D}_{t})\subseteq \mathrm{Prob}%
^{1}(K_{t})$ for each $t=0,\dots ,T$.

\textbf{Item 2, equation (\ref{Ssum})}. The argument is identical to the one
in the proof of Lemma \ref{conjugateaddU}, using the additive structure of $%
U $ we obtained in the previous step of the proof.
\end{proof}

\subsection{Divergences induced by utility functions}

\label{SecDivergence}In this section we provide the exact formulation of the
divergences induced by utility functions $u_{t}:{\mathbb{R}}\rightarrow
\lbrack -\infty ,+\infty )$, distinguishing the two cases: $\mathrm{dom}%
(u_{t})={\mathbb{R}}$ and $\mathrm{dom}(u_{t})\supseteq \lbrack 0,+\infty )$.

\begin{assumption}
\label{aU1} We consider concave, upper semicontinuous nondecreasing
functions $u_{0},\dots ,u_{T}:{\mathbb{R}}\rightarrow \lbrack -\infty
,+\infty )$ with $u_{0}(0)=\dots =u_{T}(0)=0$, $u_{t}(x)\leq x\,\,\forall
\,x\in {\mathbb{R}}$ (that is $1\in \partial u_{0}(0)\cap \dots \cap
\partial u_{T}(0)$). For each $t=0,\dots ,T$ we define $%
v_{t}(x):=-u_{t}(-x),\,x\in {\mathbb{R}}$ and 
\begin{equation}
v_{t}^{\ast }(y):=\sup_{x\in {\mathbb{R}}}(xy-v_{t}(x))\,=\sup_{x\in {%
\mathbb{R}}}(u_{t}(x)-xy)),\,\,\,\,\,\,y\in {\mathbb{R}}\,.
\label{cnvxconjvstar}
\end{equation}
\end{assumption}

We observe that $v_{t}(y)=v_{t}^{\ast \ast }(y)=\sup_{x\in {\mathbb{R}}%
}(xy-v_{t}^{\ast }(y))$ for all $y\in {\mathbb{R}}$ by Fenchel-Moreau
Theorem and that $v_{t}^{\ast }$ is convex, lower semicontinuous and lower
bounded on ${\mathbb{R}}$.

\begin{example}
\label{exampleofutils} Assumption \ref{aU1} is satisfied by a wide range of
functions. Just to mention a few with various peculiar features, we might
take $u_{t}$ of the following forms: $u_{t}(x)=1-\exp (-x)$, whose convex
conjugate is given by $v_{t}^{\ast }(y)=-\infty $ for $y<0$, $v_{t}^{\ast
}(0)=0$, $v_{t}^{\ast }(y)=(y\log (y)-y+1)$ for $y>0$; $u_{t}(x)=\alpha
x1_{(-\infty ,0]}(x)$ for $\alpha \geq 1$, so that $v_{t}^{\ast }(y)=+\infty 
$ for $y<0$, $v_{t}^{\ast }(y)=0$ for $y\in \lbrack 0,\alpha ]$, $%
v_{t}^{\ast }(y)=+\infty $ for $y>\alpha $; $u_{t}(x)=\log (x+1)$ for $x>-1$%
, $u_{t}(x)=-\infty $ for $x\leq -1$, so that $v_{t}^{\ast }(y)=+\infty $
for $y\leq 0$, $v_{t}^{\ast }(y)=y-\log (y)-1$ for $y>0$; $u_{t}(x)=-\infty $
for $x\leq -1$, $u_{t}(x)=\frac{x}{x+1}$ for $x>-1$ so that $v_{t}^{\ast
}(y)=-\infty $ for $y<0$, $v_{t}^{\ast }(y)=y-2\sqrt{y}+1$ for $y\geq 0$; $%
u_t(x)=-\infty$ for $x<0$, $u_t(x)=1-\exp(-x)$ for $x\geq 0$, so that $%
v^*_t(y)=+\infty$ for $y<0$, $v^*_t(y)=y\log(y)-y+1$ for $0 \leq y\leq 1$, $%
v^*_t(y)=0$ for $y>1$.
\end{example}

Fix $\widehat{\mu_{t}} \in \mathrm{Meas}(K_{t})$. We pose for $\mu \in 
\mathrm{Meas}(K_{t})$ 
\begin{equation}
\mathcal{D}_{v_{t}^{\ast },\widehat{\mu}_{t}}(\mu ):=%
\begin{cases}
\int_{K_{t}}v_{t}^{\ast }\left( \frac{\mathrm{d}\mu }{\mathrm{d}\widehat{\mu}%
_{t}}\right) \,\mathrm{d}\widehat{\mu}_{t}\,\,\, & \text{ if }\mu \ll 
\widehat{\mu}_{t} \\ 
+\infty \,\,\, & \text{ otherwise}%
\end{cases}%
\,.  \label{dvstar}
\end{equation}%
In the next two propositions, whose proofs are postponed to the Appendix \ref%
{Proofs}, we provide the dual representation of the divergence terms.

\begin{proposition}
\label{propvstar}Take $u_{0},\dots ,u_{T}$ satisfying Assumption \ref{aU1} with $\mathrm{dom}(u_{0})=\dots =\mathrm{dom}(u_{T})={\mathbb{R}}$,
consider closed (possibly noncompact) $K_0,\dots,K_T\subseteq {\mathbb{R}%
}$ and let $\widehat{\mu }_{t}\in \mathrm{Meas}(K_{t})$, $t=0,\dots ,T$. Then

%
%
%
\begin{equation}
\mathcal{D}_{v_{t}^{\ast },\widehat{\mu }_{t}}(\mu )=\sup_{\varphi _{t}\in 
\mathcal{C}_{b}(K_{t})}\left( \int_{K_{t}}\varphi _{t}(x_{t})\,\mathrm{d}\mu
(x_{t})-\int_{K_{t}}v_{t}(\varphi _{t}(x_{t}))\,\mathrm{d}\widehat{\mu }%
_{t}(x_{t})\right) \,.  \label{explicitdual}
\end{equation}%
\label{propconjugatesofranddiv}
\end{proposition}

Set: 
\begin{equation*}
(v_{t}^{\ast })_{\infty }^{\prime }:=\lim_{y\rightarrow +\infty }\frac{%
v_{t}^{\ast }(y)}{y},\,\,\,\,\,\,\,t=0,\dots ,T\,.
\end{equation*}%
As $u_t(0)=0$, $(v_{t}^{\ast })_{\infty }^{\prime }\in \lbrack 0,+\infty ]$
since $v^*_t(y)\geq u_t(0)-0\cdot y=0$. Let $\widehat{Q}_{t}\in \mathrm{Prob}%
(K_{t})$ and, for $\mu \in \mathrm{Meas}(K_{t})$, let $\mu =\mu _{a}+\mu
_{s} $ be the Lebesgue Decomposition of $\mu $ with respect to $\widehat{Q}%
_{t}$, where $\mu _{a}\ll \widehat{Q}_{t}$ and $\mu _{s}\perp \widehat{Q}%
_{t} $. Then we can define for $\mu \in \mathrm{Meas}(K_{t})$ 
\begin{equation*}
\mathcal{F}_{t}(\mu \mid \widehat{Q}_{t}):=\int_{K_{t}}v_{t}^{\ast }\left( 
\frac{\mathrm{d}\mu _{a}}{\mathrm{d}\widehat{Q}_{t}}\right) \,\mathrm{d}%
\widehat{Q}_{t}+(v_{t}^{\ast })_{\infty }^{\prime }\mu _{s}(K_{t})
\end{equation*}%
where we use the convention $\infty \times 0=0$, in case $(v_{t}^{\ast
})_{\infty }^{\prime }=+\infty ,$ $\mu _{s}(K_{t})=0$. Observe that the
restriction of $\mathcal{F}(\cdot \mid \widehat{Q}_{t})$ to $\mathrm{Meas}%
(K_{t})$ coincides with the functional in \cite{Sava} (2.35) with $%
F=v_{t}^{\ast }$, and that whenever $\mathrm{dom}(u_{t})={\mathbb{R}}$ we
have $(v_{t}^{\ast })_{\infty }^{\prime }=\lim_{y\rightarrow +\infty }\frac{%
v_{t}^{\ast }(y)}{y}=+\infty $ and $\mathcal{F}_{t}(\cdot \mid \widehat{Q}%
_{t})$ coincides with $\mathcal{D}_{v_{t}^{\ast },\widehat{Q}_{t}}(\cdot )$
(see \eqref{dvstar}) on $\mathrm{Meas}(K_{t})$.

\begin{proposition}
\label{PPP}Suppose that $u_{0},\dots ,u_{T}:{\mathbb{R}}\rightarrow \mathbb{%
[-\infty },+\infty )$ satisfy Assumption \ref{aU1}, and assume that $K_0,\dots,K_T\subseteq{\mathbb{R}}$
are compact. If $\widehat{Q}_{t}\in \mathrm{Prob}(K_{t}),$ $t\in \{0,\dots
,T\}$, has full support then

\begin{equation}
\mathcal{F}_{t}(\mu \mid \widehat{Q}_{t})=\sup_{\varphi _{t}\in \mathcal{C}%
_{b}(K_{t})}\left( \int_{K_{t}}\varphi _{t}(x_{t})\,\mathrm{d}\mu
(x_{t})-\int_{K_{t}}v_{t}(\varphi _{t}(x_{t}))\,\mathrm{d}\widehat{Q }%
_{t}(x_{t})\right) \,.  \label{eqfullcb}
\end{equation}

\end{proposition}

\begin{example}
\label{counterexample1} The requirement that $\widehat{Q}_{0},\dots ,%
\widehat{Q}_{T}$ have full support is crucial for the proof of Proposition %
\ref{PPP}. We provide a simple example to the fact that \eqref{eqfullcb}
does not hold in general when such an assumption is not fulfilled. To this
end, take $K=\{-2,0,2\}$, $\widehat{Q}=\frac{1}{2}\delta _{\{-2\}}+\frac{1}{2%
}\delta _{\{+2\}}$, $\mu =\delta _{\{0\}}$, $u(x):=\frac{x}{x+1}$ for $x\geq
-1$ and $u(x)=-\infty $ for $x<-1$. It is easy to see that the associated $%
v^{\ast }$ via \eqref{cnvxconjvstar} is defined by $v^{\ast }(y)=1+y-2\sqrt{y%
}$ for $y\geq 0$ and $v^{\ast }(y)=-\infty $ for $y<0$, so that $%
(v_{t}^{\ast })_{\infty }^{\prime }=1.$ It is also easy to see that $\mu
\perp \widehat{Q}$, hence in the Lebesgue decomposition with respect to $%
\widehat{Q},$ $\mu _{a}=0$ and $\mu _{s}=\mu $. Hence $\mathcal{F}(\mu \mid 
\widehat{Q})=1+1\mu (K)=2$. At the same time we see that taking $\varphi
_{N}\in \mathcal{C}_{b}(K)$ defined via $\varphi _{N}(-2)=\varphi
_{N}(2)=0,\varphi _{N}(0)=-N$ (observe that for $N$ sufficiently large $%
u(\varphi _{N})\notin \mathcal{C}_{b}(K)$) we have 
\begin{align*}
&\sup_{\varphi \in \mathcal{C}_{b}(K)}\left( \int_{K}\varphi \,\mathrm{d}\mu
-\int_{K}v(\varphi )\,\mathrm{d}\widehat{Q }\right) =\sup_{\varphi \in 
\mathcal{C}_{b}(K)}\left( \int_{K}u(\varphi )\,\mathrm{d}\widehat{Q}%
-\int_{K}\varphi \,\mathrm{d}\mu \right) \\
& \geq \sup_{N}\left( \int_{K}u(\varphi _{N})\,\mathrm{d}\widehat{Q}%
-\int_{K}\varphi _{N}\,\mathrm{d}\mu \right) \geq \sup_{N}\left( \left(
0\right) \frac{1}{2}+\left( 0\right) \frac{1}{2}-(-N)\right) =+\infty .\,
\end{align*}
\end{example}

\section{Applications in the compact case}

\label{secpartcases} 

The pricing-hedging dualities numbered 1., 2., 4., 5., 6., announced in the Introduction (page \pageref{elenco}) are proved in this section. 

\bigskip

We suppose, in the whole Section \ref{secpartcases} that the following requirements are
fulfilled.

\begin{standassumption}
\label{asssubsuper} Let $d=1$ and $\Omega :=K_{0}\times \dots \times K_{T}$
for compact sets $K_{0},\dots ,K_{T}\subseteq {\mathbb{R}}$; $K_0=\{x_0\}$
for some $x_0\in{\mathbb{R}}$; the functional $c:\Omega \rightarrow (-\infty
,+\infty ]$ is lower semicontinuous; $\widehat{Q}\in \mathrm{Mart}%
(\Omega )$ is a given martingale measure with marginals $\widehat{Q}%
_{0},\dots ,\widehat{Q}_{T}$; $c\in L^{1}(\widehat{Q})$.
\end{standassumption}

Under this assumption, $C_{0:T}=\mathcal{C}_{b}(\Omega )$  and $(C_{0:T})^{\ast }=\mathrm{ca}(\Omega )=\mathrm{ca}^1(\Omega )$. 
We observe that the stock  $X_t$ is assumed to be bounded due to
the compactness assumption on $K_0,\dots,K_T$. As a consequence, if we consider, for example, the call option $(X_t-\alpha)^+$, $\alpha \in \mathbb{R}$, then it is also bounded on $
\Omega$. The selection $\mathcal{E} \subseteq \mathcal{C}_{b}(K_{0})\times \dots
\times \mathcal{C}_{b}(K_{T})$ is then, in this context, appropriate.

\label{secsubsupercorollaries}

\subsection{Subhedging with vanilla options}

\label{caseofvanilla} As in \cite{Beig}, in this Section \ref{caseofvanilla}
we suppose that the elements in $\mathcal{E}_{t}$ represent portfolios
obtained combining call options with maturity $t$, units of the underlying
stock at time $t$ ($x_{t}$) and deterministic amounts, that is $\mathcal{E}%
_{t}$ consists of all the functions in $\mathcal{C}_{b}(K_{t})$ with the
following form: 
\begin{equation}  \label{mixcalls}
\varphi _{t}(x_{t})=a+bx_{t}+\sum_{n=1}^{N}c_{n}(x_{t}-\alpha_{n})^{+}\text{, for 
}a,b,c_{n},\alpha_{n}\in \mathbb{R},\,x_{t}\in K_{t}
\end{equation}%
and take $\mathcal{E}=\mathcal{E}_{0}\times \dots \times \mathcal{E}_{T}$.
As shown in the proofs of Corollaries \ref{corolldivinduced} and \ref%
{CorBeig}, which are the core content of this Section \ref{caseofvanilla},
one could as well take $\mathcal{E}=\mathcal{C}_{b}(K_{0})\times \dots
\times \mathcal{C}_{b}(K_{T})$ preserving validity of \eqref{subnonadd} and %
\eqref{subbeig}.

As it will become clear from the proofs, in all the results in Section \ref%
{caseofvanilla} the functional $U$ is real valued on the whole $\mathcal{E}$
and cash additive, which yields $\mathrm{dom}(U)=\mathrm{dom}(S^U)=\mathcal{E%
}$. Thus, we will exploit Corollary \ref{mainthm} and Corollary \ref{maincor},
in particular \eqref{ssub} and \eqref{ssup}, in the case $\mathrm{dom}(S^U)=%
\mathcal{E}$.

\label{secdivinducedbyutils}We set for $\varphi _{t}\in \mathcal{C}%
_{b}(K_{t})$ 
\begin{equation}
U_{\widehat{Q}_{t}}(\varphi _{t})=\sup_{\alpha ,\lambda \in {\mathbb{R}}%
}\left( \int_{K_{t}}u_{t}(\varphi _{t}(x_{t})+\alpha x_{t}+\lambda )%
\mathrm{d}\widehat{Q}_{t}(x_{t})-(\alpha x_{0}+\lambda )\right) \,.
\label{defUQ}
\end{equation}

We observe that Assumption \ref{aU1} does \textbf{not} impose that the
functions $u_{t}$ are real valued on the whole ${\mathbb{R}}$. Nevertheless,
for the functional $U_{\widehat{Q}_{t}}$ we have:

\begin{lemma}
\label{remfrominftytofiniteness} Under Assumption \ref{aU1} for each $%
t=0,\dots ,T$

\begin{enumerate}
\item $U_{\widehat{Q}_{t}}$ is real valued on $\mathcal{C}_{b}(K_{t})$ and
null in $0$.

\item $U_{\widehat{Q}_{t}}$ is concave and nondecreasing.

\item $U_{\widehat{Q}_{t}}$ is stock additive on $\mathcal{C}_{b}(K_{t})$,
namely for every $\alpha _{t},\lambda _{t}\in {\mathbb{R}}$ and $\varphi
_{t}\in \mathcal{C}_{b}(K_{t})$ 
\begin{equation*}
U_{\widehat{Q}_{t}}(\varphi _{t}+\alpha _{t}X_{t}+\lambda _{t})=U_{\widehat{Q%
}_{t}}(\varphi _{t})+\alpha _{t}x_{0}+\lambda _{t}\,.
\end{equation*}
\end{enumerate}
\end{lemma}

\begin{proof}
Clearly $U_{\widehat{Q}_{t}}(\varphi _{t})>-\infty $, as we may choose $%
\lambda _{t}\in {\mathbb{R}}$ so that $(\varphi _{t}+0X_{t}+\lambda _{t})\in 
\mathrm{dom}(u)\supseteq \lbrack 0,+\infty )$. Furthermore, 
\begin{equation}  \label{usegradientu}
\begin{split}
& U_{\widehat{Q}_{t}}(\varphi _{t})\overset{1\in \partial U_{t}(0)}{\leq }%
\sup_{\alpha ,\lambda \in {\mathbb{R}}}\left( \int_{K_{t}}\left( \varphi
_{t}+\alpha X_{t}+\lambda \right) \mathrm{d}\widehat{Q}_t-(\alpha
x_{0}+\lambda )\right) \\
& \overset{\widehat{Q}\in \mathrm{Mart}(\Omega )}{=}\sup_{\alpha ,\lambda
\in {\mathbb{R}}}\left( \int_{K_{t}}\varphi _{t}\,\mathrm{d}\widehat{Q}%
_t+(\alpha x_{0}+\lambda -\alpha x_{0}-\lambda )\right) \leq \left\Vert
\varphi _{t}\right\Vert _{\infty }.
\end{split}%
\end{equation}%
Finally, $0=\int_{K_{t}}u\left( 0\right) \,\mathrm{d}\widehat{Q}_{t}\leq U_{%
\widehat{Q}_{t}}(0)\leq \left\Vert 0\right\Vert _{\infty }$.

Item 2: monotonicity is trivial from the definition, while for concavity we see that whenever $\varphi^1,\varphi^2\in\mathcal{C}_b(K_t),0\leq \gamma\leq 1$ are given, setting $\varphi^\gamma=\gamma\varphi^1+(1-\gamma)\varphi^2$ and defining for $\lambda^1,\alpha^1,\lambda^2,\alpha^2\in\R$ the values $\lambda^\gamma,\alpha^\gamma\in\R$ analogously, we have by concavity of $U$ 
\begin{align*}
&\gamma U_{\widehat{Q}_{t}}(\varphi^1)+(1-\gamma)U_{\widehat{Q}_{t}}(\varphi^2)\\
&=\sup_{\lambda^1,\alpha^1,\lambda^2,\alpha^2\in\R}\left(\gamma\left(U(\varphi^1+\alpha^1 X_t+\lambda^1)-(\alpha^1 x_0+\lambda^1)\right)+(1-\gamma)\left(U(\varphi^2+\alpha^2 X_t+\lambda^2)-(\alpha^2 x_0+\lambda^2)\right)\right)\\
&\leq \sup_{\lambda^1,\alpha^1,\lambda^2,\alpha^2\in\R}\left( U(\varphi^\gamma+\alpha^\gamma X_t+\lambda^\gamma)-(\alpha^\gamma x_0+\lambda^\gamma)\right)\\
&\leq\sup_{\lambda,\alpha\in\R}\left( U(\varphi^\gamma+\alpha X_t+\lambda)-(\alpha x_0+\lambda)\right)=U_{\widehat{Q}_{t}}(\varphi^\gamma) \,.
\end{align*}

Item 3: we see that%
\begin{equation*}
U_{\widehat{Q}_{t}}(\varphi _{t}+\alpha _{t}X_{t}+\lambda _{t})=\sup 
_{\substack{ \alpha \in {\mathbb{R}}  \\ \lambda \in {\mathbb{R}}}}\left(
\int_{K_{t}}u_t\left( \varphi _{t}(x_{t})+(\alpha +\alpha
_{t})x_{t}+(\lambda +\lambda _{t})\right) \,\mathrm{d}\widehat{Q}%
_{t}(x_t)-(\alpha x_{0}+\lambda)\right)
\end{equation*}%
\begin{equation*}
=\sup_{\substack{ \alpha \in {\mathbb{R}}  \\ \lambda \in {\mathbb{R}}}}%
\left( \int_{K_{t}}u_t\left( \varphi _{t}(x_{t})+(\alpha +\alpha
_{t})x_{t}+(\lambda +\lambda _{t})\right) \,\mathrm{d}\widehat{Q}%
_{t}(x_t)-((\alpha _{t}+\alpha )x_{0}+(\lambda _{t}+\lambda ))\right)
+\alpha _{t}x_{0}+\lambda _{t},
\end{equation*}%
in which we recognize the definition of $U_{\widehat{Q}_{t}}(\varphi
_{t})+\alpha _{t}x_{0}+\lambda _{t}$.
\end{proof}

\bigskip

\begin{corollary}
\label{corolldivinduced} Take $u_{0},\dots ,u_{T}$ satisfying Assumption \ref%
{aU1} and suppose $\mathrm{dom}(u_{0})=\dots =\mathrm{dom}(u_{T})={\mathbb{R}%
}$. Then%
\begin{align}
\inf_{Q\in \mathrm{Mart}(\Omega )}\left( E_{Q}\left[ c\right] +\sum_{t=0}^{T}%
\mathcal{D}_{v_{t}^{\ast },\widehat{Q}_{t}}(Q_{t})\right) & =\sup \left\{
\sum_{t=0}^{T}U_{\widehat{Q}_{t}}(\varphi _{t})\mid \varphi \in \mathcal{S}%
_{sub}(c)\right\}\,.  \label{subnonadd}
\end{align}%
Moreover, if LHS of \eqref{subnonadd} is finite, a minimum point exists. %
%
%
%
%
%
%
%
%
\end{corollary}

\begin{proof}
Set $U(\varphi )=\sum_{t=0}^{T}U_{\widehat{Q}_{t}}(\varphi _{t})$ for $%
\varphi \in \mathcal{E}$. By Lemma \ref{remfrominftytofiniteness} for each $%
t=0,\dots ,T$ the monotone concave functional $\varphi _{t}\mapsto U_{%
\widehat{Q}_{t}}(\varphi _{t})$ is actually well defined, finite valued,
concave and nondecreasing on the whole $\mathcal{C}_{b}(K_{t}).$ Hence, by
the Extended Namioka-Klee Theorem (see \cite{bfnam}) it is norm continuous
on $\mathcal{C}_{b}(K_{t})$. 

We also observe that in this case we are in Setup \ref{SetupsigmaU} and we can apply Lemma 
\ref{conjugateaddU} Equation \eqref{additivestr}. We 
have 
\begin{align}
\mathcal{D}(Q)& :=\sup_{\varphi \in \mathcal{E}}\left( U(\varphi
)-\sum_{t=0}^{T}\int_{K_{t}}\varphi _{t}\,\mathrm{d}Q_{t}\right)=%
\sum_{t=0}^T\sup_{\varphi_t \in \mathcal{E}_t}\left( U_{\widehat{Q}%
_t}(\varphi_t )-\int_{K_{t}}\varphi _{t}\,\mathrm{d}Q_{t}\right)  \notag \\
&=\sum_{t=0}^T\sup_{\varphi_t \in \mathcal{C}_b(K_t)}\left( U_{\widehat{Q}%
_t}(\varphi_t )-\int_{K_{t}}\varphi _{t}\,\mathrm{d}Q_{t}\right)
\label{eq11} \\
&\overset{(\star)}{=}\sum_{t=0}^{T}\sup_{\varphi \in \mathcal{C}%
_{b}(K_t),\alpha_t\in{\mathbb{R}},\lambda_t\in{\mathbb{R}}}\left(
\int_{K_{t}}u_{t}(\varphi _{t}+\alpha_t X_t+\lambda_t)\,\mathrm{d}\widehat{Q}%
_{t}-\int_{K_{t}}(\varphi _{t}+\alpha_t X_t+\lambda_t)\,\mathrm{d}%
Q_{t}\right)  \notag \\
&\overset{(\star\star)}{=}\sum_{t=0}^{T}\sup_{\varphi \in \mathcal{C}%
_{b}(K_t)}\left( \int_{K_{t}}u_{t}(\varphi _{t})\,\mathrm{d}\widehat{Q}%
_{t}-\int_{K_{t}}\varphi _{t}\,\mathrm{d}Q_{t}\right)  \notag \\
& =\sum_{t=0}^{T}\sup_{\psi _{t}\in \mathcal{C}_{b}(K_t)}\left(
\int_{K_{t}}\psi _{t}\,\mathrm{d}Q_{t}-\int_{K_{t}}v_{t}(\psi _{t})\,\mathrm{%
d}\widehat{Q}_{t}\right) =\sum_{t=0}^{T}\mathcal{D}_{v_{t}^{\ast },\widehat{Q%
}_{t}}(Q_{t}),\,\,\,\forall Q\in \mathrm{Mart}(\Omega )  \notag
\end{align}%
where: in \eqref{eq11} we combined the aforementioned continuity of $U_{%
\widehat{Q}_t}$ on $\mathcal{C}_b(K_t)$ with the fact that $\mathcal{E}_{t}$
consists of all piecewise linear functions on $K_{t}$ hence $\mathcal{E}_t$
is norm dense in $\mathcal{C}_{b}(K_{t})$; in $(\star)$ we used the fact
that for $\widetilde{\varphi}_{t}:=\varphi _{t}+\alpha_t X_{t}+\lambda_t$,
and for every ${Q}\in \mathrm{Mart}(\Omega )$, $\int_{K_t}\widetilde{\varphi}%
\mathrm{d}{Q}_t=\int_{K_t}{\varphi}\mathrm{d}\widehat{Q}_t+\alpha_t
x_0+\lambda_t$, and in $(\star\star)$ we exploited the fact that $\widetilde{%
\varphi}_{t}\in \mathcal{E}_{t}$ for every $\varphi_t\in\mathcal{E}_t$, $%
\alpha_t ,\lambda_t \in {\mathbb{R}}$; the last equality follows from Proposition \ref%
{propvstar} Equation \eqref{explicitdual}.

Using Lemma \ref{conjugateaddU} and the fact that $U_{\widehat{Q}_{0}},\dots
,U_{\widehat{Q}_{T}}$ are cash additive we get $S^{U}(\varphi
)=\sum_{t=0}^{T}S^{U_{\widehat{Q}_{t}}}(\varphi _{t})=\sum_{t=0}^{T}U_{%
\widehat{Q}_{t}}(\varphi _{t})=U(\varphi )$. By Lemma \ref%
{remfrominftytofiniteness} Item 1, the assumptions of Corollary \ref%
{maincorcasecpt} are satisfied so that we obtain 
\begin{equation*}
\inf_{Q\in \mathrm{Mart}(\Omega )}\left( E_{Q}\left[ c(X)\right]
+\sum_{t=0}^{T}\mathcal{D}_{v_{t}^{\ast },\widehat{Q}_{t}}(Q_{t})\right)
=\sup \left\{ \sum_{t=0}^{T}U_{\widehat{Q}_{t}}(\varphi _{t})\mid \varphi
\in \,\mathcal{S}_{sub}(c)\right\} \,.
\end{equation*}
Existence of optima follows again from Corollary \ref{maincorcasecpt}.
\end{proof}


We stress the fact that in Corollary \ref{corolldivinduced} we assume that
all the functions $u_{0},\dots ,u_{T}$ are real valued on the whole ${%
\mathbb{R}}$. A more general result can be obtained when weakening this
assumption, but it requires an additional assumption on the marginals of $%
\widehat{Q}$. %

\begin{corollary}
\label{corwithsingular} Suppose Assumption \ref{aU1} is fulfilled. Assume
that $\widehat{Q}_{0},\dots ,\widehat{Q}_{T}$ have full support on $%
K_{0},\dots ,K_{T}$ respectively\footnote{%
\noindent Observe that we are requesting the full support property on $%
K_{0},\dots ,K_{T}$ with respect to their induced (Euclidean) topology. In
particular, this means that whenever $k_{t}\in K_{t}$ is an isolated point, $%
\widehat{Q}_{t}(\{k_{t}\})>0$. This is consistent with the assumption $%
K_{0}=\{x_{0}\}$, which implies $\mathrm{Prob}(K_{0})$ reduces to the Dirac
measure, $\mathrm{Prob}(K_{0})=\{\delta _{\{x_{0}\}}\}$.}. Then Equation %
\eqref{subnonadd} holds true replacing $\mathcal{D}_{v_{t}^{\ast },\widehat{Q%
}_{t}}(Q_{t})$ with $\mathcal{F}_{t}(Q_{t}| \widehat{Q}_{t})$. Moreover,
finiteness of the problem in the LHS of \eqref{subnonadd} implies the existence of a minimum.
\end{corollary}

\begin{proof}
The proof can be carried over almost literally as the proof of Corollary \ref%
{corolldivinduced}, with the exception of replacing the reference to
Proposition \ref{propconjugatesofranddiv} with the reference to Proposition %
\ref{PPP}.
\end{proof}

\bigskip

We now take $u_{t}(x)=x$ for each $t=0,\dots ,T$, and get $U_{\widehat{Q}%
_{t}}(\varphi _{t})=V_{\widehat{Q}_{t}}(\varphi _{t})={E}_{\widehat{Q}%
_{t}}[\varphi _{t}]$. Hence with an easy computation we have 
\begin{equation*}
\mathcal{D}_{v_{t}^{\ast },\widehat{Q}_{t}}(Q_{t})=%
\begin{cases}
0 & \,\,\,\,\,\text{ if }Q_{t}\equiv \widehat{Q}_{t} \\ 
+\infty  & \,\,\,\,\,\text{ otherwise. }%
\end{cases}%
\,\quad \text{for all }Q\in \mathrm{Mart}(\Omega ).
\end{equation*}%
Recalling that $\mathrm{Mart}(\widehat{Q}_{0},\dots ,\widehat{Q}_{T})=\{Q\in 
\mathrm{Mart}(\Omega )\mid Q_{t}\equiv \widehat{Q}_{t}\,\,\forall
\,t=0,\dots ,T\}$, from Corollary \ref{corolldivinduced} we can recover the
following result of \cite{Beig} (under the compactness assumption, which
will be dropped in Corollary \ref{CorNew}).

\begin{corollary}
\label{CorBeig}The following equality holds: 
\begin{align}
\inf_{Q\in \mathrm{Mart}(\widehat{Q}_{0},\dots ,\widehat{Q}_{T})}E_{Q}\left[
c\right] & =\sup \left\{ \sum_{t=0}^{T}{E}_{\widehat{Q}_{t}}[\varphi
_{t}]\mid \varphi \in \mathcal{S}_{sub}(c)\right\}\,.  \label{subbeig}
\end{align}
Moreover, if LHS of \eqref{subbeig} is finite, a minimum point exists.
\end{corollary}

\subsection{Subhedging without Options}

The pricing-hedging duality without options takes the following form.

\begin{corollary}
\label{corollrob} The following equality hold: 
\begin{equation}
\inf_{Q\in \mathrm{Mart}(\Omega )}E_{Q}\left[ c\right]=\sup \left\{ m\in {%
\mathbb{R}}\mid \exists \Delta\in\mathcal{H}\text{ with }m+I^\Delta\leq
c\right\} :=\Pi ^{\text{sub}}(c)\,.  \label{robustsub}
\end{equation}
Moreover, if LHS of \eqref{robustsub} is
finite, a minimum point exists.
\end{corollary}

\begin{proof}
We take $\mathcal{E}_{0}=\dots =\mathcal{E}_{T}={\mathbb{R}}$ and $\mathcal{E%
}=\mathcal{E}_{0}\times \dots \times \mathcal{E}_{T}={\mathbb{R}}^{T+1}.$ For each $\varphi \in \mathcal{E}$ with $%
\varphi =[m_{1},\dots ,m_{T}]$, $m\in {\mathbb{R}}^{T+1}$ we select\footnote{%
As $\mathcal{E}={\mathbb{R}}^{T+1}$, we notice that if $u_{t}(x_{t})=x_{t},$ 
$t=0,...,T,$ and $\widehat{Q}\in \mathrm{Mart}(\Omega )$, the functional $U_{%
\widehat{Q}_{t}}$ defined in \eqref{defUQ} is given by $U_{\widehat{Q}%
_{t}}(m_{t})=m_{t}$ and so $U(m)=\sum_{t=0}^{T}U_{\widehat{Q}%
_{t}}(m_{t})=\sum_{t=0}^{T}m_{t}$ for all $m\in \mathcal{E}$.} $U(\varphi
)=\sum_{t=0}^{T}m_{t}$. Then by definition of $\mathcal{D}$ (see Section \ref{sectionsetting}) we get 
\begin{equation*}
\mathcal{D}(\gamma )=%
\begin{cases}
0 & \text{ for }\gamma \in \mathrm{ca}(\Omega )\text{ s.t. }\gamma (\Omega
)=1 \\ 
+\infty & \text{ otherwise.}%
\end{cases}%
.
\end{equation*}%
In particular $\mathcal{D}(Q)=0$ for every $Q\in \mathrm{Mart}(\Omega )$.
Moreover we observe that $\mathcal{S}^{U}(\varphi )=U(\varphi )$ for every $%
\varphi \in \mathcal{E}$. Applying Corollary \ref{maincorcasecpt}, from Equation \eqref{subhedgingcor}
we get that 
\begin{equation*}
\inf_{Q\in \mathrm{Mart}(\Omega )}E_{Q}\left[ c\right] =\sup \left\{
\sum_{t=0}^{T}m_{t}\mid m_{0},\dots ,m_{T}\in {\mathbb{R}}\text{ s.t. }%
\exists \,\Delta \in \mathcal{H}\text{ with }\sum_{t=0}^{T}m_{t}+I^{\Delta
}\leq c\right\} .
\end{equation*}%
We recognize in the RHS above the RHS of \eqref{robustsub}. Finally, existence of
optima follows again from Corollary \ref{maincorcasecpt}.
\end{proof}

\subsection{Subhedging without the market\label{SecEOT}}

As in Example \ref{ExEOT} we let $\{0\}\subseteq \mathcal{A}\subseteq
-(C_{0:T})_{+}$, so that $\mathrm{Prob}^{1}(\Omega )=\mathcal{A}^{\circ
}\cap \mathrm{Prob}^{1}(\Omega )$.

\begin{corollary}
\label{CorEOT} Let $\mathcal{E}=\bigtimes_{t=0}^T\mathcal{
C}_b(K_t)$ and suppose that $U(\varphi)=0$ for some $\varphi\in\mathcal{E}$.
Then: 
\begin{align}
\inf_{Q\in \mathrm{Prob}^{1}(\Omega )}\left( E_{Q}[c]+\mathcal{D}(Q)\right)
=\sup \left\{ S^{U}(\varphi )\mid \varphi \in \mathcal{E},\text{ }%
\sum_{t=0}^{T}\varphi _{t}\leq c\right\} \,.  \label{SubEOT}
\end{align}
\end{corollary}

\begin{proof}
It is enough to apply Corollary \ref{mainthm}, and to observe that since $%
f\in-\mathcal{A}\Rightarrow f\in (C_{0:T})_+$ we have  
\begin{equation*}
\sup \left\{ S^{U}(\varphi )\mid \varphi \in \mathcal{E},\text{ }%
\sum_{t=0}^{T}\varphi _{t}+f\leq c\text{ for some }f\in -\mathcal{A}%
\right\}=\sup \left\{ S^{U}(\varphi )\mid \varphi \in \mathcal{E},\text{ }%
\sum_{t=0}^{T}\varphi _{t}\leq c\right\} \,.
\end{equation*}
\end{proof}

The LHS of \eqref{SubEOT} is the Entropy Optimal Transport problem as in 
\cite{Sava}, but with $\mathrm{Prob}^{1}(\Omega )$ replacing $\mathrm{Meas}%
(\Omega )$ and \eqref{SubEOT} is the subhedging price where only options can
be used for hedging. Notice that a similar argument yields the same duality
in the noncompact case, provided that one invokes (if the corresponding
assumptions on $c$ and $U$ are satisfied) Theorem \ref%
{mainEMOTtheoremgeneral} in place of its Corollary \ref{mainthm}.

\subsection{Penalization with market price}

\label{CasePenalization1gen} In this Section we change our perspective.
Instead of starting from a given $U$, we will give a particular form of the
penalization term $\mathcal{D}$ and proceed in identifying the corresponding 
$U$ in the spirit of Remark \ref{remfenchelmoreau}. For each $t=0,\dots ,T$
we suppose that finite sequences $(c_{t,n})_{1\leq n\leq N_{t}}\subseteq {%
\mathbb{R}}$ and $(f_{t,n})_{1\leq n\leq N_{t}}\subseteq \mathcal{C}%
_{b}(K_{t})$ are given. The functions $(f_{t,n})_{1\leq n\leq
N_{t}}\subseteq \mathcal{C}_{b}(K_{t})$ represent payoffs of options whose
prices $(c_{t,n})_{1\leq n\leq N_{t}}\subseteq {\mathbb{R}}$ are known from
the market. We also take $\mathcal{E}=\bigtimes_{t=0}^T\mathcal{
C}_b(K_t)$. 

We introduce the notion of a loss function that will be useful
here and also in the sequel (see
Section \ref{secwasser} and Section \ref{secconvergewasser}) to build penalization functions.

\begin{definition}
\label{condpenalty} A function $G:{\mathbb{R}}\rightarrow (-\infty ,+\infty ]$
is called a \emph{loss function} if it is convex, nondecreasing, lower
semicontinuous and satisfies $G(0)=0$. We set $\mathrm{dom}(G):=\left\{ x\in 
{\mathbb{R}}\mid G(x)<+\infty \right\} $. The conjugate function $G^{\ast }:{%
\mathbb{R}}\rightarrow (-\infty ,+\infty ]$, defined by $G^{\ast
}(y)=\sup_{x\in {\mathbb{R}}}(xy-G(x)),$ satisfies, due to the monotonicity of $G$, $G^{\ast }(y)=+\infty $ for every $y<0$.
\end{definition}

Our requirements allow for a wide range of penalizations.
For example, we might use power-like penalizations, i.e. $G(x)=\frac{%
x^{p}}{{p}}$ for $x>0$ and for ${p}\in (1,+\infty )$, $%
G(x)=0$ for $x\leq 0$. In such a case for every $x\geq 0$ we have $%
G^{\ast }(x)=\frac{x^{q}}{{q}}$ for $\frac{1}{p}+%
\frac{1}{q}=1$. Alternatively, we
might take for $x,y\in {\mathbb{R}}$ 
\begin{equation}
\label{threshold}
G(x)=%
\begin{cases}
0 & \,\,\,\,\text{ if }x\leq \varepsilon  \\ 
+\infty  & \,\,\,\,\text{ otherwise}%
\end{cases}%
\,\,\,\,\,\,\,\,\,\Longrightarrow \,\,\,\,\,\,\,\,\,G^{\ast
}(y)=\varepsilon y\text{ for }y\geq 0\,.
\end{equation}

We define 
\begin{equation}
\mathrm{Mart}_{t}(K_{t})=\left\{ \gamma _{t}\in \mathrm{Prob}(K_{t})\mid
\exists \,Q\in \mathrm{Mart}(\Omega )\text{ with }\gamma _{t}\equiv
Q_{t}\right\} \subseteq\mathrm{ca}(K_t)  \label{marginalmartgen}
\end{equation}
and for for $\gamma _{t}\in \mathrm{ca}(K_{t})$ 
\begin{equation*}
\mathcal{D}^{G}_t(\gamma_t):=%
\begin{cases}
\sum_{n=1}^{N_{t}}G_{t,n}\left(\left\vert \int_{K_{t}}f_{t,n}\,\mathrm{d}%
\gamma _{t}-c_{t,n}\right\vert\right) & \,\,\,\,\text{ for }\gamma _{t}\in 
\mathrm{Mart}_t(K_{t}) \\ 
+\infty & \,\,\,\,\text{ otherwise }%
\end{cases}
\label{defDct1gen}
\end{equation*}


\begin{proposition}
\label{propdualwithbestfit} For every $n=0,\dots
,N_{t},\,t=0,\dots ,T$, we assume that $G_{n,t}:{\mathbb{R}}\rightarrow
(-\infty ,+\infty ]$\ is a loss function and we suppose that the martingale measure $\widehat{Q}\in 
\mathrm{Mart}(\Omega )$ in Standing Assumption \ref{asssubsuper} also
satisfies $\left\vert \int_{K_{t}}f_{t,n}\,\mathrm{d}\widehat{Q}%
_{t}-c_{t,n}\right\vert \in \mathrm{dom}(G_{t,n})$. Then

\begin{equation}
\inf_{Q\in \mathrm{Mart}(\Omega )}\left( E_{Q}\left[ c\right] +\sum_{t=0}^{T}%
\mathcal{D}_{t}^{G}(Q_{t})\right) =\sup \left\{
\sum_{t=0}^{T}U_{t}^{G}(\varphi _{t})\mid \varphi \in \mathcal{S}%
_{sub}(c)\right\} ,  \label{subpenal}
\end{equation}%
where 
\begin{equation*}
U_{t}^{G}(\varphi _{t}):=\sup_{y_{t}\in {\mathbb{R}}^{N_{t}}}\left( \Pi ^{%
\text{sub}}\left( \varphi _{t}+\sum_{n=1}^{N_{t}}y_{t,n}(f_{t,n}-c_{t,n})\right)
-\sum_{n=1}^{N_{t}}G_{t,n}^{\ast }(y_{t,n})\right) 
\end{equation*}%
is a stock additive functional and $\Pi ^{\text{sub}}$ is given in %
\eqref{robustsub}. 
%
%
Finally, if LHS of \eqref{subpenal} is finite, a minimum point exists.
\end{proposition}

\begin{proof}
$\,$

\textbf{Step 1}: Set $g_{t,n}:=f_{t,n}-c_{t,n}
$. 
For any $t\in \{0,\dots ,T\}$ we prove that the
functional $\mathcal{D}_{t}^{G}$ is $\sigma (\mathrm{ca}(K_{t}),\mathcal{C}%
_{b}(K_{t}))$-lower semicontinuous and  that, for every $\varphi _{t}\in \mathcal{C}%
_{b}(K_{t})$, its Fenchel-Moreau (convex) conjugate satisfies 
\begin{equation*}
V_{t}^{G}(\varphi _{t}):=\sup_{\gamma _{t}\in \mathrm{ca}(K_{t})}\left(
\int_{K_{t}}\varphi _{t}\,\mathrm{d}\gamma _{t}-\mathcal{D}_{t}^{G}(\gamma
_{t})\right) =\inf_{y_{t}\in {\mathbb{R}}^{N_{t}}}\left( \Pi ^{{sup}%
}\left( \varphi _{t}-\sum_{n=1}^{N_{t}}y_{t,n}g_{t,n}\right)
+\sum_{n=1}^{N_{t}}G_{t,n}^{\ast }(y_{t,n})\right) \,,
\end{equation*}%
and thus 
\begin{equation}
U_{t}^{G}(\varphi _{t}):=-V_{t}^{G}(-\varphi _{t})=\sup_{y_{t}\in {\mathbb{R}%
}^{N_{t}}}\left( \Pi ^{{sub}}\left( \varphi
_{t}+\sum_{n=1}^{N_{t}}y_{t,n}g_{t,n}\right)
-\sum_{n=1}^{N_{t}}G_{t,n}^{\ast }(y_{t,n})\right) \,.  \label{utpenalized}
\end{equation}%
Here we are using the definition of the superhedging price $\Pi ^{{sup}}(g):=-\Pi ^{{sub}}(-g)=\sup_{Q\in \mathrm{Mart}(\Omega )}E_{Q}\left[ g\right] $, by Corollary \ref{corollrob}.
We observe that $\mathcal{D}_{t}^{G}$ is $\sigma (\mathrm{ca}(K_{t}),%
\mathcal{C}_{b}(K_{t}))$-lower semicontinuous (it is a sum of functions,
each being composition of a lower semicontinuous function and a continuous
function on $\mathrm{Mart}_{t}(K_{t})$ which is $\sigma (\mathrm{ca}(K_{t}),%
\mathcal{C}_{b}(K_{t}))$-compact by Lemma \ref{lemmamargmartcptgen}). We now
need to compute 
\begin{equation*}
V_{t}^{G}(\varphi _{t})=\sup_{\gamma _{t}\in \mathrm{ca}(K_{t})}\left(
\int_{K_{t}}\varphi _{t}\,\mathrm{d}\gamma _{t}-\mathcal{D}_{t}^{G}(\gamma
_{t})\right) =\sup_{Q_{t}\in \mathrm{Mart}_{t}(K_{t})}\left(
\int_{K_{t}}\varphi _{t}\,\mathrm{d}Q_{t}-\mathcal{D}_{t}^{G}(Q_{t})\right)
\,.
\end{equation*}%
Recall now that from Fenchel-Moreau Theorem that $G_{t,n}(x)=\sup_{y\in {%
\mathbb{R}}}(xy-G_{t,n}^{\ast }(y))$. Hence 
\begin{align*}
V_{t}^{G}(\varphi _{t})& =\sup_{Q_{t}\in \mathrm{Mart}_{t}(K_{t})}\left(
\int_{K_{t}}\varphi _{t}\,\mathrm{d}Q_{t}-\sum_{n=1}^{N_{t}}\sup_{y_{t,n}\in 
{\mathbb{R}}}\left( y_{t,n}\int_{K_{t}}g_{t,n}\,\mathrm{d}%
Q_{t}-G_{t,n}^{\ast }(y_{t,n})\right) \right)  \\
& =\sup_{Q_{t}\in \mathrm{Mart}_{t}(K_{t})}\left( \int_{K_{t}}\varphi _{t}\,%
\mathrm{d}Q_{t}-\sum_{n=1}^{N_{t}}\sup_{y_{t,n}\in \mathrm{dom}%
(G_{t,n}^{\ast })}\left( y_{t,n}\int_{K_{t}}g_{t,n}\,\mathrm{d}%
Q_{t}-G_{t,n}^{\ast }(y_{t,n})\right) \right)  \\
& =\sup_{Q_{t}\in \mathrm{Mart}_{t}(K_{t})}\inf_{y_{t}\in \mathrm{dom}%
}\left( \int_{K_{t}}\left( \varphi
_{t}-\sum_{n=1}^{N_{t}}y_{t,n}g_{t,n}\right) \,\mathrm{d}Q_{t}+%
\sum_{n=1}^{N_{t}}G_{t,n}^{\ast }(y_{t,n})\right)  \\
& =:\sup_{Q_{t}\in \mathrm{Mart}(K_{t})}\inf_{y_{t}\in \mathrm{dom}}\mathcal{%
T}(y_{t},Q_{t}),
\end{align*}%
where $\mathrm{dom}=\mathrm{dom}(G_{t,1}^{\ast })\times \dots \times \mathrm{%
dom}(G_{t,N_{t}}^{\ast })\subseteq {\mathbb{R}}^{N_{t}}$. We now see that $%
\mathcal{T}$ is real valued on $\mathrm{dom}\times \mathrm{Mart}_{t}(K_{t})$%
, is convex in the first variable and concave in the second. Moreover, $\{%
\mathcal{T}(y_{t},\cdot )\geq C\}$ is $\sigma (\mathrm{Mart}_{t}(K_{t}),%
\mathcal{C}_{b}(K_{t}))$-closed in $\mathrm{Mart}_{t}(\Omega )$ for every $%
y_{t}\in \mathrm{dom}$, and $\mathrm{Mart}_{t}(K_{t})$ is $\sigma (\mathrm{%
Mart}_{t}(K_{t}),\mathcal{C}_{b}(K_{t}))$-compact (by Lemma \ref%
{lemmamargmartcptgen}). As a consequence $\mathcal{T}(y_{t},\cdot )$ is $%
\sigma (\mathrm{Mart}_{t}(K_{t}),\mathcal{C}_{b}(K_{t}))$-lower
semicontinuous on $\mathrm{Mart}_{t}(K_{t})$. We can apply \cite{Simons}
Theorem 3.1 with $A=\mathrm{dom}$ and $B=\mathrm{Mart}_{t}(K_{t})$ endowed
with the topology $\sigma (\mathrm{Mart}_{t}(K_{t}),\mathcal{C}_{b}(K_{t}))$%
, and interchange $\inf $ and $\sup $. From our previous computations we
then get 
\begin{align*}
V_{t}^{G}(\varphi _{t})& =\sup_{Q_{t}\in \mathrm{Mart}_{t}(K_{t})}%
\inf_{y_{t}\in \mathrm{dom}}\mathcal{T}(y_{t},Q_{t})=\inf_{y_{t}\in \mathrm{%
dom}}\sup_{Q_{t}\in \mathrm{Mart}_{t}(K_{t})}\mathcal{T}(y_{t},Q_{t}) \\
& =\inf_{y_{t}\in \mathrm{dom}}\left( \sup_{Q_{t}\in \mathrm{Mart}%
_{t}(K_{t})}\int_{K_{t}}\left( \varphi
_{t}-\sum_{n=1}^{N_{t}}y_{t,n}g_{t,n}\right) \,\mathrm{d}Q_{t}+%
\sum_{n=1}^{N_{t}}G_{t,n}^{\ast }(y_{t,n})\right)  \\
& =\inf_{y_{t}\in \mathrm{dom}}\left( \sup_{Q\in \mathrm{Mart}(\Omega
)}\int_{\Omega }\left( \varphi _{t}-\sum_{n=1}^{N_{t}}y_{t,n}g_{t,n}\right)
\,\mathrm{d}Q+\sum_{n=1}^{N_{t}}G_{t,n}^{\ast }(y_{t,n})\right)  \\
& {=}\inf_{y_{t}\in \mathrm{dom}}\left( \Pi ^{{%
sup}}\left( \varphi _{t}-\sum_{n=1}^{N_{t}}y_{t,n}g_{t,n}\right)
+\sum_{n=1}^{N_{t}}G_{t,n}^{\ast }(y_{t,n})\right)  \\
& =\inf_{y_{t}\in {\mathbb{R}}^{N_{t}}}\left( \Pi ^{{sup}}\left( \varphi
_{t}-\sum_{n=1}^{N_{t}}y_{t,n}g_{t,n}\right)
+\sum_{n=1}^{N_{t}}G_{t,n}^{\ast }(y_{t,n})\right) .
\end{align*}%
Equation \eqref{utpenalized} can be obtained with minor manipulations.

\textbf{Step 2}: conclusion. We are clearly in the setup of Corollary  \ref{maincorcasecpt} with $\mathcal{D}$ given as in Setup \ref{SetupsigmaD} from $%
\mathcal{D}_{0}^{G},\dots ,\mathcal{D}_{T}^{G}$, and by definition $\mathrm{%
dom}(\mathcal{D}_{t}^{G})\subseteq \mathrm{Prob}(K_{t})$ for each $t=0,\dots
,T$. 
Using Lemma \ref{sumislsc} Item 2, together with the computations in Step 1
and the fact that clearly $S^{U_{t}^{G}}\equiv U_{t}^{G}$ by cash additivity
of $U_{t}^{G}$, we get the desired equality from Corollary \ref%
{maincorcasecpt} Equation \eqref{subhedgingcor}: indeed, $G_{t,n}^{\ast }$
is lower bounded and proper by our assumptions on $G_{t,n}$, and $\Pi ^{sub}$
is real valued and cash additive on bounded continuous functions. This
guarantees that $V_{t}^{G}(\varphi _{t})$ is null for an appropriate choice
of (constant) $\varphi _{t}$. Existence of optima follows again from
Corollary \ref{maincorcasecpt}.
\end{proof}

\begin{remark}
Our assumption of existence of a particular $\widehat{Q}\in \mathrm{Mart}%
(\Omega )$ in Proposition \ref{propdualwithbestfit} expresses the fact that
we are assuming our market prices $(c_{t,n})_{t,n}$ are close enough to
those given by expectations under some martingale measure.
\end{remark}

\begin{example}
Proposition \ref{propdualwithbestfit} covers a wide range of penalizations.
For example, we might impose a threshold for the
fitting, by taking into account only those martingale measure $Q$ such
that $\left\vert \int_{\Omega }f_{t,n}\,\mathrm{d}Q_{t}-c_{t,n}\right\vert
\leq \varepsilon _{t,n}$ for some $\varepsilon _{t,n}\geq 0$. To express this, just
take $G_{t,n}$ in the form \eqref{threshold} for $\varepsilon=\varepsilon_{t,n}$
\end{example}

\begin{example}
\label{examplepenalmkttomot}
We now study the convergence of the penalized problem described above to the classical MOT problem. We suppose that our information on the marginal
distributions increases, by increasing the number of prices available from
the market. We take $f_{t,n}(x_t)=(x_t-\alpha_n)^+$ to be call options with
maturity $t$ and strikes $(\alpha_n)_n$ that are a dense subset of ${\mathbb{%
R}}$.

We take as loss functions $G_{t,n}(x)=0$ for $x\leq 0$ and $%
G_{t,n}(x)=+\infty $ for all $x>0$ $t=0,\dots ,T$, $n\geq 1$. This means
that in LHS of \eqref{subpenal} the infimum is equivalently taken only over
martingale measures whose theoretical prices exactly match the ones for the
data, namely $c_{t,n}$. For each $t=0,\dots ,T$, $(c_{t,n})_{n}$ is a given
sequence of prices, and we suppose that they are all computed under the same
martingale measure $\widehat{Q}\in \mathrm{Mart}(\Omega )$. We consider for
each $k\in \mathbb{N}$ the initial segment $c_{t,1},\dots ,c_{t,N_{t}(k)}$
for sequences $N_{t}(k)\uparrow _{k}\infty $, $t=0,\dots ,T$.

This means that 
\begin{align*}
\mathcal{D}_{k}(Q)& :=\sum_{t=0}^{T}\sum_{n=1}^{N_{t}(k)}G_{t,n}\left(
\left\vert \int_{K_{t}}f_{t,n}\,\mathrm{d}Q_{t}-c_{t,n}\right\vert \right) 
\\
& \leq \sum_{t=0}^{T}\sum_{n=1}^{N_{t}(k+1)}G_{t,n}\left( \left\vert
\int_{K_{t}}f_{t,n}\,\mathrm{d}Q_{t}-c_{t,n}\right\vert \right) =:\mathcal{D}%
_{k+1}(Q)\,\,\,\forall Q\in \mathrm{Mart}(\Omega )\,.
\end{align*}%
and

\begin{equation*}
\mathcal{D}_{\infty }(Q)=\sup_{k}\mathcal{D}_{k+1}(Q)=\sum_{t=0}^{T}%
\sum_{n=1}^{+\infty }G_{t,n}\left( \left\vert \int_{K_{t}}f_{t,n}\,\mathrm{d}%
Q_{t}-c_{t,n}\right\vert \right) 
\end{equation*}%
so that

\begin{equation*}
\mathcal{D}_\infty(Q)=%
\begin{cases}
0\,\, & \text{if }\int_{K_{t}}f_{t,n}\,\mathrm{d}Q _{t}=c_{t,n}\,\forall
0\leq t\leq T, n\geq 1 \\ 
+\infty\,\, & \text{otherwise }%
\end{cases}%
\,.
\end{equation*}

From the density of $(\alpha_n)_n$ we conclude that 
\begin{equation*}
\mathcal{D}_\infty(Q)=%
\begin{cases}
0\,\, & \text{if }Q_t\equiv\widehat{Q}_t,\forall 0\leq t\leq T \\ 
+\infty\,\, & \text{otherwise }%
\end{cases}%
\,.
\end{equation*}

As a consequence, by Proposition \ref{propstability}, we have the convergence 
$$\inf_{\probq\in\mathrm{Mart}(\Omega)}\left(E_\probq[c]+\mathcal{D}_k(\probq)\right)\rightarrow_k \inf_{Q\in 
\mathrm{Mart}(\widehat{Q}_{0},\dots ,\widehat{Q}_{T})}E_{Q}\left[ c\right]\,.$$
\end{example}
\subsection{Penalty terms given via Wasserstein distance}

\label{secwasser} Let $d_{t}$ be a metric on $K_{t}$ (equivalent to the
Euclidean one). Call $W_{t}:\mathrm{Prob}(K_{t})\times \mathrm{Prob}%
(K_{t})\rightarrow {\mathbb{R}}$ the (1-)Wasserstein distance induced by $%
d_{t}$. Let $\mathrm{Lip}(1,K_{t})$ be the the class of $d_{t}$-Lipschitz
functions on $K_{t}$ with Lipschitz constant not bigger than $1$. Notice
that $\mathrm{Lip}(1,K_{t})\subseteq \mathcal{C}_{b}(K_{t})$ since $d_{t}$
is equivalent to the Euclidean metric. For each $t$, let $G_{t}:{\mathbb{R}}%
\rightarrow (-\infty ,+\infty ]$ be a loss function as given in Definition \ref{condpenalty}. For $\mathrm{Mart}%
_{t}(K_{t})$ as given in Section \ref{CasePenalization1gen}, we introduce 
\begin{equation}
Q_{t}\in \mathrm{Prob}(K_{t})\mapsto \mathcal{D}^W_{t}(Q_{t}):=%
\begin{cases}
G_{t}(W_{t}(Q_{t},\widehat{Q}_{t})) & \text{for }Q_{t}\in \mathrm{Mart}%
_{t}(K_{t}) \\ 
+\infty  & \text{otherwise}%
\end{cases}.
\label{penalwasser}
\end{equation}%
$\mathcal{D}^W_{t}$ is lower semicontinuous w.r.t. the topology of weak
convergence of probability measures, since the Wasserstein metric metrizes
it for compact underlying spaces and $\mathrm{Mart}_{t}(K_{t})$ is compact
under such a topology, by Lemma \ref{lemmamargmartcptgen}. We are then in Setup \ref{SetupsigmaD} and of Lemma %
\ref{sumislsc} Item 2.
As is Section \ref{CasePenalization1gen} we take $\mathcal{E}=\bigtimes_{t=0}^T\mathcal{
C}_b(K_t)$. 
\begin{proposition}
\label{propwasserstein}For each $t=0,\dots ,T$, suppose that $G_t$ is a loss function, that there exists a $Q\in \mathrm{Mart}%
(\Omega )$ such that $G_{t}(W_{t}(Q_{t},\widehat{Q}_{t}))<+\infty$
,\,where $\widehat{Q}\in \mathrm{Mart}(\Omega )$ is the
martingale measure from Standing Assumption \ref{asssubsuper} and take 
$\mathcal{D}^W_{t}$ as defined in \eqref{penalwasser}. Then 
\begin{equation}
\inf_{Q\in \mathrm{Mart}(\Omega )}\left( E_{Q}\left[ c\right] +\sum_{t=0}^{T}%
\mathcal{D}^W_{t}(Q_{t})\right) =\sup \left\{ \sum_{t=0}^{T}U^W_{t}(\varphi
_{t})\mid \varphi \in \mathcal{S}_{sub}(c)\right\} ,  \label{subwasser}
\end{equation}%
where 
\begin{equation*}
U^W_{t}(\varphi _{t}):=\sup_{\substack{ y\geq 0 \\ \ell _{t}\in \mathrm{Lip}%
(1,K_{t})}}\left( \Pi ^{sub}(\varphi _{t}+y\ell _{t})-\int_{K_{t}}y\ell _{t}%
\mathrm{d}\widehat{Q}_{t}-G_{t}^{\ast }(y)\right) 
\end{equation*}%
is a stock additive functional and $\Pi ^{\text{sub}}$ is given in %
\eqref{robustsub}. 
%
%
Finally, if LHS of \eqref{subwasser} is finite, a minimum point exists.
\end{proposition}

\begin{proof}
Starting from $\mathcal{D}^W_t$, we compute now the associated $V^W_t$ as

\begin{align*}
V^W_{t}(\varphi _{t})& :=\sup_{\gamma \in \mathrm{ca}(K_{t})}\left(
\int_{K_{t}}\varphi _{t}\mathrm{d}\gamma -\mathcal{D}^W_{t}(\gamma )\right)
=\sup_{Q\in \mathrm{Mart}_{t}(K_{t})}\left( \int_{K_{t}}\varphi _{t}\mathrm{d%
}Q-G_{t}(W_{t}(Q,\widehat{Q}_{t}))\right)  \\
& \overset{(a)}{=}\sup_{Q\in \mathrm{Mart}_{t}(K_{t})}\left(
\int_{K_{t}}\varphi _{t}\mathrm{d}Q-\sup_{y\geq 0}\left( yW_{t}(Q,\widehat{Q}%
_{t})-G_{t}^{\ast }(y)\right) \right)  \\
& =\sup_{Q\in \mathrm{Mart}_{t}(K_{t})}\inf_{y\geq 0}\left(
\int_{K_{t}}\varphi _{t}\mathrm{d}Q-yW_{t}(Q,\widehat{Q}_{t})+G_{t}^{\ast
}(y)\right)  \\
& \overset{(b)}{=}\sup_{Q\in \mathrm{Mart}_{t}(K_{t})}\inf_{y\in \mathrm{dom}%
(G_{t}^{\ast })}\left( \int_{K_{t}}\varphi _{t}\mathrm{d}Q-y\sup_{\ell \in 
\mathrm{Lip}(1,K_{t})}\left( \int_{K_{t}}\ell \mathrm{d}Q-\int_{K_{t}}\ell 
\mathrm{d}\widehat{Q}_{t}\right) +G_{t}^{\ast }(y)\right)  \\
& \overset{}{=}\sup_{Q\in \mathrm{Mart}_{t}(K_{t})}\inf_{\substack{ y\in 
\mathrm{dom}(G_{t}^{\ast }) \\ \ell \in \mathrm{Lip}(1,K_{t})}}\left(
\int_{K_{t}}(\varphi _{t}-y\ell )\mathrm{d}Q+\int_{K_{t}}y\ell \mathrm{d}%
\widehat{Q}_{t}+G_{t}^{\ast }(y)\right)  \\
& \overset{(c)}{=}\inf_{\substack{ y\in \mathrm{dom}(G_{t}^{\ast }) \\ \ell
\in \mathrm{Lip}(1,K_{t})}}\left( \sup_{Q\in \mathrm{Mart}_{t}(K_{t})}\left(
\int_{K_{t}}(\varphi _{t}-y\ell )\mathrm{d}Q\right) +\int_{K_{t}}y\ell _{t}%
\mathrm{d}\widehat{Q}_{t}+G_{t}^{\ast }(y)\right)  \\
& \overset{(d)}{=}\inf_{\substack{ y\in \mathrm{dom}(G_{t}^{\ast }) \\ \ell
_{t}\in \mathrm{Lip}(1,K_{t})}}\left( \Pi ^{sup}(\varphi _{t}-y\ell
_{t})+\int_{K_{t}}y\ell _{t}\mathrm{d}\widehat{Q}_{t}+G_{t}^{\ast
}(y)\right) =\inf_{\substack{ y\geq 0 \\ \ell _{t}\in \mathrm{Lip}(1,K_{t})}}%
\left( \Pi ^{sup}(\varphi _{t}-y\ell _{t})+\alpha (y,\ell _{t})\right) 
\end{align*}%
for the penalty $\alpha (y,\ell _{t}):=\int_{K_{t}}y\ell _{t}\mathrm{d}%
\widehat{Q}_{t}+G_{t}^{\ast }(y)$. In the equality chain above we use the
following facts: in (a) we used the dual representation of $G_{t}$, in (b)
we exploited the definition of $\mathrm{dom}(G_{t}^{\ast })$ and the
classical Kantorovich-Rubinstein Duality (see \cite{Villani09} Remark 6.5),
in (c) we applied \cite{Simons} Theorem 3.1, observe that $\mathrm{Mart}%
_{t}(K_{t})$ is compact by Lemma \ref{lemmamargmartcptgen}), in (d) we used
the definition of the superhedging price $\Pi ^{{sup}}(g):=-\Pi ^{{sub}}(-g)=\sup_{Q\in \mathrm{Mart}(\Omega )}E_{Q}\left[ g\right] $, by Corollary \ref{corollrob}.
Once we have $V^W_{t}$, we have $U^W_{t}$ and we can
argue as in Step 2 of the proof of Proposition \ref{propdualwithbestfit},
also regarding existence of an optimum.
\end{proof}

\begin{remark}
If $U^W_{t}$ (as well as $U_{t}^{G}$ in the previous Proposition \ref%
{propdualwithbestfit}) is real valued on $\mathcal{C}_{b}(K_{t})$, one might
take $\mathcal{E}_{t}$ as the set of functions in the form \eqref{mixcalls}
in place of $\mathcal{E}_{t}=\mathcal{C}_{b}(K_{t})$ in both Proposition \ref%
{propdualwithbestfit} and \ref{propwasserstein}, using norm density of the
piecewise linear functions just as in the proof of Corollary \ref%
{corolldivinduced}.
\end{remark}

\begin{remark}
\label{remnotmartinwasser} The reader can check that the property $\widehat{Q%
}\in \mathrm{Mart}(\Omega )$ is not used in the proof, and that it would
suffice to have only $\widehat{Q}\in \mathrm{Prob}(\Omega )$. This will be
exploited in Section \ref{secconvergewasser}.
\end{remark}

\begin{example}
\label{exampleemotwasserdistance} Taking $G_{t}(x)=0$ if $x\leq \varepsilon
_{t}$, $G_{t}(x)=+\infty $ otherwise, we get $G_{t}^{\ast }(y)=\varepsilon
_{t}y$ if $y\geq 0$ and $G_{t}^{\ast }(y)=+\infty $ otherwise. In this case {\small 
\begin{equation}
\inf_{Q\in \mathrm{Mart}(\Omega )}\left( E_\probq[c]+\sum_{t=0}^T
\mathcal{D}^W_t(Q)\right) =\inf \left\{ E_\probq[c]\mid Q\in 
\mathrm{Mart}(\Omega )\text{ and }\,\,W_{t}(Q_{t},\widehat{Q}_{t})\leq
\varepsilon _{t} \, \forall t=0,\dots ,T\right\} \,.
\label{infmarginalsepsilonclose}
\end{equation}%
}

One can verify, with the same techniques of Example \ref{examplepenalmkttomot},
that we have convergence, as  $\varepsilon _{t}\downarrow 0$ for every $t=0,\dots,T$, of the
problems in RHS of \eqref{infmarginalsepsilonclose} to the MOT problem in LHS of
\eqref{subbeig}.

For the dual problem, we show that in this particular case, calling $V_{\varepsilon _{t}}$ the associated $V_t^W$
as computed above, $\lim_{\varepsilon _{t}\downarrow
0}V_{\varepsilon _{t}}(\varphi _{t})=V_{0}(\varphi _{t})$. Indeed
\begin{equation*}
\lim_{\varepsilon _{t}\downarrow 0}V_{\varepsilon _{t}}(\varphi
_{t})=\inf_{\varepsilon _{t}>0}V_{\varepsilon _{t}}(\varphi _{t})=\inf
_{\substack{ y\geq 0 \\ \ell _{t}\in \mathrm{Lip}(1,K_t)}}\inf_{\varepsilon
_{t}>0}\left( \Pi ^{sup}(\varphi _{t}-y\ell _{t})+\int_{K_{t}}y\ell _{t}%
\mathrm{d}\widehat{Q}_{t}+G_{t}^{\ast }(y)\right) 
\end{equation*}%
\begin{equation*}
=\inf_{\substack{ y\geq 0 \\ \ell _{t}\in \mathrm{Lip}(1,K_t)}}\left( \Pi
^{sup}(\varphi _{t}-y\ell _{t})+\int_{K_{t}}y\ell _{t}\mathrm{d}\widehat{Q}%
_{t}\right) =V_{0}(\varphi _{t})\,.
\end{equation*}%

Now, one can also see that $V_{0}(\varphi _{t})={E}_{\widehat{Q}}[\varphi _{t}]$ using the definition of $V_{0}$ and observing that for $\varepsilon=0$ we have $\mathcal{D}^W_t=\delta_{\widehat{\probq_t}}$.
\end{example}

\subsubsection{Convergence with Wasserstein induced penalization}

\label{secconvergewasser}

%

As already mentioned, in the classical MOT framework, the marginals $%
\widehat{Q}_{0}\dots ,\widehat{Q}_{T}$ need to be determined, potentially
from the prices of vanilla options. It is then reasonable to suppose that in
a real-world situation, one proceeds by approximation, that is one
determines sequences of candidates $(\widehat{Q}_{t}^{n})_{n}\subseteq 
\mathrm{Prob}(\Omega )$ for $t=0,\dots ,T$. If such approximation
scheme (whose details are beyond the scope of this paper) is working, one
has a convergence of these sequences to the true marginals. One suitable
candidate for such a convergence is the weak one, that is one might suppose
that $\widehat{Q}_{t}^{n}\rightarrow _{n}\widehat{Q}_{t}^{\infty }:=\widehat{%
Q}_{t}$ for $t=0,\dots ,T$ in the weak sense for probability measures. We are here supposing that $K_{0}\,\dots ,K_{T}$
are compact sets and so such a weak convergence is equivalent to the
convergence in the Wasserstein distance. Proposition \ref%
{propconvergencewasser} will show how the EMOT problems treated in
Proposition \ref{propwasserstein} and associated to the approximating
measures $\widehat{Q}_{t}^{n},t=0,\dots ,T$ converge to the original MOT
problem for the true marginals $\widehat{Q}_{0},\dots ,\widehat{Q}_{T}$,
provided that the loss functions $G_{t}^{n}$ appropriately converge.

Let us now move to a mathematically precise formulation. For each $n\in 
\mathbb{N}\cup \{\infty \}$ and $t=0,\dots ,T$ let $\widehat{Q}_{t}^{n}\in 
\mathrm{Prob}(K_{t})$ be a fixed probability measure. Notice that, for $n\in \mathbb{N}$%
, we do not require these measures to be marginals of some martingale
measure.

\begin{proposition}
\label{propconvergencewasser} For each $n\in \mathbb{N}\cup \{\infty \}$ and $%
t=0,\dots ,T$ let $G_{t}^{n}$ be a loss functions with $\sup_{n\in \mathbb{N}}G_{t}^{n}(x)=G_{t}^{%
\infty }(x)$ for every $x\in {\mathbb{R}}$ and $G_{t}^{\infty
}(x)=+\infty $ for every $x>0$. For every $%
t=0,\dots ,T$ assume that $W_{t}(%
\widehat{Q}_{t}^{n},\widehat{Q}_{t}^{\infty })\rightarrow _{n}0$  and that there
exists a $Q^*\in \mathrm{Mart}(\Omega )$, with $c\in L^{1}(Q^*)$, having
marginals $\widehat{Q}_{t}^{\infty }$. If $\lim_{n}G_{t}^{n}(W_{t}(\widehat{Q}_{t}^{n},\widehat{Q}_{t}^{\infty }))=0$
for every $t=0,\dots ,T$, then 
\begin{equation}
\lim_{n}\mathfrak{P}_{n}^{W}(c)=\mathfrak{P}_{\infty }^{W}(c)=\inf_{Q\in 
\mathrm{Mart}(\widehat{Q}_{0}^{\infty },\dots ,\widehat{Q}_{T}^{\infty
})}E_{Q}\left[ c\right]\, ,  \label{stabwasser}
\end{equation}
where
\begin{equation*}
\mathfrak{P}_{n}^{W}(c)=\inf_{Q\in \mathrm{Mart}(\Omega )}\left( E_{Q}\left[
c\right] +\sum_{t=0}^{T}G_{t}^{n}(W_{1}(Q_{t},\widehat{Q}_{t}^{n}))\right), \,\,\,\,\,\,\,n\in \mathbb{N}\cup \{\infty \}
\,.
\end{equation*}

\end{proposition}

\begin{proof}
We first prove that $\mathfrak{P}_{\infty }^{W}(c)\leq \liminf_{n}\mathfrak{P%
}_{n}^{W}(c)$. Observe that 
there exists by assumption a martingale measure $Q^{\ast }$ with $c\in
L^{1}(Q^{\ast })$ and marginals $Q_{t}^{\ast }\equiv \widehat{Q}_{t}^{\infty
},t=0,\dots ,T$. By hypothesis, $G_{t}^{n}(W_{t}(\widehat{Q}_{t}^{n},{Q}%
_{t}^{\ast }))=G_{t}^{n}(W_{t}(\widehat{Q}_{t}^{n},\widehat{Q}_{t}^{\infty
})\rightarrow _{n}0$ for every $t=0,\dots ,T$, thus $G_{t}^{n}(W_{t}(%
\widehat{Q}_{t}^{n},{Q}_{t}^{\ast }))<+\infty $ for every $t=0,\dots ,T$, $%
n\in \mathbb{N}$ big enough and $n=+\infty $. Since $\Omega $ is compact and 
$G_{t}^{n},n\in \mathbb{N}\cup \{\infty \},t=0,\dots ,T$ are nonnegative on $%
[0,+\infty )$ we have 
\begin{equation*}
-\infty <\inf_{x\in \Omega }c(x)\leq \inf_{Q\in \mathrm{Mart}(\Omega
)}\left( E_{Q}\left[ c\right] +\sum_{t=0}^{T}G_{t}^{n}(W_{t}(Q_{t},\widehat{Q%
}_{t}^{n}))\right) \leq E_{Q^{\ast }}\left[ c\right] +%
\sum_{t=0}^{T}G_{t}^{n}(W_{t}(Q_{t}^{\ast },\widehat{Q}_{t}^{n}))<+\infty 
\end{equation*}%
for $n\in \mathbb{N}$ big enough and for $n=+\infty $. Hence $\mathfrak{P}%
_{n}(c)<+\infty $ for $n\in \mathbb{N}$ big enough and $n=+\infty $, and
w.l.o.g. we assume that $\mathfrak{P}_{n}^{W}(c)$ is finite for each $n\in 
\mathbb{N}\cup \{\infty \}$. Given then optima $Q^{n}$ (which exist by
Proposition \ref{propwasserstein} applied together with Remark \ref%
{remnotmartinwasser}) we can take a subsequence such that $E_{Q^{n_{k}}}%
\left[ c\right] +\sum_{t=0}^{T}G_{t}^{n_{k}}(W_{t}(Q_{t}^{n_{k}},\widehat{Q}%
_{t}^{n_{k}}))$ converges to $\liminf_{n}\mathfrak{P}_{n}^{W}(c)$ and such
that $W_{t}(Q^{n_{k}},\widetilde{Q})\rightarrow _{k}0$ for some $\widetilde{Q%
}\in \mathrm{Mart}(\Omega )$ (since $\mathrm{Mart}(\Omega )$ is compact,
recall the proof of Lemma \ref{lemmamargmartcptgen}). Since $Q\mapsto
\int_{\Omega }c\mathrm{d}Q$ is lower semicontinuous w.r.t. the weak
convergence, $G_{t}^{n}$ is lower semicontinuous and $W_{t}(\widetilde{Q}%
_{t},\widehat{Q}_{t}^{\infty })=\lim_{k}W_{t}(Q_{t}^{n_{k}},\widehat{Q}%
_{t}^{n_{k}})$, we have

\begin{align*}
& \left( E_{\widetilde{Q}}\left[ c\right] +\sum_{t=0}^{T}G_{t}^{N}(W_{t}(%
\widetilde{Q}_{t},\widehat{Q}_{t}^{\infty }))\right) =\lim_{k}\left(
E_{Q^{n_{k}}}\left[ c\right] +\sum_{t=0}^{T}G_{t}^{N}(W_{t}(Q_{t}^{n_{k}},%
\widehat{Q}_{t}^{n_{k}}))\right)  \\
& =\liminf_{n}\left( E_{Q^{n}}\left[ c\right] +%
\sum_{t=0}^{T}G_{t}^{N}(W_{t}(Q_{t}^{n},\widehat{Q}_{t}^{n}))\right)  \\
& \leq \liminf_{n\geq N}\inf_{Q\in \mathrm{Mart}(\Omega )}\left( E_{Q}\left[
c\right] +\sum_{t=0}^{T}G_{t}^{N}(W_{t}(Q_{t},\widehat{Q}_{t}^{n}))\right) 
\\
& \leq \liminf_{n\geq N}\inf_{Q\in \mathrm{Mart}(\Omega )}\left( E_{Q}\left[
c\right] +\sum_{t=0}^{T}G_{t}^{n}(W_{t}(Q_{t},\widehat{Q}_{t}^{n}))\right)
=\liminf_{n}\mathfrak{P}_{n}^{W}(c)\,.
\end{align*}%
Taking then a supremum over $N$ and an infimum over martingale measures, and
using the particular form of $G_{t}^{\infty },t=0,\dots ,T$, 
\begin{align*}
\mathfrak{P}_{\infty }^{W}(c)& =\inf_{Q\in \mathrm{Mart}(\widehat{Q}%
_{0}^{\infty },\dots ,\widehat{Q}_{T}^{\infty })}E_{Q}\left[ c\right]
=\inf_{Q\in \mathrm{Mart}(\Omega )}\sup_{N}\left( E_{Q}\left[ c\right]
+\sum_{t=0}^{T}G_{t}^{N}(W_{t}(Q_{t},\widehat{Q}_{t}^{n}))\right)  \\
& \leq \liminf_{n}\inf_{Q\in \mathrm{Mart}(\Omega )}\left( E_{Q}\left[ c%
\right] +\sum_{t=0}^{T}G_{t}^{n}(W_{t}(Q_{t},\widehat{Q}_{t}^{n}))\right)
=\liminf_{n}\mathfrak{P}_{n}^{W}(c)\,.
\end{align*}%
We now move on proving that $\mathfrak{P}_{\infty }^{W}(c)\geq \limsup_{n}%
\mathfrak{P}_{n}^{W}(c)$. Since $\mathfrak{P}_{\infty }^{W}(c)<+\infty $, we
have an optimum $Q^{\infty }$ and its marginals satisfy: $Q_{t}^{\infty }=%
\widehat{Q}_{t}^{\infty },t=0,\dots ,T$. Then 
\begin{align*}
\mathfrak{P}_{\infty }^{W}(c)& =E_{Q^{\infty }}\left[ c\right]
=\lim_{n}\left( E_{Q^{\infty }}\left[ c\right] +%
\sum_{t=0}^{T}G_{t}^{n}(W_{t}(\widehat{Q}_{t}^{\infty },\widehat{Q}%
_{t}^{n}))\right)  \\
& =\lim_{n}\left( E_{Q^{\infty }}\left[ c\right] +%
\sum_{t=0}^{T}G_{t}^{n}(W_{t}({Q}_{t}^{\infty },\widehat{Q}_{t}^{n}))\right) 
\\
& \geq \limsup_{n}\inf_{Q\in \mathrm{Mart}(\Omega )}\left( E_{Q}\left[ c%
\right] +\sum_{t=0}^{T}G_{t}^{n}(W_{t}(Q_{t},\widehat{Q}_{t}^{n}))\right)
=\limsup_{n}\mathfrak{P}_{n}^{W}(c)\,.
\end{align*}

Since now we have $\limsup_n\mathfrak{P}_n^W(c)\leq \mathfrak{P}%
_\infty^W(c)\leq \limsup_n\mathfrak{P}_n^W(c)$, \eqref{stabwasser} follows.
\end{proof}
\section{Applications in the noncompact case}

\label{sectionexamplesfornoncpaddition} 

In the noncompact case, Corollary \ref{corolldivinduced} takes the following
form:

\begin{corollary}
\label{corolldivinducednoncpt} Take $d=1$, $K_{0}=\{x_0\}$ for some $x_0\in{%
\mathbb{R}}$ and let $K_{1},\dots ,K_{T}\subseteq {\mathbb{R}}$ be closed subsets of $\mathbb{R}$. Consider nondecreasing $u_{0},\dots ,u_{T}$ satisfying Assumption %
\ref{aU1}, and suppose $\mathrm{dom}(u_{0})=\dots =\mathrm{dom}(u_{T})={%
\mathbb{R}}$. Take for each $t=0,\dots,T$ the vector space $\mathcal{E}_t\subseteq C_t$
of functions in the form \eqref{mixcalls}, let $\mathcal{E}=\mathcal{E}_{0}\times \dots \times \mathcal{E}_{T}$ and fix a $\widehat{Q}\in\mathrm{Mart}(\Omega)$ such that 
\begin{equation}  \label{condizOrlicz}
\int_{K_{t}}v_{t}\left(\alpha (1+\left|x_{t}\right|)\right)\,\mathrm{d}%
\widehat{Q }_{t}(x_{t})<+\infty\,\,\,\,\,\forall\alpha>0, t=0,\dots,T\,.
\end{equation}
Suppose that $c:\Omega\rightarrow(-\infty,+\infty]$ is lower semicontinuous
and satisfies \eqref{controlfrombelow}. 
Then 
\begin{equation}  \label{subnonaddnoncpt}
\begin{split}
\mathfrak{P}(c) &=\sup \left\{ \sum_{t=0}^{T}U_{\widehat{Q}_{t}}(\varphi
_{t})\mid \varphi\in\mathcal{E} \text{ is s.t. } \exists \Delta \in \mathcal{%
H}\text{ with }\sum_{t=0}^{T}\varphi _{t}(x_{t})+I^{\Delta }(x)\leq c(x)%
\text{ }\forall x\in \Omega \right\} \\
& =\inf_{Q\in \mathrm{Mart}(\Omega )}\left( E_{Q}\left[ c\right]
+\sum_{t=0}^{T}\mathcal{D}_{v_{t}^{\ast },\widehat{Q}_{t}}(Q_{t})\right)
\end{split}%
\end{equation}

%
where $U_{\widehat{Q}_{t}}(\varphi _{t})$ is defined in \eqref{defUQ} for
general $\varphi _{t}\in C_{t}$ and $\mathcal{D}_{v_{t}^{\ast },\widehat{Q}%
_{t}}$ is given in \eqref{explicitdual}. Moreover, the infimum in %
\eqref{subnonaddnoncpt} is a minimum provided that $\mathfrak{P}(c)<+\infty $%
.

%
%
%
%
%
%
\end{corollary}

\begin{proof}
All the claims follow from Theorem \ref{mainthm3noncpt}, provided that we
show all its hypotheses are satisfied. To do so, we check that: (i) $%
S_t(\varphi_t):=U_{\widehat{Q}_t}(\varphi_t)$ is real valued on ${C}_t$, concave and cash additive;
(ii) $\mathcal{D}_t(Q_t)=\mathcal{D}_{v_t^*,\widehat{Q}_t}(Q_t)$ for every $%
Q\in\mathrm{Mart}(\Omega)$; (iii)  $S_t$ is null in $0$ for every $t=0,\dots,T$ and the conditions \eqref{ineqoutsidecpt} and %
\eqref{condconvergencecalls} hold setting $f^n_t(x_t):=(\left|x_t%
\right|-n)^+,t=0,\dots,T$ as in Example \ref{examplecalls}. To check (i)
observe that for every $t=0,\dots,T$ and $\varphi_t\in C_t$

\begin{equation}
\begin{split}
-\infty&<-\int_{K_{t}}v_{t}\left(\left\|
\varphi_t\right\|(1+\left|x_{t}\right|)\right)\,\mathrm{d}\widehat{Q }%
_{t}(x_{t})=\int_{K_{t}}u_{t}\left(-\left\| \varphi_t\right\|
(1+\left|x_{t}\right|)\right)\,\mathrm{d}\widehat{Q }_{t}(x_{t}) \\
&\leq \int_{K_{t}}u_{t}\left(\varphi_t\right)\,\mathrm{d}\widehat{Q }%
_{t}(x_{t}) \leq U_{\widehat{Q}_t}(\varphi_t)\overset{(\circledast )}{\leq }%
\int_{K_t}\varphi_t\mathrm{d}\widehat{Q}_t\leq\int_{K_{t}}\left\|
\varphi_t\right\|_t (1+\left|x_{t}\right|)\,\mathrm{d}\widehat{Q }%
_{t}(x_{t}) <+\infty
\end{split}
\label{Uisfinitebelow}
\end{equation}
where in $(\circledast )$ we argued as in \eqref{usegradientu} and the
finiteness of the last term comes from the fact that $\widehat{Q}\in\mathrm{%
Prob}^1(\Omega)$. Notice that we just showed also that $\widehat{Q}%
_t\in\partial U_{\widehat{Q}_t}(0)$. Concavity and cash additivity can be checked along the lines of the proof of Lemma \ref{remfrominftytofiniteness} Item (2)-(3). Coming to (ii), from Proposition \ref{propvstar} and Fenchel inequality for every $Q\in\mathrm{Prob}^1(\Omega)$
and $t=0,\dots,T$ 
\begin{equation}
\begin{split}
\mathcal{D}_{v_{t}^{\ast },\widehat{Q }_{t}}(Q_t )&\overset{%
\eqref{explicitdual}}{=}\sup_{\varphi _{t}\in \mathcal{C}_{b}(K_{t})}\left(
\int_{K_{t}}\varphi _{t}(x_{t})\,\mathrm{d}Q_t
(x_{t})-\int_{K_{t}}v_{t}(\varphi _{t}(x_{t}))\,\mathrm{d}\widehat{Q }%
_{t}(x_{t})\right) \\
&\leq\sup_{\varphi _{t}\in {C}_t}\left( \int_{K_{t}}\varphi _{t}(x_{t})\,%
\mathrm{d}Q_t (x_{t})-\int_{K_{t}}v_{t}(\varphi _{t}(x_{t}))\,\mathrm{d}%
\widehat{Q }_{t}(x_{t})\right) \leq \mathcal{D}_{v_{t}^{\ast },\widehat{Q }%
_{t}}(Q_t )\,.
\end{split}
\label{conjugatetocontsublinear}
\end{equation}

Since $\mathcal{E}_t\subseteq C_t$ the proof of (ii) is then concluded
observing that 
for every $\varphi_t\in C_t$ there
exists a sequence of $(\varphi_t^n)_n\subseteq C_t$, with each $\varphi_t^n$
in the form \eqref{mixcalls}, such that $\varphi_t^n \rightarrow \varphi_t$
pointwise on $K_t$ and $\sup_n\left\| \varphi_t^n\right\|_t<+\infty$.
Indeed, this in turns implies, using $\int_{K_{t}}v_{t}\left(\alpha
(1+\left|x_{t}\right|)\right)\,\mathrm{d}\widehat{Q }_{t}(x_{t})<+\infty\,%
\forall\alpha>0$, that one has 
\begin{equation*}
\int_{K_{t}}\varphi _{t}(x_{t})\,\mathrm{d}Q
(x_{t})-\int_{K_{t}}v_{t}(\varphi _{t}(x_{t}))\,\mathrm{d}\widehat{Q }%
_{t}(x_{t})=\lim_n\left(\int_{K_{t}}\varphi^n _{t}(x_{t})\,\mathrm{d}Q
(x_{t})-\int_{K_{t}}v_{t}(\varphi^n _{t}(x_{t}))\,\mathrm{d}\widehat{Q }%
_{t}(x_{t})\right)
\end{equation*}
for every $Q\in \mathrm{Prob}^1(\Omega)$, by Dominated Convergence Theorem.
Finally, (iii) is verified observing that 
\begin{equation*}
0=\int_\Omega0\mathrm{d}\widehat{Q}\geq S_t(0) \geq S_t(-f_t^n)\geq
\int_{\Omega}u_t\left(-f^n_t\right)\mathrm{d}\widehat{Q}\rightarrow_n 0\,
\end{equation*}
where the limit in RHS is motivated by Dominated Convergence Theorem
together with \eqref{condizOrlicz}
\end{proof}

\begin{remark}
\label{marginalsnotofmart} Observe that Corollary \ref%
{corolldivinducednoncpt} would remain valid replacing $U_{\widehat{Q}%
_{t}}(\cdot )$ with $S_{\widehat{Q}_{t}}^{{u}_{t}}(\cdot ):=\sup_{\alpha \in 
{\mathbb{R}}}\left( \int_{K_{t}}u_{t}(\cdot +\alpha )\mathrm{d}\widehat{Q}%
_{t}-\alpha \right) $ for general $\widehat{Q}_{t}\in \mathrm{Prob}%
^{1}(K_{t})$ without requesting these are marginals of a martingale measure.
Indeed, the only point where the martingale property played a role was in
showing $(\circledast )$, which remains valid if one replaces $U_{\widehat{Q}%
_{t}}(\cdot )$ with $S_{\widehat{Q}_{t}}^{{u}_{t}}(\cdot )$ (the argument is
essentially identical). A prominent role was given to $U_{\widehat{Q}_{t}}$
since it serves as an example of stock additive functional (see Example \ref%
{exutstockadd}). 
\end{remark}

Just as we obtained Corollary \ref{CorBeig} from Corollary 
\ref{corolldivinduced}, using the linear utility functions $u_t (x_t)=x_t$, we now deduce the following result from Corollary \ref{corolldivinducednoncpt} \

\begin{corollary}[\protect\cite{Beig} Theorem 1.1 and Corollary 1.2]
\label{corbeiggeneralnoncpt}
Take $d=1$, $K_{0}=\{x_0\}$ for some $x_0\in{%
\mathbb{R}}$ and let $K_{1},\dots ,K_{T}\subseteq {\mathbb{R}}$ be closed subsets of $\mathbb{R}$.
Take for each $t=0,\dots,T$ the vector space $\mathcal{E}_t\subseteq C_t$
of functions in the form \eqref{mixcalls}, let $\mathcal{E}=\mathcal{E}_{0}\times \dots \times \mathcal{E}_{T}$ and fix a $\widehat{Q}\in\mathrm{Mart}(\Omega)$. 
Then for any $c:\Omega\rightarrow(-\infty,+\infty]$ lower semicontinuous
and satisfying \eqref{controlfrombelow}
\begin{equation}
\begin{split}
\mathfrak{P}(c)& =\sup \left\{ \sum_{t=0}^{T}{E}_{\widehat{Q}_{t}}(\varphi
_{t})\mid \varphi \in \mathcal{E}\text{ is s.t. }\exists \Delta \in \mathcal{%
H}\text{ with }\sum_{t=0}^{T}\varphi _{t}(x_{t})+I^{\Delta }(x)\leq c(x)%
\text{ }\forall x\in \Omega \right\} \\
& =\inf_{Q\in \mathrm{Mart}(\widehat{Q}_{0},\dots ,\widehat{Q}%
_{T})}E_{Q}[c]\,
\end{split}
\label{dualrepreBeing}
\end{equation}
and if $\mathfrak{P}(c)<+\infty$, a minimum point exists for the infimum in \eqref{dualrepreBeing}.
\label{CorNew}
\end{corollary}

\begin{example}
\label{exampleutilitiestoMOT} We now come to the study of the convergence to the MOT problem. Take $u_{0},\dots ,u_{T}:{\mathbb{R}}%
\rightarrow {\mathbb{R}}$ satisfying Assumption \ref{aU1}, and assume
additionally that these are all differentiable in $0$ (which implies that $%
\{1\}=\partial u_{0}(0)=\dots =\partial u_{T}(0)$). 
Observe that setting $u_{t}^{n}(x):=nu_{t}\left( \frac{x}{n}\right) $, $x\in 
{\mathbb{R}},t=0,\dots ,T$ the functions $u_{0}^{n},\dots ,u_{T}^{n}$ still
satisfy Assumption \ref{aU1}. Moreover $(v_{t}^{n})^{\ast }(y)=\sup_{x\in {%
\mathbb{R}}}(u_{t}^{n}(x)-xy))=nv_{t}^{\ast }(y),\,\,\,y\in {\mathbb{R}}$.
Since $u_{t}(0)=0$, we have $v_{t}^{\ast }\geq 0$ and as a consequence 
\begin{equation*}
\sup_{n}(v_{t}^{n})^{\ast }(y)=%
\begin{cases}
0\,\, & \text{if }v_{t}^{\ast }(y)=0 \\ 
+\infty \,\, & \text{otherwise}%
\end{cases}%
\,.
\end{equation*}%
Moreover $v_{t}^{\ast }(y)=0\Rightarrow y\in \partial u_{t}(0)=\{1\}$.
Consider the set $\mathcal{A}^{\varepsilon }$ of $\varepsilon $-martingale
measures defined in equation \eqref{polarepsilonmarts}, take $\widehat{Q}\in 
\mathrm{Mart}(\Omega )$ and a sequence $\varepsilon _{n}\downarrow _{n}0$. Using 
\eqref{convergenceepsmarttomart}, for every $Q\in \mathrm{Prob}^{1}(\Omega )$
\begin{equation*}
\sum_{t=0}^{T}\mathcal{D}_{(v_{t}^{n})^{\ast },\widehat{Q}%
_{t}}(Q_{t})+\sigma _{\mathcal{A}^{\varepsilon _{n}}}(Q_{t})\uparrow _{n}%
\mathcal{D}_{\infty }(Q)+\sigma _{\mathcal{A}_{\infty }}(Q)
\end{equation*}%
where 
\begin{equation*}
\mathcal{D}_{\infty }(Q)+\sigma _{\mathcal{A}_{\infty }}(Q)=%
\begin{cases}
0\,\,\, & \text{if }Q\in \mathrm{Mart}(\Omega )\text{ and }Q_{t}\equiv 
\widehat{Q}_{t}\,\forall t=0,\dots ,T \\ 
+\infty \,\,\, & \text{otherwise}%
\end{cases}%
\,.
\end{equation*}%
As a consequence, by Proposition \ref{propstability} $$\inf_{Q\in \mathrm{Mart}(\Omega )}\left( E_{Q}\left[ c\right]
+\sum_{t=0}^{T}\mathcal{D}_{(v^n_{t})^{\ast },\widehat{Q}_{t}}(Q_{t})\right)\rightarrow_n \inf_{Q\in \mathrm{Mart}(\widehat{Q}_{0},\dots ,\widehat{Q}%
_{T})}E_{Q}[c]\,.$$
\end{example}
\subsection{Dual representation for generalized OCE associated to the indirect utility function}
\label{GOCE}
In the following we will treat a slightly different problem, which however helps
understanding how also the extreme case $\mathcal{E}_{t}={C}_{0:t},t=0,\dots
,T$ is of interest. Corollary \ref{maincor} yields the following dual robust
representation of the generalized Optimized Certainty Equivalent associated
to the indirect utility function. We stress here the fact that, again, $%
\widehat{Q}\in \mathrm{Mart}(\Omega )$ is a fixed martingale measure, but we
will not focus anymore on its marginals only, as will become clear in the
following.

\begin{proposition}
Take $u:{\mathbb{R}}\rightarrow{\mathbb{R}}$ such that $u_0=\dots,u_T:=u$
satisfy Assumption \ref{aU1} and let $v^{\ast }$ be defined in %
\eqref{cnvxconjvstar} with $u$ in place of $u_t$. Take $\widehat{Q}\in%
\mathrm{Mart}(\Omega)$, such that \eqref{condizOrlicz} holds.

Let $U_{\widehat{Q}}^{\mathcal{H}}:C_{0:T}\rightarrow {\mathbb{R}}$ be the
associated indirect utility 
\begin{equation*}
U_{\widehat{Q}}^{\mathcal{H}}(\varphi ):=\sup_{\Delta \in \mathcal{H}%
}\mathbb{E}_{\widehat{\probq} }[u(\varphi +I^{\Delta })]\,.
\end{equation*}%
and $S^{U_{\widehat{Q}}^{\mathcal{H}}}$ be the associated Optimized
Certainty Equivalent defined according to \eqref{defSU}, namely 
\begin{equation*}
S^{U_{\widehat{Q}}^{\mathcal{H}}}(\varphi):=\sup_{\beta\in{\mathbb{R}}%
}\left(U_{\widehat{Q}}^{\mathcal{H}}(\varphi+\beta)-\beta\right)\,\,\,\,\varphi%
\in C_{0:T}\,.
\end{equation*}
Then for every $c\in C_{0:T}$  the dual representation holds:
\begin{equation*}
S^{U_{\widehat{Q}}^{\mathcal{H}}}(c)=\inf_{Q\in \mathrm{Mart}(\Omega
)}\left( \mathbb{E}_\probq[c]+\mathcal{D}_{\widehat{Q}}(Q)\right)
\end{equation*}

where for $Q \in\mathrm{Prob}^1(\Omega)$ 
\begin{equation*}
\mathcal{D}_{\widehat{Q}}(Q):=%
\begin{cases}
\int_{\Omega }v^{\ast }\left( \frac{\mathrm{d}Q}{\mathrm{d}\widehat{Q}}%
\right) \,\mathrm{d}\widehat{Q} & \,\,\text{ if }Q\ll \widehat{Q} \\ 
+\infty & \,\,\text{ otherwise}%
\end{cases}%
.  \label{SSS}
\end{equation*}
\end{proposition}

\begin{proof}
Take $\mathcal{E}_{t}=C_{0:t}$ for $t=0,\dots ,T$. Define  $U(\psi ):=U_{%
\widehat{Q}}^{\mathcal{H}}\left( \sum_{t=0}^{T}\psi _{t}\right) $ for  $\psi \in 
\mathcal{E}=\mathcal{E}_{0}\times ...\times \mathcal{E}_{T}$. From %
\eqref{condizOrlicz}, it follows that $U(\psi )>-\infty $ for any $\psi \in 
\mathcal{E}$, similarly to the argument in \eqref{Uisfinitebelow}. Since $%
\widehat{Q}\in \mathrm{Mart}(\Omega )$ and $u(x)\leq x$ for all $x\in {%
\mathbb{R}}$ we also have $U(\psi )\leq \int_\Omega \sum_{t=0}^T\psi_t%
\mathrm{d}\widehat{Q}<+\infty $ (since $\mathrm{Mart}(\Omega)\subseteq 
\mathrm{Probq}^1(\Omega)$). Moreover it is easy to verify that for any $Q\in 
\mathrm{Mart}(\Omega )$ we have 
\begin{equation*}
\mathcal{D}(Q):=\sup_{\psi \in \mathcal{E}}\left( U(\psi )-\int_{\Omega
}\left( \sum_{t=0}^{T}\psi _{t}\right) \mathrm{d}Q\right) =\sup_{\varphi \in
C_{0:T}}\left( \int_{\Omega }u(\varphi )\,\mathrm{d}\widehat{Q}-\int_{\Omega
}\varphi \,\mathrm{d}Q\right) \,
\end{equation*}%
and arguing as in \eqref{conjugatetocontsublinear} we get $\mathcal{D}(Q)=%
\mathcal{D}_{\widehat{Q}}(Q)$. The assumptions of Corollary \ref{maincor}
are satisfied: take indeed $f^n_t(x_0,\dots,x_t):= -(\left|x_t\right|-n)^+$
as in Example \ref{examplecalls}, so that \eqref{ineqoutsidecpt} holds.
Then, using the fact that $U(\psi )\leq \int_\Omega \sum_{t=0}^T\psi_t%
\mathrm{d}\widehat{Q}$, we have 
\begin{equation*}
0=\int_\Omega0\mathrm{d}\widehat{Q}\geq U(0,\dots,0) \geq
U(-f^n_0,\dots,-f^n_T)\geq \int_{\Omega}u\left(-\sum_{t=0}^Tf^n_t\right)%
\mathrm{d}\widehat{Q}\rightarrow_n 0\,
\end{equation*}
where the limit in RHS is motivated by Dominated Convergence Theorem
together with \eqref{condizOrlicz}. This imples that also %
\eqref{condconvergencecalls} is satisfied. $\mathcal{A}=\mathcal{I}$ is
clearly a convex cone with $0\in\mathcal{A}$, and $U(0,\dots,0)=0$ (as can
be easily verified). Hence, Assumption \ref{asscallstozero}.(i)-(ii) are
satisfied.

By Theorem \ref{mainEMOTtheoremgeneral}, as a consequence, we have 
\begin{equation*}
\inf_{Q\in \mathrm{Mart}(\Omega )}\left( E_{Q}\left[ c(X)\right] +\mathcal{D}%
_{\widehat{Q}}(Q)\right) =\inf_{Q\in \mathrm{Mart}(\Omega )}\left( E_{Q}%
\left[ c(X)\right] +\mathcal{D}(Q)\right) =\sup_{\Delta \in \mathcal{H}%
}\sup_{\psi \in \mathbf{\Phi }_{I^\Delta }(c)}S^{U}\left( \psi \right) \,.
\end{equation*}%
Observe now that $S^{U}$ satisfies 
\begin{equation*}
S^{U}(\psi ):=\sup_{\lambda \in {\mathbb{R}}^{T+1}}\left( U(\psi +\lambda
)-\sum_{t=0}^{T}\lambda _{t}\right) =\sup_{\lambda \in {\mathbb{R}}%
^{T+1}}\left( U_{\widehat{Q}}^{\mathcal{H}}\left( \sum_{t=0}^{T}\psi
_{t}+\sum_{t=0}^{T}\lambda _{t}\right) -\sum_{t=0}^{T}\lambda _{t}\right)
\end{equation*}

\begin{equation*}
=\sup_{\beta\in{\mathbb{R}}}\left(U_{\widehat{Q}}^\mathcal{H}%
\left(\sum_{t=0}^T\psi_t+\beta\right)-\beta\right)=:S^{U_{\widehat{Q}}^\mathcal{H%
}}\left(\sum_{t=0}^T\psi_t\right)\,.
\end{equation*}
It can be asily verified that $S^{U_{\widehat{Q}}^\mathcal{H}}$ is real
valued on $C_{0:T}$, since $\int_\Omega\varphi\mathrm{d}\widehat{Q}\geq
S^{U_{\widehat{Q}}^\mathcal{H}}(\varphi)\geq U_{\widehat{Q}}^\mathcal{H}%
(\varphi) $, and nondecreasing. Furthermore, $S^{U_{\widehat{Q}}^\mathcal{H}%
}:C_{0:T}\rightarrow {\mathbb{R}}$ is (IA). Hence 
\begin{equation*}
\sup_{\Delta \in \mathcal{H}}\sup_{\psi \in \mathbf{\Phi }_{I^\Delta
}(c)}S^{U_{\widehat{Q}}^\mathcal{H}}\left(\sum_{t=0}^T\psi_t\right)
=\sup_{\Delta \in \mathcal{H}}\sup_{\psi \in \mathbf{\Phi }_{I^\Delta
}(c)}S^{U_{\widehat{Q}}^\mathcal{H}}\left(\sum_{t=0}^T\psi_t+I^\Delta\right)
=S^{U_{\widehat{Q}}^\mathcal{H}}(c)
\end{equation*}%
by definition of $\mathbf{\Phi }_{I^\Delta }(c)$ and since $c\in C_{0:T}$.
\end{proof}

\appendix

\section{Appendix}
\label{appndx}
Let $\mathbb{X}$ be a metric space and $C(\mathbb{X})$ be the space of
continuous functions on $\mathbb{X}$. Let $\mathcal{E}_{t}\subseteq C(%
\mathbb{X}),$ $t=0,\dots ,T$, be a vector space and $\mathcal{E}=\mathcal{E}%
_{0}\times \dots \times \mathcal{E}_{T}\,$, let $U:\mathcal{E}\rightarrow
\lbrack -\infty ,+\infty )$ be proper and concave and define $S^{U}:\mathcal{%
E}\rightarrow \lbrack -\infty ,+\infty ]$ as in \eqref{defSU}. 

\begin{lemma}
\label{remalmostconcave} Assume the usual convention $\infty \cdot 0=0\cdot
\infty =0$. Then $S^{U}$ is concave on the convex set 
\begin{equation*}
\mathrm{dom}(S^{U}):=\{\varphi \in \mathcal{E}\mid S^{U}(\varphi )>-\infty
\},
\end{equation*}%
that is: if $\varphi ^{1},\varphi ^{2}\in \mathrm{dom}(S^{U})$ then for
every $0\leq \alpha \leq 1$ we have $\alpha \varphi ^{1}+(1-\alpha )\varphi
^{2}\in \mathrm{dom}(S^{U})$ and 
\begin{equation}  \label{almostconcave}
S^{U}(\alpha \varphi ^{1}+(1-\alpha )\varphi ^{2})\geq \alpha S^{U}(\varphi
^{1})+(1-\alpha )S^{U}(\varphi ^{2})\,\,\,\,\,\forall \,\varphi ^{1},\varphi
^{2}\in \mathrm{dom}(S^{U})
\end{equation}%
where the expression makes sense and holds even if either $S^{U}(\varphi
^{1})$ or $S^{U}(\varphi ^{2})$ or both are equal to $+\infty $. Finally, $%
S^U$ satisfies:

\begin{equation*}
S^U(\varphi+\alpha)=S^U(\varphi)+\sum_{t=0}^T\alpha_t\,\,\,\forall\,\varphi%
\in\mathcal{E},\,\forall\, \alpha\in{\mathbb{R}}^{T+1}\,.
\end{equation*}
\end{lemma}

\begin{proof}
Since the claim is trivially true for $\alpha =0,1$ (assuming the usual
convention $\infty \cdot 0=0\cdot \infty =0$), we can assume $0<\alpha <1$.
Observe that since $\varphi ^{1},\varphi ^{2}\in \mathrm{dom}(S^{U})$ for
some $\beta ^{1},\beta ^{2}\in {\mathbb{R}}^{T+1}$ we have 
\begin{align*}
-\infty & <\alpha U(\varphi ^{1}+\beta ^{1})+(1-\alpha )U(\varphi ^{2}+\beta
^{2})-\left[ \alpha \sum_{t=0}^{T}\beta _{t}^{1}+(1-\alpha )\sum_{t=0}^{T}\beta
_{t}^{2}\right]  \\
& \leq U(\alpha \varphi ^{1}+(1-\alpha )\varphi ^{2}+\left[ \alpha \beta
^{1}+(1-\alpha )\beta ^{2}\right] )-\left[ \alpha \sum_{t=0}^{T}\beta
_{t}^{1}+(1-\alpha )\sum_{t=0}^{T}\beta _{t}^{2}\right]  \\
& \leq \sup_{\beta \in {\mathbb{R}}^{T+1}}\left( U(\alpha \varphi
^{1}+(1-\alpha )\varphi ^{2}+\beta )-\sum_{t=0}^{T}\beta _{t}\right)
=:S^{U}(\alpha \varphi ^{1}+(1-\alpha )\varphi ^{2})\,.
\end{align*}%
Now if both $S^{U}(\varphi ^{1}),S^{U}(\varphi ^{2})<+\infty $ it is enough
to take suprema over $\beta ^{1},\beta ^{2}\in {\mathbb{R}}^{T+1}$ to get %
\eqref{almostconcave}. If otherwise either $S^{U}(\varphi ^{1})=+\infty $ or 
$S^{U}(\varphi ^{2})=+\infty $, taking the supremum over $\beta ^{1},\beta
^{2}\in {\mathbb{R}}^{T+1}$ yields $S^{U}(\alpha \varphi ^{1}+(1-\alpha
)\varphi ^{2})=+\infty =\alpha S^{U}(\varphi ^{1})+(1-\alpha )S^{U}(\varphi
^{2})$, and again we get \eqref{almostconcave}. Also, \eqref{almostconcave}
clearly implies that if $\varphi ^{1},\varphi ^{2}\in \mathrm{dom}(S^{U})$
then for every $0\leq \alpha \leq 1$ we have $\alpha \varphi ^{1}+(1-\alpha
)\varphi ^{2}\in \mathrm{dom}(S^{U})$. As to the last property, fix $\varphi
\in \mathcal{E}$. There are three cases cases: either (i) $U(\varphi +\beta
)=-\infty \,\forall \beta \in {\mathbb{R}}^{T+1}$, so that $S^{U}(\varphi
+\alpha )=S^{U}(\varphi )=S^{U}(\varphi )+\sum_{t=0}^{T}\alpha _{t}=-\infty $%
, or (ii) $U(\varphi +\beta )=+\infty $ for some $\beta \in {\mathbb{R}}^{T+1}$
in which case $S^{U}(\varphi +\alpha )=S^{U}(\varphi )=S^{U}(\varphi
)+\sum_{t=0}^{T}\alpha _{t}=+\infty $, or (iii) we can write 
\begin{align*}
S^{U}(\varphi +\alpha )& =\sup \left\{ U\left( \varphi +\alpha +\beta \right)
-\sum_{t=0}^{T}\beta _{t}\mid \beta \in {\mathbb{R}}^{T+1},U(\varphi +\beta
+\alpha )>-\infty \right\}  \\
& =\sup \left\{ U\left( \varphi +(\alpha +\beta )\right) -\sum_{t=0}^{T}(\beta
_{t}+\alpha _{t})+\sum_{t=0}^{T}\alpha _{t}\mid \beta \in {\mathbb{R}}%
^{T+1},U(\varphi +(\beta +\alpha ))>-\infty \right\}  \\
& =\sup \left\{ U\left( \varphi +\eta \right) -\sum_{t=0}^{T}\eta
_{t}+\sum_{t=0}^{T}\alpha _{t}\mid \eta \in {\mathbb{R}}^{T+1},U(\varphi
+\eta )>-\infty \right\}  \\
& \overset{(\star )}{=}\sup \left\{ U\left( \varphi +\eta \right)
-\sum_{t=0}^{T}\eta _{t}\mid \eta \in {\mathbb{R}}^{T+1},U(\varphi +\eta
)>-\infty \right\} +\sum_{t=0}^{T}\alpha _{t} \\
& =\sup_{\eta \in {\mathbb{R}}^{T+1}}\left( U(\varphi +\eta
)-\sum_{t=0}^{T}\eta _{t}\right) +\sum_{t=0}^{T}\alpha _{t}=S^{U}(\varphi
)+\sum_{t=0}^{T}\alpha _{t}
\end{align*}%
where in $(\star )$ we are exploiting the fact that we are in the case of $%
U(\varphi +\beta )<+\infty \,\forall \beta \in {\mathbb{R}}^{T+1}$.
\end{proof}

\begin{lemma}
\label{lemmaintrostockaddfunctions} Take closed sets $K_{1},\dots
,K_{T}\subseteq {\mathbb{R}}$, and suppose that $K_{0}=\{x_{0}\}$ and $\text{%
card}(K_{t+1})\geq \text{card}(K_{t})$ for every $t=0,\dots ,T-1$. Take $%
\mathcal{E}=\mathcal{E}_{0}\times \dots \times \mathcal{E}_{T}$ for vector
subspaces $\mathcal{E}_{t}\subseteq C_{t}(K_{t})$ such that $%
X_{t}\in \mathcal{E}_{t}$ and $\mathcal{E}_{t}+{\mathbb{R}}=\mathcal{E}_{t}$%
, for $t=0,\dots ,T$. Suppose there exist $\varphi ,\psi \in \mathcal{E}$
and $\Delta \in \mathcal{H}$, where $\mathcal{H}$ is defined in %
\eqref{strategies}, such that $\sum_{t=0}^{T}\varphi _{t}=\sum_{t=0}^{T}\psi
_{t}+I^{\Delta }$. Then there exist constants $k_{0},\dots
,k_{T},h_{0},\dots ,h_{T}\in {\mathbb{R}}$ such that for each $t=0,\dots ,T$ 
$\psi _{t}(x_{t})=\varphi _{t}(x_{t})+k_{t}x_{t}+h_{t},\,\,\forall x_{t}\in
K_{t}$. In particular for $S_{t}:\mathcal{E}_{t}\rightarrow {\mathbb{R}}%
,\,t=0,\dots ,T$ stock additive functionals we have 
\begin{equation*}
\sum_{t=0}^{T}S_{t}(\varphi _{t})=\sum_{t=0}^{T}S_{t}(\psi _{t})\,.
\end{equation*}%
and for $\mathcal{V}:=\sum_{t=0}^{T}\mathcal{E}_{t}+\mathcal{I}$ (see %
\eqref{StoIntegra}) the map 
\begin{equation*}
v=\sum_{t=0}^{T}\varphi _{t}+I^{\Delta }\mapsto
S(v):=\sum_{t=0}^{T}S_{t}(\varphi _{t})
\end{equation*}%
is well defined on $\mathcal{V}$, (CA) and (IA).
\end{lemma}

\begin{proof}
$\,$

\textbf{Step 1}: we prove that if $\sum_{t=0}^{T}\varphi
_{t}=\sum_{t=0}^{T}\psi _{t}+I^{\Delta }$ then $\Delta =[\Delta _{0},\dots
,\Delta _{T-1}]\in \mathcal{H}$ is a deterministic vector $\Delta \in {%
\mathbb{R}}^{T}$. If $\text{card}(K_{T})=1$ this is trivial. We can then
suppose $\text{card}(K_{T})\geq 2$ We see that 
\begin{equation*}
\varphi _{T}(x_{T})-\psi _{T}(x_{T})=\sum_{t=0}^{T-1}(\psi (x_{t})-\varphi
_{t}(x_{t}))+\sum_{t=0}^{T-2}\Delta _{t}(x_{0},\dots ,x_{t})(x_{t+1}-x_{t})+
\end{equation*}%
\begin{equation*}
+\Delta _{T-1}(x_{0},\dots ,x_{T-1})(x_{T}-x_{T-1})=f(x_{0},\dots
,x_{T-1})+\Delta _{T-1}(x_{0},\dots ,x_{T-1})x_{T}
\end{equation*}%
for some function $f$. If $\Delta _{T-1}$ were not constant, on two points
it would assume values $a\neq b$, with corresponding values of $f$ that we
call $f_{a},f_{b}$. Then $f_{a}+ax_{T}=f_{b}+bx_{T}$ has a unique solution,
contradicting the fact that all the equalities need to hold on the whole $%
K_{0},\dots ,K_{T}$ and in particular for two different values of $x_{T}$.
We proceed one step backward. If $\text{card}(K_{T-1})=1$, the claim
trivially follows, given our previous step. If $\text{card}(K_{T-1})\geq 2$,
similarly to the previous computation 
\begin{equation*}
\varphi _{T-1}(x_{T-1})-\psi _{T-1}(x_{T-1})=\sum_{s\neq T-1}(\psi
_{s}(x_{s})-\varphi _{s}(x_{s}))+\sum_{t=0}^{T-3}\Delta _{t}(x_{0},\dots
,x_{t})(x_{t+1}-x_{t})+
\end{equation*}%
\begin{equation*}
+\Delta _{T-2}(x_{0},\dots ,x_{T-2})(x_{T-1}-x_{T-2})+\Delta
_{T-1}(x_{T}-x_{T-1})
\end{equation*}%
\begin{equation*}
=f(x_{s},s\neq T-1)+(\Delta _{T-2}(x_{0},\dots ,x_{T-2})-\Delta
_{T-1})x_{T-1}\,.
\end{equation*}%
An argument similar to the one we used in the previous time step shows that $%
\Delta _{T-2}(x_{0},\dots ,x_{T-2})-\Delta _{T-1}$ is constant, hence so is $%
\Delta _{T-2}$. Our argument can be clearly be iterated up to $\Delta _{0}$.

\textbf{Step 2}: we prove existence of the vectors $k,h\in {\mathbb{R}}%
^{T+1} $, as stated in the Lemma. From Step 1 it is clear that there exist
constants $k_{0},\dots ,k_{T}$ such that $I^{\Delta
}(x)=\sum_{t=0}^{T}k_{t}x_{t}$. Hence $\sum_{t=0}^{T}\varphi
_{t}(x_{t})=\sum_{t=0}^{T}(\psi _{t}(x_t)+k_{t}x_{t})$ for all $x\in \Omega $%
, which yields for each $t=0,\dots ,T$ that $\varphi _{t}(x_{t})-(\psi
_{t}(x_t)+k_{t}x_{t})$ does not depend on $x_{t}$, hence is constant, call
it $-h_{t}$. Then $k_{0},\dots ,k_{T},h_{0},\dots ,h_{T}\in {\mathbb{R}}$
satisfy our requirements. The last claim $\sum_{t=0}^{T}S_{t}^{U}(\varphi
_{t})=\sum_{t=0}^{T}S_{t}^{U}(\psi _{t})$ is then an easy consequence of
stock additivity.

\textbf{Step 3}: well posedness and properties of $S$. Observe that whenever 
$\varphi ,\psi \in \mathcal{E},\,\Delta ,H\in \mathcal{H}$ are given with $%
\sum_{t=0}^{T}\varphi _{t}+I^{\Delta }=\sum_{t=0}^{T}\psi _{t}+I^{H}$ we
have by Steps 1-2 that $\sum_{t=0}^{T}S_{t}^{U}(\varphi
_{t})=\sum_{t=0}^{T}S_{t}^{U}(\psi _{t})\,.$ As a consequence, $S$ is well
defined. Cash Additivity is inherited from $S_{0},\dots ,S_{T}$ while
Integral Additivity is trivial from the definition.
\end{proof}

\begin{proposition}
\label{propcontrolcones} There exist $a=a(d,T), \beta=\beta(d,T)>0$ such
that for every $A>1$ 
\begin{equation*}
1+\sum_{s=0}^T\sum_{j=1}^d\left|x^j_s\right|\leq a\sum_{s=0}^T
\sum_{j=1}^d f^{\frac{A}{\beta}}_{j,s}(x^j_s)\,\,\,\,\forall
[x_0,\dots,x_T]\in ({\mathbb{R}}^d)^{T+1}\setminus ([-A,A]^d)^{T+1}
\end{equation*}
where $f^\alpha_{j,s}$ is defined in \eqref{defeqcalls}.
\end{proposition}

\begin{proof}
Observe that $f^\frac{A}{\beta}_{j,s}(x^j_s)=\left(\left|x^j_s\right|-\frac{A%
}{\beta}\right)^+$. Fix $x\in ({\mathbb{R}}^d)^{T+1}\setminus
([-A,A]^d)^{T+1}$ Define

\begin{align*}
I(x)&:=\left\{{j,s}\in\{1,\dots,d\}\times\{0,\dots,T\}\mid
\left|x_s^j\right|>A\right\} \\
I^c(x)&:=\left\{{j,s}\in\{1,\dots,d\}\times\{0,\dots,T\}\mid
\left|x_s^j\right|\leq A\right\}\,.
\end{align*}

Then $\{1,\dots ,d\}\times \{0,\dots ,T\}=I(x)\cup I^{c}(x)$ and $I(x)\neq
\emptyset $. Moreover for $\beta>1$
\begin{align*}
& a\sum_{s=0}^T\sum_{j=1}^{d}f_{j,s}^{\frac{A}{\beta }}(x_{s}^{j})-\left(
1+\sum_{s=0}^T\sum_{j=1}^{d}\left\vert x_{s}^{j}\right\vert \right)  \\
& =a\sum_{s=0}^T\sum_{j=1}^{d}\left( \left\vert x_{s}^{j}\right\vert -%
\frac{A}{\beta }\right) ^{+}-\left( 1+\sum_{s=0}^T\sum_{j=1}^{d}\left\vert
x_{s}^{j}\right\vert \right)  \\
& =\sum_{(j,s)\in I(x)}\left( a\left( \left\vert x_{s}^{j}\right\vert -\frac{%
A}{\beta }\right) -\left\vert x_{s}^{j}\right\vert \right) +\sum_{(j,s)\in
I^{c}(x)}\left( \left( \left\vert x_{s}^{j}\right\vert -\frac{A}{\beta }%
\right) ^{+}-\left\vert x_{s}^{j}\right\vert \right) -1 \\
& \geq (a-1)\sum_{(j,s)\in I(x)}\left\vert x_{s}^{j}\right\vert -\frac{aA}{%
\beta }|I(x)|-\sum_{(j,s)\in I^{c}(x)}\left\vert x_{s}^{j}\right\vert -1 \\
& \geq (a-1)A-\frac{aA}{\beta }d(T+1)-d(T+1)A-1 \\
& =aA\left( 1-d\frac{T+1}{\beta }\right) -\left[ A+d(T+1)A+1\right] 
\end{align*}%
Then, selecting e.g. $\beta =2d(T+1)$, we can solve the inequality 
\begin{equation*}
aA\left( 1-\frac{1}{2}\right) -\left[ A+d(T+1)A+1\right] \geq 0
\end{equation*}%
getting $a\geq 2d(T+2)+\frac{2}{A}$. This yields a possible selection $%
a=2d(T+2)+2$ too, depending only on the dimensions $d,T$.
\end{proof}

\begin{lemma}
\label{lemmamargmartcptgen} The set $\mathrm{Mart}_t(K_t)$ defined in \eqref{marginalmartgen} is $\sigma(%
\mathrm{ca}(K_t),\mathcal{C}_b(K_t))$-compact.
\end{lemma}

\begin{proof}
We see that $\mathrm{Mart}(\Omega )$ is a $\sigma (\mathrm{ca}(\Omega ),%
\mathcal{C}_{b}(\Omega ))$-closed subset of the $\sigma (\mathrm{ca}(\Omega
),\mathcal{C}_{b}(\Omega ))$-compact set $\mathrm{Prob}(\Omega )$ (which is
compact since $\Omega $ is a compact Polish Space, see \cite{Aliprantis}
Theorem 15.11), hence it is compact himself. $\mathrm{Mart}_{t}(K_{t})$ is
then the image of a compact set via the marginal map $\gamma \mapsto \gamma
_{t}$ which is $\sigma (\mathrm{ca}(\Omega ),\mathcal{C}_{b}(\Omega
))-\sigma (\mathrm{ca}(K_{t}),\mathcal{C}_{b}(K_{t}))$ continuous, hence it
is $\sigma (\mathrm{ca}(K_{t}),\mathcal{C}_{b}(K_{t}))$ compact.
\end{proof}

\begin{proposition}
\label{propdaniellstone} Let $L\in (C_\psi)^*$ be continuous, linear and
positive. Suppose that for every $\varepsilon>0$ there exists a compact $%
K_\varepsilon\subseteq \mathbb{X}$ such that 
\begin{equation*}
\varphi\in C_\psi, \varphi|_{K_\varepsilon}\equiv 0\Longrightarrow
\left|\langle \varphi,L\rangle\right|\leq \varepsilon \left\|
\varphi\right\|_\psi\,.
\end{equation*}
Then for every sequence $(c_n)_n\in C_\psi$ with $c_n\downarrow_n 0$
pointwise on $\mathbb{X}$ it holds that $L(c_n)\downarrow 0$.
\end{proposition}

\begin{proof}
Fix $\varepsilon>0$ and take the associated compact $K_\varepsilon$. By
Dini's Lemma, $\sup_{x\in K_\varepsilon} c_n(x)\downarrow_n 0$. Take $n$ big
enough in such a way that $\sup_{x\in K_\varepsilon}c_n(x) <\varepsilon$.
Define $0\leq g^\varepsilon_n:=\min(c_n,\varepsilon)$. Then clearly 
\begin{equation}  \label{firstineq}
g_n^\varepsilon(x)=\left|g_n^\varepsilon(x)\right|\leq \varepsilon\leq
\varepsilon(1+\left|\psi(x)\right|)\,\,\,\forall x\in\mathbb{X}%
\Rightarrow\left\| g_n^\varepsilon\right\|_\psi\leq \varepsilon\Rightarrow
\left|\langle g_n^\varepsilon,L\rangle\right|\leq \left\| L\right\|\left\|
g_n^\varepsilon\right\|_\psi\leq \left\| L\right\|\varepsilon
\end{equation}
where $\left\| L\right\|$ is the operator norm ($\left\| L\right\|<+\infty$
since $L$ is continuous). Also, since $\sup_{x\in
K_\varepsilon}c_n(x)<\varepsilon$, $c_n$ and $g_n^\varepsilon$ coincide on $%
K_\varepsilon$, namely $(c_n-g_n^\varepsilon)|_{K_\varepsilon}\equiv 0$. But
using the hypothesis on $L$ we then get 
\begin{equation}  \label{secondineq}
(c_n-g_n^\varepsilon)|_{K_\varepsilon}\equiv 0 \Rightarrow \left|\langle
c_n-g_n^\varepsilon,L\rangle\right|\leq \varepsilon \left\|
c_n-g_n^\varepsilon\right\|_\psi\leq \varepsilon(\left\|
c_n\right\|_\psi+\left\| g_n^\varepsilon\right\|_\psi)\leq \varepsilon(
\left\| c_1\right\|_\psi+\varepsilon)
\end{equation}
where in the last step we used the Banach lattice property of $\left\|
\cdot\right\|_\psi$ and the fact that as shown before $\left\|
g_n^\varepsilon\right\|_\psi\leq\varepsilon$. We now combine %
\eqref{firstineq} and \eqref{secondineq} to get 
\begin{equation*}
0\leq \langle c_n,L\rangle=\langle c_n-g_n^\varepsilon,L\rangle+\langle
g_n^\varepsilon,L\rangle\leq \left|\langle
c_n-g_n^\varepsilon,L\rangle\right|+\langle g_n^\varepsilon,L\rangle\leq
\varepsilon( \left\| c_1\right\|_\psi+\varepsilon)+\left\|
L\right\|\varepsilon\,.
\end{equation*}
Since $\varepsilon>0$ is arbitrary, $L(c_n)\downarrow 0$.
\end{proof}

\begin{theorem}[Daniell-Stone]
\label{thmdaniellstone} Let $\mathfrak{F}$ be a vector lattice of functions
(i.e. $f,g\in\mathfrak{F}\Rightarrow \max(f,g)\in\mathfrak{F}$) on a set $%
\mathbb{X}$ such that $1\in\mathfrak{F}$. Let $L$ be a linear functional on $%
\mathfrak{F}$ with the following properties: $L(f)\geq 0$ whenever $f\geq 0$%
, $L(1) = 1$, and $L(f_n)\rightarrow_n 0$ for every sequence of functions $%
f_n$ in $\mathfrak{F}$ monotonically decreasing to zero. Then, there exists
a unique probability measure $\mu$ on the $\sigma$-algebra $\mathcal{F} =
\sigma(\mathfrak{F})$ generated by $\mathfrak{F}$ such that $\mathfrak{F}
\subseteq L^1(\mu)$ and 
\begin{equation*}
L(f)=\int_{\mathbb{X}}f\,\mathrm{d}\mu\,,\,\,\,\forall\,f\in\,\mathfrak{F}\,.
\end{equation*}
\end{theorem}

\begin{proof}
See \cite{Bogachev07} Theorem 7.8.1.
\end{proof}

\subsection{Proofs \label{Proofs}}

\begin{proof}[Proof of Lemma \protect\ref{propformaconiugata}]
Fix $\lambda\in (C_{0:T})^*$ such that $\lambda\geq 0$. Then 
\begin{align}
\mathfrak{P}^*(\lambda)&=\sup_{c\in C_{0:T}}\left(\mathfrak{P}(c)-\langle
c,\lambda\rangle\right)=\sup_{c\in C_{0:T}}\left(\sup_{z \in -\mathcal{A}%
}\sup_{\varphi \in \mathbf{\Phi }_{z }(c)}S^{U}\left( \varphi
\right)-\langle c,\lambda\rangle\right)  \notag \\
&\overset{\eqref{deftildefrakPgeneral}}{=}\sup_{c\in C_{0:T}}\left(\sup_{z
\in -\mathcal{A}}\sup_{\varphi \in \widetilde{\mathbf{\Phi }}_{z
}(c)}S^{U}\left( \varphi \right)-\langle c,\lambda\rangle\right)  \notag \\
&=\sup_{c\in C_{0:T}}\sup_{z \in -\mathcal{A}}\sup_{\varphi \in \widetilde{%
\mathbf{\Phi }}_{z }(c)}\left(S^{U}\left( \varphi \right)-\langle
c,\lambda\rangle\right)  \label{rewritepstar} \\
&\leq \sup_{c\in C_{0:T}}\sup_{z \in -\mathcal{A}}\sup_{\varphi \in 
\widetilde{\mathbf{\Phi }}_{z }(c)}\left(S^{U}\left( \varphi \right)-\langle
\sum_{t=0}^T\varphi_t+z,\lambda\rangle\right)  \notag \\
&=\sup_{z \in -\mathcal{A}}\sup_{\varphi \in \mathcal{E}}\left(S^{U}\left(
\varphi \right)-\langle
\sum_{t=0}^T\varphi_t+z,\lambda\rangle\right)=\sup_{z\in \mathcal{A}}\langle
z,\lambda\rangle+\sup_{\varphi \in \mathcal{E}}\left(S^{U}\left( \varphi
\right)-\sum_{t=0}^T\langle \varphi_t,\lambda_t\rangle\right)  \notag \\
&=\sigma_\mathcal{A}(\lambda)+(S^U)^*(\lambda_0,\dots,\lambda_T)\,.  \notag
\end{align}
Consequently 
\begin{equation}  \label{ineqpsv1}
\mathfrak{P}^*(\lambda)\leq \sigma_\mathcal{A}(\lambda)+(S^U)^*(\lambda_0,%
\dots,\lambda_T)\,.
\end{equation}
At the same time, for every $\varphi\in\mathcal{E}, z\in -%
\mathcal{A}$ and for $\widehat{c}=\sum_{t=0}^T\varphi_t+z\in C_{0:T}$ we have that $%
\varphi\in \Phi_z(\widehat{c})$. Thus 
\begin{equation*}
S^{U}\left( \varphi \right)-\langle
\sum_{t=0}^T\varphi_t+z,\lambda\rangle\leq \sup_{\varphi \in \widetilde{%
\mathbf{\Phi }}_{z}(\widehat{c})}\left(S^{U}\left( \varphi \right)-\langle 
\widehat{c},\lambda\rangle\right)\overset{\eqref{rewritepstar}}{\leq}%
\mathfrak{P}^*(\lambda)
\end{equation*}
and hence

\begin{equation}  \label{ineqpsv2}
\sup_{z \in \mathcal{A}}\,\langle z,\lambda\rangle+\sup_{\varphi \in 
\mathcal{E}}\left(S^{U}\left( \varphi \right)-\langle
\sum_{t=0}^T\varphi_t,\lambda\rangle\right)=\sup_{z \in -\mathcal{A}%
}\sup_{\varphi \in \mathcal{E}}\left(S^{U}\left( \varphi \right)-\langle
\sum_{t=0}^T\varphi_t+z,\lambda\rangle\right)\leq \mathfrak{P}^*(\lambda)\,.
\end{equation}

Combining \eqref{ineqpsv1} and \eqref{ineqpsv2} we get that $\mathfrak{P}%
^*(\lambda)=\sigma_\mathcal{A}(\lambda)+(S^U)^*(\lambda_0,\dots,\lambda_T)$.

If additionally $
\lambda(1)=1$ then we have 
\begin{align*}
(S^U)^*(\lambda_0,\dots,\lambda_T)&=\sup_{\varphi \in \mathcal{E}%
}\left(S^{U}(\varphi)-\langle \sum_{t=0}^T\varphi_t,\lambda_t\rangle\right)
\\
&=\sup_{\varphi \in \mathcal{E}}\left(\sup_{\alpha\in{\mathbb{R}}^{T+1}}
\left( U(\varphi+\alpha)-\sum_{t=0}^T\alpha_t\right)-\langle
\sum_{t=0}^T\varphi_t,\lambda\rangle\right)\\
&=\sup_{\varphi \in \mathcal{E}}\sup_{\alpha\in{\mathbb{R}}%
^{T+1}}\left(\left( U(\varphi+\alpha)-\sum_{t=0}^T\alpha_t\right)-\langle
\sum_{t=0}^T\varphi_t,\lambda\rangle\right) \\
&=\sup_{\varphi \in \mathcal{E}}\sup_{\alpha\in{\mathbb{R}}%
^{T+1}}\left(U(\varphi+\alpha)-\sum_{t=0}^T\langle
\varphi_t+\alpha_t,\lambda\rangle\right) \\
&=\sup_{\varphi \in \mathcal{E}}\left(U(\varphi)-\langle
\sum_{t=0}^T\varphi_t,\lambda\rangle\right)=\mathcal{D}(\lambda_0,\dots,%
\lambda_T)=\mathcal{D}(\lambda)\,.
\end{align*}
\end{proof}

\begin{proof}[Proof of Proposition \protect\ref{propvstar}]
We will use \cite{Sava} Theorem 2.7 and \cite{Sava} Remark 2.8. To do so,
let us rename $F:=v_{t}^{\ast }$ ((see \eqref{cnvxconjvstar} for the
definition of $v^{\ast }$), which implies that $F^{\circ }(y):=-F^{\ast
}(-y) $ of \cite{Sava} Equation (2.45) satisfies $F^{\circ }(y):=-F^{\ast
}(-y)=-v_{t}^{\ast \ast }(-y)=-v_{t}(-y)=u_{t}(y)$, by Fenchel-Moreau
Theorem. All the assumptions of \cite{Sava} Section 2.3 on $F$ are
satisfied, since for every $y\geq 0$ $F(y)\geq u_{t}(0)-0y=0$ and $%
F(1)=\sup_{x\in {\mathbb{R}}}(u_{t}(x)-x)\leq 0$ (recall $u_{t}(x)\leq
x,\,\forall \,x\in {\mathbb{R}}$). Also, since $\mathrm{dom}(u_{t})={\mathbb{%
R}}$, $\lim_{y\rightarrow +\infty }\frac{F(y)}{y}=F_{\infty }^{\prime
}=+\infty $. We can then apply \cite{Sava} Theorem 2.7 and \cite{Sava}
Remark 2.8, obtaining \eqref{explicitdual}. We stress the fact that since $%
u_{t}$ is finite valued on the whole ${\mathbb{R}}$, it is continuous there
and for every $\varphi _{t}\in \mathcal{C}_{b}(K_{t})$, $F^{\circ }(\varphi
_{t})=u_{t}(\varphi _{t})\in \mathcal{C}_{b}(K_{t})$, hence the additional
constraint $F^{\circ }(\varphi _{t})\in \mathcal{C}_{b}(K_{t})$ (below \cite%
{Sava} (2.49)) would be redundant in our setup.
\end{proof}

\begin{proof}[Proof of Proposition \protect\ref{PPP}]
We will exploit again \cite{Sava} Theorem 2.7 and \cite{Sava} Remark 2.8
(with $u_{t}$ in place of $F^{\circ }$) , as we explain now. Since $u_{t}$
is nondecreasing, either its domain is in the form $[M,+\infty )$ or $%
(M,+\infty )$, with $M\leq 0$. Given a $\varphi _{t}\in \mathcal{C}%
_{b}(K_{t})$ and a $\mu \in \mathrm{Meas}(K_{t})$

\begin{itemize}
\item Either $\inf (\varphi _{t}({\mathbb{R}}))>M$, in which case $%
u_{t}(\varphi _{t})\in \mathcal{C}_{b}(K_{t})$ since $u_{t}$ is continuous
on the interior of its domain.

\item Or $\inf (\varphi _{t}({\mathbb{R}}))<M$, in which case $\{\varphi
_{t}<M\}$ is open nonempty and hence has positive $\widehat{Q}_{t}$ measure,
as $\widehat{Q}_{t}$ has full support. Thus $\int_{K_{t}}u_{t}(\varphi
_{t})\,\mathrm{d}\widehat{Q}_{t}=-\infty $.

\item Or $\inf (\varphi _{t}({\mathbb{R}}))=M$ in which case $u_{t}(\varphi
_{t})=\lim_{\varepsilon \downarrow 0}u_{t}(\max (\varphi _{t},M+\varepsilon
))$ (since $u_{t}$ is nondecreasing and upper semicontinuous) $u_{t}(\max
(\varphi _{t},M+\varepsilon ))\in \mathcal{C}_{b}(K_{t})$ (see first bullet)
and by Monotone Convergence Theorem 
\begin{equation*}
\int_{K_{t}}u_{t}(\varphi _{t})\,\mathrm{d}\widehat{Q}_{t}-\int_{K_{t}}%
\varphi _{t}\,\mathrm{d}\mu =\lim_{\varepsilon \downarrow 0}\left(
\int_{K_{t}}u_{t}(\max (\varphi _{t},M+\varepsilon ))\,\mathrm{d}\widehat{Q}%
_{t}-\int_{K_{t}}\max (\varphi _{t},M+\varepsilon )\,\mathrm{d}\mu \right)
\,.
\end{equation*}%
Then we infer that 
\begin{align}
&\sup_{\varphi _{t}\in \mathcal{C}_{b}(K_{t})}\left( \int_{K_{t}}\varphi
_{t}\,\mathrm{d}\mu -\int_{K_{t}}v_{t}(\varphi _{t})\,\mathrm{d}\widehat{Q}%
_{t}\right) =\sup_{\varphi_{t} \in \mathcal{C}_{b}(K_{t})}\left(
\int_{K_{t}}u_{t}(\varphi_{t} )\,\mathrm{d}\widehat{Q}_{t}-\int_{K_{t}}%
\varphi_{t} \,\mathrm{d}\mu \right)  \notag \\
&= \sup \left\{ \int_{K_{t}}u_{t}(\varphi _{t})\,\mathrm{d}\widehat{Q}%
_{t}-\int_{K_{t}}\varphi _{t}\,\mathrm{d}\mu \mid \varphi _{t},u_{t}(\varphi
_{t})\in \mathcal{C}_{b}(K_{t})\right\} ,  \label{eqfromconttogeneralphi}
\end{align}


and from \cite{Sava} Theorem 2.7, \cite{Sava} Remark 2.8 and from %
\eqref{eqfromconttogeneralphi} we conclude the thesis.
\end{itemize}
\end{proof}

\label{appnoncpt}

\bibliographystyle{abbrv}
\bibliography{BibAll}

\end{document}